%% file: main.tex
\documentclass[11pt]{article}

\usepackage[letterpaper, left=1in, right=1in, top=1in,bottom=1in]{geometry}
\usepackage[shortlabels]{enumitem}
\usepackage{bm}
\usepackage{subfig}
\usepackage{anyfontsize}

\input{macros}

\bibliography{references}

\title{Learning in Stackelberg Games with Non-myopic Agents}

\author{
Nika Haghtalab\thanks{UC Berkeley, \texttt{nika@berkeley.edu}} 
\and 
Thodoris Lykouris\thanks{Massachusetts Institute of Technology, \texttt{lykouris@mit.edu}}  
\and Sloan Nietert\thanks{Cornell University, \texttt{nietert@cs.cornell.edu}} 
\and Alexander Wei\thanks{UC Berkeley, \texttt{awei@berkeley.edu}} }

\date{First version: August 2022\\
Current version: May 2025\footnote{A preliminary version was accepted for presentation at the Conference on Economics and Computation (EC 2022).}}

\begin{document}

\maketitle

\begin{abstract}
\input{0-abstract}
\end{abstract}

\addtocounter{page}{-1}
\thispagestyle{empty}

\newpage

\input{1-introduction}

\input{2-framework}
\input{3-clinch}
\input{4-non-myopic-clinch}
\input{5-beyond-SSGs}
\input{6-conclusion}

\subsection*{Acknowledgements} The authors thank the EC'22 reviewers for their helpful feedback and Bobby Kleinberg for insightful discussions on robust linear programming, security games, and bandit convex optimization. This work is partially supported by the National Science Foundation under grant CCF-2145898 and Graduate
Research Fellowships DGE-1650441 and DGE-2146752. The authors are also grateful to the Simons Institute for the Theory of Computing as part of this work was done during the Fall'22 semester-long program on \emph{Data Driven Decision Processes}.

\printbibliography

\appendix
\input{A1-framework}
\input{A2-clinch}
\input{A3-clinch-experiments}
\input{A4-non-myopic-clinch-experiments}
\input{A5-demand-learning}
\input{A6-finite-stackelberg}
\input{A7-strategic-classification}

\end{document}

%% file: macros.tex
\usepackage[utf8]{inputenc}
\usepackage[T1]{fontenc}
\usepackage{microtype} %

\usepackage[style=trad-alpha,natbib=true,maxnames=4,maxalphanames=4,maxbibnames=99]{biblatex}

\usepackage{xparse}

\usepackage{booktabs}

\usepackage{sidecap}

\usepackage{amsmath, amsthm}
\usepackage{mathtools}
\usepackage{thmtools}
\usepackage{newtxmath}

\usepackage[noend]{algpseudocode}

\usepackage{dsfont}
\usepackage{multirow}

\usepackage{color-edits}

\usepackage[linesnumbered,noend,ruled]{algorithm2e}
\SetKwComment{Comment}{$\triangleright$\ }{}

\usepackage{hyperref}
\hypersetup{citecolor = green,}
\usepackage[capitalise]{cleveref}

\renewcommand{\epsilon}{\varepsilon}

\DeclareMathOperator*{\argmin}{arg\,min}
\DeclareMathOperator*{\argmax}{arg\,max}    

\makeatletter
\NewDocumentCommand{\BestResp}{O{} O{}}{\mathsf{BR}\@ifnotempty{#1}{^{#1}}\@ifnotempty{#2}{_{#2}}}
\makeatother

\newcommand{\eps}{\varepsilon}

\newcommand{\conv}{\mathrm{conv}}
\newcommand{\diam}{\mathrm{diam}}

\newcommand{\E}{\mathrm{E}}

\newcommand{\Unif}{\mathrm{Unif}}
\newcommand{\R}{\mathbb{R}}
\newcommand{\cA}{\mathcal{A}}
\newcommand{\cB}{\mathcal{B}}
\newcommand{\cX}{\mathcal{X}}
\newcommand{\cY}{\mathcal{Y}}

\newcommand{\cR}{\mathcal{R}}

\newcommand{\BR}{\BestResp}

\newcommand{\brr}{\mathsf{br}}
\newcommand{\br}{\mathsf{br}}
\newcommand{\poly}{\mathrm{poly}}

\newcommand{\MultiThreadedFiniteAlg}{\mbox{{\normalfont\textsc{MultiThreadedRobustStack}}}}

\newcommand{\Dist}{\mathcal{D}}

\newcommand{\UCB}{\textsc{UCB}}
\newcommand{\BatchedBinarySearch}{\mbox{{\normalfont\textsc{BatchedBinarySearch}}}}

\newcommand{\SE}{{\normalfont\textsc{SuccessiveElimination}}}
\newcommand{\SEDelayed}{{\normalfont\textsc{SuccElimDelayed}}}
\newcommand{\SEUpdate}{{\normalfont\textsc{UpdateBounds}}}

\newcommand{\Oracle}{\mbox{{\normalfont\textsc{Oracle}}}}
\newcommand{\ClinchSimplex}{\mbox{{\normalfont\textsc{Clinch.Simplex}}}}
\newcommand{\Clinch}{\mbox{{\normalfont\textsc{Clinch}}}}

\newcommand{\SecuritySearch}{\mbox{{\normalfont\textsc{SecuritySearch}}}}
\newcommand{\Perturb}{\mbox{{\normalfont\textsc{Perturb}}}}
\newcommand{\BatchedClinch}{\mbox{{\normalfont\textsc{BatchedClinch}}}}
\newcommand{\MultiThreadedClinch}{\mbox{{\normalfont\textsc{MultiThreadedClinch}}}}
\newcommand{\GDwoG}{\mbox{{\normalfont\textsc{GDwoG}}}}
\newcommand{\Regret}{R}

\newcommand{\ConserveMass}{\mbox{{\normalfont\textsc{ConserveMass}}}}

\newcommand{\Stabilize}{\mbox{{\normalfont\textsc{ConserveMass}}}}

\newcommand{\FiniteAlg}{\mbox{{\normalfont\textsc{RobustStack}}}}

\newcommand{\vol}{\mathrm{vol}}         

\newcommand{\bvec}[1]{\mathbf{#1}}
\newcommand{\bv}[1]{\mathbf{#1}}

\allowdisplaybreaks

\theoremstyle{plain}
\newtheorem{lemma}{Lemma}[section]
\newtheorem{theorem}[lemma]{Theorem}

\newtheorem{proposition}[lemma]{Proposition}
\newtheorem{remark}[lemma]{Remark}

\theoremstyle{definition}
\newtheorem{definition}[lemma]{Definition}
\newtheorem{assumption}[lemma]{Assumption}

%% file: 0-abstract.tex
We study Stackelberg games where a principal repeatedly interacts with a non-myopic long-lived agent, without knowing the agent's payoff function.
Although learning in Stackelberg games is well-understood when the agent is myopic, dealing with non-myopic agents poses additional complications. In particular, non-myopic agents may strategize and select actions that are inferior in the present in order to mislead the principal's learning algorithm and obtain better outcomes in the future.

We provide a general framework that reduces learning in presence of non-myopic agents to robust bandit optimization in the presence of myopic agents. Through the design and analysis of \emph{minimally reactive} bandit algorithms, our reduction trades off the statistical efficiency of the principal's learning algorithm against its effectiveness in inducing near-best-responses.
We apply this framework to \emph{Stackelberg security games (SSGs), pricing with unknown demand curve,} \emph{general finite Stackelberg games,} and \emph{strategic classification}. In each setting, we characterize the type and impact of misspecifications present in near-best responses and develop a learning algorithm robust to such misspecifications.

On the way, we improve the state-of-the-art query complexity of learning in SSGs with $n$ targets from $O(n^3)$ to a near-optimal $\widetilde{O}(n)$ by uncovering a fundamental structural property of these games. The latter result is of independent interest beyond learning with non-myopic~agents.

%% file: 1-introduction.tex
\newpage
\section{Introduction}

Stackelberg games are a canonical model for strategic principal-agent interactions. Consider a defense system that distributes its security resources across high-risk targets prior to attacks being executed; or a seller who chooses a price prior to knowing a customer's proclivity to buy; or a tax policymaker who sets rules on when audits are triggered prior to seeing filed tax reports. In each of these scenarios, a \emph{principal} first selects an action $x\in\mathcal{X}$ and then an \emph{agent} reacts with an action $y\in\mathcal{Y}$, where $\mathcal{X}$ and $\mathcal{Y}$ are the principal's and agent's action spaces, respectively. In the examples above, agent actions correspond to which target to attack, how much to purchase, and how much tax to pay to evade an audit, respectively.
Typically, the principal wants an action $x$ that maximizes their payoff when the agent plays a best response $y = \brr(x)$; such a pair $(x, y)$ is a \emph{Stackelberg equilibrium}. By \emph{committing} to a strategy, the principal can guarantee they achieve a higher payoff than in the fixed point equilibrium of the corresponding simultaneous-play game. However, finding such a strategy requires knowledge of the agent's payoff function.

When faced with unknown agent payoffs, the principal can attempt to learn a best response via repeated interactions with the agent. If a (na\"ive) agent is unaware that such learning occurs and always plays a best response, the principal can use classical online learning approaches to optimize their own payoff in the stage game. Learning from such \emph{myopic} agents has been extensively studied in multiple Stackelberg games, including security games \citep{letchford2009learning,blum2014,peng2019learning}, demand learning \citep{kleinberg03value,besbes2009dynamic}, and strategic classification \citep{dong2018, chen2020learning}.

However, long-lived agents will generally not volunteer information that can be used against them in the future.
This is especially true in online environments where a learner seeks to exploit recently learned patterns of behavior as soon as possible, thus the agent can see a tangible advantage for deviating from its instantaneous best response and leading the learner astray. This trade-off between the (statistical) efficiency of learning algorithms and the perverse long-term incentives they may create brings us to the main questions of this work:
\begin{center}
    \emph{What design principles lead to efficient learning in Stackelberg games with non-myopic agents?} \\
    \emph{How can insights from learning with
    myopic agents be applied to non-myopic agents?}
\end{center}

A typical assumption for non-myopic learning is that the principal is willing to wait longer for future rewards than the agent. This is modeled as an
\emph{asymmetry in patience} where the  agent, unlike the principal, receives  geometrically discounted utilities according to a discount factor $\gamma \in (0,1)$. For example, in the auctions literature, this modeling choice is rooted in the assumption that the auctioneer has greater means and is therefore more willing to accept deferred utilities compared to buyers who value immediate rewards (see Section \ref{ssec:related_work} for more details). 
Indeed, some asymmetry in patience is provably needed to enable a principal to learn effectively from strategic interactions in various domains 
\citep{amin2013learning,ananthakrishnan2024knowledge}.

This agent impatience favors principal policies which are \emph{slow} to implement lessons from each round of feedback. Such algorithms incentivize the agent to $\eps$-approximately best respond by making it unappealing for the agent to sacrifice payoff more than $\eps$ in the present for the effect their actions will have only far into the future. Thus, a key technical challenge for learning with non-myopic agents is the design of \emph{robust learning algorithms} that tolerate inexact best responses. In high-dimensional Stackelberg games, even non-robust learning requires care due to numerous discontinuities in the principal's payoff function and the difficulty of identifying well-behaved optimization subproblems. With inexact feedback, the principal must further understand the complex sets of strategies which can be rationalized by an $\eps$-approximately best-responding agent.

Furthermore, the statistical efficiency of the principal's learning algorithm must be traded off against its effectiveness at inducing approximate best responses. 
The more reactive an algorithm, i.e., the  faster it is  in implementing lessons learned from individual rounds of feedback, the more robust it has to be in order to handle deviations from an agent's best response.
Therefore, another technical challenge is to devise principled approaches for designing \emph{minimally reactive} (or \emph{optimally ``slowed-down''}) robust learning algorithms that support and encourage approximate best responses.

\subsection{Our contribution}

In Section 2, we present a general framework for non-myopic learning in Stackelberg games. We aim to minimize (\emph{Stackelberg}) regret, which compares our cumulative utility to that of an omniscient principal who always plays a Stackelberg equilibrium strategy. With this objective, we reduce non-myopic learning to \emph{robust bandit optimization with delayed feedback}, formalizing the two challenges outlined above. A na\"ive application of this reduction is the following: for a fixed delay $D$, cycle through $D$ copies of a robust bandit policy with regret bound $R(T)$ against $\eps$-approximate best-responding agents, where $\eps = \gamma^D/(1-\gamma)$. Facing this policy, a $\gamma$-discounting agent will always provide $\eps$-approximate best responses, so we achieve a non-myopic regret bound of $D \cdot R(T)$, where $D$ typically scales with the agent's effective time horizon $T_\gamma = \frac{1}{1-\gamma}$. Another common way of mediating principal-agent information flow is to require that the principal submit their actions in \emph{batches} of size $D$. We prove an equivalence between these two approaches, both of which serve as generic and user-friendly entry points for non-myopic algorithm design. For sharper regret bounds in specific applications, we employ time-varying batch sizes and insights from bandits with delays.

As our main application, we consider \emph{Stackelberg security games} (SSGs), a canonical setting that models strategic interactions between an attacker (agent) and a defender (principal). Here, the principal wishes to fractionally allocate defensive resources across $n$ targets, and the agent aims to attack while evading the principal's defense. Existing approaches solve $n$ separate convex optimization subproblems, one per target $y$ over the set of $x$ with $\brr(x) = y$, using agent feedback to learn the region each action $x$ belongs to \citep{conitzer2006computing, letchford2009learning, blum2014, balcan2015commitment}. 
However, $\eps$-approximate best responses $\brr_\eps(x)$ can corrupt this feedback adversarially anywhere near the boundaries of these high-dimensional regions.

\paragraph{Robust and optimal search algorithm for SSGs.} Towards a non-myopic learning algorithm for SSGs, Section~\ref{sec:clinch} begins with the simpler problem of \emph{robust search}. That is, we seek to estimate a Stackelberg equilibrium principal strategy using queries to an $\eps$-approximately best-responding agent.  Seeking an analytically tractable algorithm for this corrupted feedback setting, we uncover a clean structure that characterizes the principal's optimal solution against best-responding agents in a single-shot game. 
We show that all $n$ regions and sub-problems share a unique optimal solution $x^\star$ when considering a conservative allocation of the principal's resources. This leads us to a single optimization problem which we solve with a variant of the cutting plane method. 

The resulting algorithm---$\Clinch$---solves the myopic learning problem with near-optimal $\widetilde{O}(n)$ query complexity, improving upon on the state-of-the-art $O(n^3)$ dependence on the number of targets \citep{peng2019learning}. This asymptotic improvement is realized in practice, as we demonstrate empirically by implementing \Clinch{} and comparing it with the \SecuritySearch{} algorithm of \cite{peng2019learning}. Moreover, the simplicity of our new algorithm lets us extend it seamlessly to $\varepsilon$-approximately best-responding agents; the uniqueness of $x^\star$ allows us to approach it from any direction in the principal's strategy space while tolerating small perturbations. 

\paragraph{Extension to non-myopic agents with unknown discount factor.}
In Section~\ref{sec:non-myopic-clinch}, we turn \Clinch{} into an effective principal policy against $\gamma$-discounting agents, using the reduction from Section~\ref{sec:framework}. To improve upon the na\"ive cycling approach, we observe that \Clinch{} can advance with coarse feedback in initial rounds, only requiring accurate best responses as it nears $x^\star$. This motivates us to employ a gradually increasing batch-size schedule, and the resulting policy achieves regret  $\widetilde{O}(n(\log T + T_\gamma))$ against $\gamma$-discounting agents. Since knowledge of $\gamma$ may be impractical, we also develop a policy for unknown discount factor. Adapting an approach originally developed for adversarial corruptions \citep{LykourisMirrokniPaesLeme18}, we run $\log T$ parallel copies of \Clinch{} with geometrically increasing feedback delays, sharing information between copies by intersecting their confidence sets. This multi-threaded policy only incurs a $\log T$ multiplicative increase in regret compared to the batched algorithm. Via simulations against a restricted class of non-myopic agents, we demonstrate that the batched approach incurs linear regret if the guess for $\gamma$ is too small, while the multi-threaded approach always achieves sublinear regret at a mild overhead over the best~batched~policy.

\paragraph{Beyond SSGs.}
In Section~\ref{sec:beyond-ssgs}, we apply our framework to three more Stackelberg games: pricing with an unknown demand curve, general finite Stackelberg games, and strategic classification.
In each application area, we require a robust bandit optimization algorithm;
however, the context and type of noise to which we must be robust is application-dependent. Writing $u:\cX \times \cY \to \R$ for the principal's payoff function, we identify two general types of noise: 1) \emph{pointwise} errors, which refer to settings where $u(x, \brr_\eps(x))$ is close $u(x, \brr(x))$ for all $x$ and 2) bounded-region errors, where $\brr_\eps(x) = \brr(x)$ except in bounded and structured regions with no guarantees for $u(x,\brr_\eps(x))$. The relevant robust learning domains are convex optimization with a separation oracle (SSGs), stochastic multi-armed bandits (demand learning), convex optimization with a membership oracle (finite Stackelberg games), and bandit convex optimization (strategic classification).

\Cref{tab:my_label} presents our obtained regret guarantees alongside their corresponding error types and robust algorithms. For SSGs and demand learning, careful minimally reactive policies achieve regret scaling additively with the effective time horizon $T_\gamma$, while the others exhibit multiplicative scaling.  For demand learning, this is achieved using techniques from bandits with delays \citep{lancewicki2021stochastic}. Employing the same multi-threaded approach used for SSGs, we obtain $\gamma$-agnostic guarantees for all settings except strategic classification, at the cost of a $\log T$ multiplicative overhead.

\begin{table}
\small
\begin{tabular}{llp{4.4cm}p{4.4cm}} 
\toprule
Environment & Error Type & Robust Learning Algorithm & Non-myopic Regret\\
\midrule
SSGs & bounded-region & $\Clinch$ [this work] & $\widetilde{O}(n(\log T + T_\gamma))$\\
demand learning & pointwise  & \textsc{SuccElim} \citep{evendar06action} & $\widetilde O(\sqrt{T} + T_\gamma)$\\
finite Stackelberg & bounded-region & \FiniteAlg{} [this work] & $\widetilde{O}(T_\gamma \log^{4}\!T (V^{-1}\sqrt{m} + n m^{2.5}))$\\
strategic classification & pointwise & \GDwoG{} \citep{flaxman2005} &
$\widetilde{O}(T_\gamma^{1/4}\sqrt{d} \,T^{3/4})$ \\
\bottomrule
\end{tabular}
\caption{\raggedright For each primary learning environment, we list the error type to which we must be robust, the main robust learning algorithm employed, and the non-myopic regret bound achieved with known discount factor $\gamma$. Here, $n$ is the number of agent actions, %
$m$ is the number of principal actions  and $V^{-1}$ is an inverse volume quantity for finite Stackelberg games, and $d$ is the dimension of the feature space in strategic classification.}
\label{tab:my_label}
\end{table}

\subsection{Related work}\label{ssec:related_work}
\input{1b-related-work}

%% file: 1b-related-work.tex
Learning in the presence of non-myopic agents has been well-studied in the context of auctions \citep{amin2013learning, mohri2014, liu2018, abernethy2019}. There, batching and delays are often used to limit the extent of strategic manipulation from bidders and non-myopia is frequently modeled via $\gamma$-discounted utility maximizing agents. Initial research focused on posted prices \citep{amin2013learning, mohri2014, drutsa2017horizon} with later work including multi-bidder auctions with reserve prices \citep{liu2018, abernethy2019}, formal guarantees for incentive compatibility \citep{kanoria2014}, and more nuanced, contextual valuations \citep{golrezaei2019, drutsa20contextual,golrezaei2023incentive}. 
In \Cref{ssec:demand}, we provide direct comparisons to the posted-price setting, where we improve the state-of-the-art dependence on $T_\gamma$ for both fixed and stochastic buyer valuations.
Compared to high-dimensional settings like SSGs, analysis of approximate best response behavior is simpler in posted-price auctions, where 
agents can be viewed as having slightly perturbed one-dimensional values. 
In some multi-bidder settings, the seller's choice of personalized reserve prices is multi-dimensional \citep{golrezaei2019, golrezaei2023incentive}. Here, the distribution over reserve prices is unknown to the buyers until after bidding, so this setting does not fall into the Stackelberg framework, and the techniques employed do not translate directly to our applications.
Differential privacy has also been employed as a tool for filtering information flow in various mechanism design settings \citep{mcsherry2007,kobbi2012,kearns2014,liu2018,abernethy2019}.

Utility discounting has a long history in stochastic control, economics, and reinforcement learning (see, e.g., \citealp{blackwell1965, abreu1988,fedenberg1986, littman1994}), where it models uncertainty of future participation and acts as a tractable middle-ground between finite-horizon and infinite-horizon repeated games without discounting. In our setting, the discounted horizon~$T_\gamma$ models the agent's patience and would typically scale inversely with actual time between repeated interactions. Related notions of patience in repeated games include the extreme version of long-run (non-discounting) and short-run (best-responding) players \citep{fudenberg1990repeated}, and, for dynamic pricing, models where buyers persist for a limited number of rounds, during which they can make a single purchase \citep{ahn2007pricing,liu2015, lobel2020dynamic}.

In behavioral economics and psychology, there is a well-documented empirical phenomenon of individuals preferring immediate but poorer rewards over greater delayed rewards  \citep{ferrari1995procrastination, national1999pathological}. This impatience is often modeled via utility discounting \citep{steel2007nature, ross2012midbrain, kirby1999heroin, suranovic1999economic}, with \cite{kirby1999heroin} specifically comparing discount rates between drug users and non-users. More broadly, the field of picoeconomics explores the psychological roots of this behavior and its implications for decision-making \citep{ainslie1992picoeconomics}.
While this literature generally favors hyperbolic discounting \citep{laibson1997golden}, compared to this work's geometric discount model, we note that our multi-threaded algorithms rely quite weakly on this modeling choice. In particular, we observe favorable performance in empirical simulations of multi-threaded \Clinch{} against hyperbolic discounting agents (see Remark~\ref{rem:hyperbolic-discounting}).

Our results have similarities to several lines of work in the multi-armed bandit literature. 
First, multi-armed bandits with delayed feedback has been studied extensively in various settings, initially with stochastic arm-independent delays \citep{joulani2013online,vernade2017}, later with arm-dependent delays \citep{gael20stochastic,lancewicki2021stochastic}, and also with feedback aggregation \citep{pike2018bandits}. Our analysis in Section~\ref{ssec:demand} extends that of \cite{lancewicki2021stochastic} to handle adversarial perturbations. Second, multi-armed bandits with adversarial corruptions have been well-explored \citep{LykourisMirrokniPaesLeme18,GuptaKorenTalwar19,ZimmertSeldin21}. Unlike our setting where the agent acts according to a non-myopic behavioral model, this line of work assumes that the feedback can be completely adversarial in a bounded number of rounds $C$ and purely stochastic otherwise. That said, our handling of unknown discount factor uses the multi-threading paradigm from this literature \citep{LykourisMirrokniPaesLeme18} and draws an interesting connection between the role of feedback delays in non-myopic learning and the role of the corruption budget $C$.

Stackelberg Security games~\citep{conitzer2006computing,tambe2011security} have been well-studied in recent literature. In particular, regret and query complexity bounds have been given for online and offline learning \citep{blum2014,balcan2015commitment,XuTJ16,blum2017learning,peng2019learning}. We improve over the state-of-the-art query complexity results for security games (see \Cref{ssec:ssg-comparison} for details). 
Since SSGs often model interaction between long-lived institutions and adversaries aiming for short-term profit, our asymmetric discounting assumption is natural for this setting. Several works in security games have explored robustness to noisy best responses \citep{HaghtalabFNSPT16,pita2012robust}
using behavioral assumptions that model the noise. The resulting algorithms are adversarially robust as needed by our framework. \cite{nguyen2019tackling} treated a distinct Bayesian model of non-myopia using mixed integer linear programming.

Finally, we comment on our static equilibrium benchmark $T \max_{x \in \cX} u(x,\br(x))$. This is the maximum utility the principal can obtain against best-responding agents and is thus the standard benchmark for learning in Stackelberg games (see, e.g., the works in the previous paragraph). Interestingly, against non-myopic agents with known utilities, the principal can often beat the static benchmark \citep{zuo2015optimal,collina2023efficient}, by threatening to defect if the agent deviates from desired behavior. However, for unknown agent utilities, the principal may be unable to design a similarly effective threat. In particular, a $\Omega(T_\gamma)$ regret bound for dynamic pricing \citep{amin2013learning} implies that the principal cannot beat the static benchmark for pricing or generic security games and finite games (since pricing can be embedded as a special case of each). Moreover, any regret bounds proved via reduction to robust bandit learning must apply for best-responding agents, against which the static benchmark is optimal. Specifically, our regret analysis applies to any agents which approximately best respond when feedback is delayed, including best-responding (fully myopic) agents. In contrast, ``threat-based'' approaches strongly rely on agents being forward-looking (non-myopic) and believing the principal's threat.

%% file: 2-framework.tex
\section{Framework}\label{sec:framework}
We consider learning in general Stackelberg environments, in which a principal (the ``leader'') aims to learn an optimal strategy while interacting repeatedly with a non-myopic agent (the ``follower''). In this section, we first describe the basic model for principal-agent interaction and then introduce our general approach for learning in non-myopic principal-agent settings. 

\subsection{Model}
A \emph{Stackelberg game} is a tuple $(\mathcal{X}, \mathcal{Y}, u, v)$ of principal action set $\mathcal{X}$, agent action set $\mathcal{Y}$, principal payoff function $u\colon\mathcal{X}\times\mathcal{Y}\to [0, 1]$, and agent payoff function $v\colon\mathcal{X}\times\mathcal{Y}\to [0, 1]$. The principal leads with an action $x\in\mathcal{X}$, observed by the agent and the agent follows with an action $y\in\mathcal{Y}$, observed by the principal. Finally, the principal and the agent receive payoffs $u(x, y)$ and $v(x, y)$, respectively. 

We consider \emph{repeated} Stackelberg games, in which the same principal and agent play a
sequence of Stackelberg games $((\mathcal{X}, \mathcal{Y}, u, v_t))_{t=1}^T$ over $T$ rounds, with both participants observing the outcome of each game before proceeding to the next round.
Notice that the agent's payoff function $v$ may depend on the round $t$, possibly drawn from some distribution over possible payoff functions.
Furthermore, we assume that the agent knows both the principal payoff function $u$ and the distribution over each future agent payoff function $v_t$, while the principal knows only $u$. When considering the principal learning in this context, we also assume that the agent knows the principal's learning algorithm and can thus compute its forward-looking utility (as we discuss further below).

\paragraph{Discounting.}
A common assumption in repeated games is that agents discount future payoffs; our agent acts with a discount factor of $\gamma$, for some $0 < \gamma<1$. Formally, for a sequence $((x_1, y_1), (x_2, y_2), \ldots, (x_T, y_T))$ of actions, the principal's total utility is $\sum_{t=1}^T u(x_t, y_t)$ and the agent's \emph{$\gamma$-discounted} utility is $\sum_{t=1}^T \gamma^t v_t(x_t, y_t)$. We make the behavioral assumption that the agent acts to maximize their expected $\gamma$-discounted utility and may thus trade off present utility for future (discounted) payoffs. A canonical motivation for this assumption is that the agent leaves the game with probability $\gamma$ at each round and is replaced by another agent from the same population. 

\paragraph{(Approximate) best responses.}
To bound the loss in present utility compared to the (myopic) best response, we consider $\epsilon$-approximate best responses. Considering approximate best responses lets us move beyond myopic agents who always maximize present-round utility, as typically studied in Stackelberg games, to non-myopic agents whose actions take future payoffs into account.

Formally, define $\BestResp(x)\coloneqq \{ y\in\mathcal{Y} : v(x, y) = \max_{y'\in\mathcal{Y}} v(x, y') \}$ to be the agent's best response set and $\BestResp[\varepsilon](x) \coloneqq \{  y\in\mathcal{Y} : v(x, y) \ge \max_{y'\in\mathcal{Y}} v(x, y') - \varepsilon \}$ to be their
\emph{$\varepsilon$-approximate} best response set to $x\in\mathcal{X}$. When the agent payoff functions $v_t$ vary with the round number $t$, we write $\BestResp[][t](x)$ and $\BestResp[\eps][t](x)$ to denote the agent's ($\eps$-approximate) best response sets with respect to $v_t$.

\paragraph{Histories, policies, and regret.}
A \emph{history} $H$ is an element of $\mathcal{H}\coloneqq\bigcup_{t\ge 0} (\mathcal{X}\times\mathcal{Y})^t$ representing actions played in previous rounds. A principal policy $\mathcal{A}\colon\mathcal{H}\to\mathcal{X}$ is a (possibly random) function that takes a history $H_{t - 1} = ((x_{s}, y_{s}))_{s=1}^{t-1}$ and outputs an action $x_t = \mathcal{A}(H_{t - 1})$. An agent policy $\mathcal{B}$ is a (possibly random) function that, given a history $H_{t - 1}$, realized utility functions $v_{1:t} = (v_s)_{s=1}^t$, and a principal action $x_t$, outputs an action $y_t = \mathcal{B}(H_{t - 1}, v_{1:t},x_t)$.

The principal commits to a policy $\mathcal{A}$ before the {start of the} game. The agent then chooses a policy $\mathcal{B}$. To measure the performance of $\cA$ against $\cB$, we use \emph{Stackelberg} (or \emph{strategic}) \emph{regret}
\begin{equation}
\label{eq:stackelberg-regret}
    \Regret_{\cA,\cB}(T) \coloneqq \max_{x\in\mathcal{X}} \left({\E\left[\sum_{t=1}^T \Bigl( \max_{y\in\BestResp[][t](x)} u(x, y) - u(x_t, y_t) \Bigr)\right]}\right),
\end{equation}
where the expectation is taken over the random history $H_T$ induced by these policies and the (possibly random) agent utilities $v_{1:T}$. This regret compares the principal's realized payoff to that obtained against a best-responding agent. When the optimal choice of $y\in \BestResp_t(x)$ is not unique, we consider the choice of $y$ that corresponds to an agent tie-breaking in favor of the principal, as this yields the highest standard against which one can compete. When agent payoffs $v_t$ are stochastic and drawn i.i.d., the regret benchmarks $\cA$ against optimal Stackelberg equilibrium play in the stage game and decomposes into $T \max_{x \in \cX} \E[u(x,\brr(x))] - \E[\sum_{t=1}^T u(x_t,y_t)]$, where $\brr(x) \in \argmax_{y \in \BestResp_1(x)} u(x,y)$ again breaks ties in favor of the principal.

Generally, we consider $\mathcal{B}$ belonging to a class of agent policies $\mathfrak{B}$ (potentially depending on $\cA$) and minimize the \emph{worst-case Stackelberg regret} $R_{\cA,\mathfrak{B}}(T) \coloneqq \sup_{\cB \in \mathfrak{B}} R_{\cA,\cB}(T)$. In our non-myopic setting, $\mathfrak{B} = \mathfrak{B}_\gamma(\cA)$ is the family of policies which maximize the agent's $\gamma$-discounted utility given $\cA$.
Since our framework will relate non-myopic agents to approximately best-responding agents, we also consider the class $\mathfrak{B}^\eps$ of policies $\cB$ with $\cB(H_{t-1},x_t) \in \BestResp^\eps(x_t)$ for all $t$, where $\mathfrak{B}^0$ corresponds to the traditional myopic setting. Define $R_{\cA}(T,\gamma)\coloneqq R_{\cA,\mathfrak{B}_\gamma(\cA)}(T)$ and $R_{\cA}^\eps(T)\coloneqq R_{\cA,\mathfrak{B}^\eps}(T)$, respectively.

\subsection{Reduction to robust and minimally reactive learning}
\label{ssec:reduction}

As noted above, a major challenge {in} our learning setting is that agents may play actions that are far from best responses in any given round to obtain higher discounted future utility. At a high level, this is remedied by choosing a principal policy that is minimally reactive to agent feedback, so that each agent action has a bounded impact on said utility. Concretely,
a simple technique to decrease the influence that individual agent actions have on the principal policy, and thus the agent's incentive to manipulate their action in the present round, is to delay the principal's response to agent actions. Formally, %
a principal policy $\mathcal{A}$ is \emph{$D$-delayed} if each action $x_t = \mathcal{A}(H_{t-1})$ relies only on the prefix $H_{t-D}$ of $H_{t-1}$, i.e., \smash{$\mathcal{A}(H_{t-1}) = \mathcal{A}'(H_{t - D})$} for some \smash{$\mathcal{A}'$}. 
With sufficient delay, the agent will have little incentive to manipulate their action and will play  an approximate best response; this
can be thought of as a (possibly adversarial) perturbation of the actual best response. Previous work has explored such an idea in the context of auctions with non-myopic agents (see \Cref{ssec:related_work}); in contrast, we focus on distilling design principles that apply to general principal-agent settings. Towards this goal, we present a black-box reduction from learning with non-myopic agents to the better-understood problem of bandit learning from adversarially perturbed inputs.

\begin{proposition}\label{prop:delayed-feedback}Let $0 < \gamma < 1$ and $\eps \geq 0$. Fix $D = \lceil T_\gamma \log (T_\gamma/\eps) \rceil$, where $T_\gamma = \frac{1}{1-\gamma}$ is the agent's discounted time horizon.
Then, if principal policy $\cA$ is $D$-delayed, we have $R_\cA(T,\gamma) \leq R_\cA^\eps(T)$.
\end{proposition}

Our proof (\Cref{app:robust-learning-with-delays}) observes that the total discounted utility for rounds after time $t$ is at most $\frac{1}{1-\gamma}\gamma^t$. While \Cref{prop:delayed-feedback} simplifies the principal's learning problem, it presents two new---but more tractable---challenges: (i) designing an adversarially robust bandit algorithm, and (ii) implementing this algorithm with delayed feedback.
For (i), we translate the guarantee $y_t \in \BestResp_t^\eps(x_t)$ to a more standard error type in a context-specific way, showing that $y_t \in \BestResp_t(x_t')$ for some $x_t'$ near $x_t$ and bounding the deviation of $y_t$ from $\BestResp_t(x_t)$. Achieving (i) alone is insufficient, since the agent need not approximately best respond if the principal policy reacts quickly to the agent's actions.

\paragraph{Design principles for minimally reactive learning.} 
For (ii), we note that any bandit algorithm $\cA$ can be simply converted to a $D$-delayed  algorithm with up to a multiplicative in $D$ overhead in regret, by interleaving $D$ copies of $\cA$ and (somewhat wastefully) running them in parallel; this was first observed by \citet{weinberger02delayed}. However, this approach is far from optimal in most of our applications; often, we are able to collect less wasteful feedback using non-reactive but more diverse and variable query schedules that allow us to incur less regret while maintaining the same delay. To design these non-reactive schedules, we relate designing efficient delayed algorithms to designing \emph{batched} algorithms, in which the principal makes queries and receives feedback in batches of size $B$. Formally, we say a principal policy is \emph{$B$-batched} if each action $x_t = \cA(H_{t-1})$ relies only on the prefix $H_{B\lfloor (t-1)/B \rfloor}$. By definition, any $D$-delayed policy is also $D$-batched, but there is also a useful reduction in the opposite direction (Proposition~\ref{prop:batch-delay-reduction}). The proof (\Cref{app:batch-delay}) runs two copies of $\cA$ in parallel, alternating between batches, and even applies to a wider class of abstract bandit learning problems (though the relevant case for our setting is $\mathfrak{B} = \mathfrak{B}^\eps$).

\begin{proposition}
\label{prop:batch-delay-reduction}
Any $B$-batched principal policy $\cA$ can be converted into a $B$-delayed policy $\cA'$ such that $R_{\cA',\mathfrak{B}}(T) \leq 2 R_{\cA,\mathfrak{B}}(T)$ for any  class of agent policies $\mathfrak{B}$.
\end{proposition}

%% file: 3-clinch.tex
\section{\Clinch{}: A Near-Optimal Robust Search Algorithm for SSGs}\label{sec:clinch}

We now turn to Stackelberg security games (SSGs), 
a canonical setting for principal-agent learning. 
In SSGs, the principal must allocate their limited resources to defend a set of targets from the agent, who aims who aims to attack advantageous targets left unprotected by the principal \citep{keikintveld2009}.
The principal first commits to a strategy, i.e., a probabilistic assignment of its resources to protect targets, and the agent then chooses a target to attack based on the principal's strategy.

Before treating non-myopic agents, we first design a robust search algorithm for SSGs.
For $n$-target games, our algorithm \Clinch{} approximates an optimal strategy for the principal using $\widetilde{O}(n)$ queries to a near best-responding agent (see \Cref{thm:clinch}). This query complexity is nearly optimal and improves upon the state-of-the-art of $O(n^3)$ queries for search with exact best responses \citep{peng2019learning}. To achieve this, we identify and leverage new structural properties of SSG equilibria to cast the principal's learning problem as quasi-convex optimization  with a separation oracle.

\subsection{Model and preliminaries}\label{sec:security-games-model}
A \emph{Stackelberg security game} (SSG) is a Stackelberg game $(\cX, \cY, u, v)$ where the agent attacks a target from the set $\cY = \{1,2,\ldots,n\}$ and the principal commits to a defense in the strategy space $\cX\subseteq [0, 1]^n$. A defense $\bvec{x}\in\cX$ corresponds to target $y\in\cY$ being defended with probability $x_y$. We assume that $\cX$ is closed, convex, and downward closed (i.e., if $\bv{x}\in\mathcal{X}$ and $\bv{x'} \in [0,1]^n$ is such that $x'_y \leq x_y$ for all $y$, then $\bv{x'} \in \cX$ as well). Finally, we assume payoffs depend only on the target attacked and the extent to which it was defended. Specifically, for each $y\in\cY$, $u(\bv{x}, y) = u^y(x_y)$ and $v(\bv{x}, y) = v^y(x_y)$, where $u^y\colon [0, 1]\to [0, 1]$ and $v^y\colon [0, 1]\to [0, 1]$ are, respectively, strictly increasing and strictly decreasing continuous functions in $x_y$. We note that this model generalizes a classical framework where $\cX$ is the space of marginal coverage probabilities achievable by a randomized allocation of defensive resources to certain schedules, under the ``subsets of schedules are schedules'' (SSAS) assumption \citep{korzhyk2011stackelberg}. This is the standard setting in the literature for learning in SSGs \citep{letchford2009learning,blum2014,peng2019learning}. Throughout, we write $\Unif(S)$ for the uniform distribution over a set $S \subseteq \cX$, $\mathbf{e}_y \in [0,1]^n$ for the standard basis vector of $y \in \cY$, and $\bv{0}_n,\bv{1}_n$ for the all zeros and ones vectors. We write $\bv{x} \leq \bv{x'}$ for $\bv{x},\bv{x'} \in \R^n$ if $x_y \leq x'_y$ for all $y \in \cY$.

\begin{remark}
\label{rem:clinch-simplex}
One important case of interest is  where the principal can defend only one target at a time (but is allowed to mix over which target to defend). Mathematically, this corresponds to the setting where $\mathcal{X}$ is the downward closure of the probability simplex $\Delta_{n-1}\coloneqq\{\bv{x} : \|\bv{x}\|_1 = 1 \land x_y\ge 0\:\forall y\}$, i.e., $\cX = \, \Delta_{n-1}^\leq\coloneqq\{\bv{x} : \|\bv{x}\|_1 \leq 1 \land x_y\ge 0\:\forall y\}$. We use the specialization of our framework to the simplex in \Cref{app:security-simplex} to facilitate the exposition of the main ideas of our approach.
\end{remark}

We consider learning a fixed SSG over $T$ rounds. During the $t$-th round, the principal announces a defense $\bv{x}^{(t)} \in \cX$, the agent attacks a target $y_t \in\cY$, and the players receive payoffs $u(\bv{x}^{(t)},y_t)$ and $v(\bv{x}^{(t)},y_t)$, respectively, with the agent payoff function $v$ unknown to the principal.
Recall that Stackelberg regret is given by $T \max_{\bv{x} \in \cX} u(\bv{x},\brr(\bv{x})) - \E\left[\sum_{t=1}^T u(\bv{x}^{(t)},y_t)\right]$.
When the agent is myopic and $y_t \in \BestResp(\bv{x}^{(t)})$ for all $t$, this task can be reframed as learning an optimal strategy for the principal using queries to a best response oracle that returns an arbitrary representative from $\BestResp(\bv{x})$ when given $\bv{x} \in \cX$. For non-myopic agents, we consider learning with an approximate best response oracle returning an element of $\BestResp^\eps(\bv{x})$, for some small $\eps > 0$.

\paragraph{Regularity assumptions.}

Additional structural assumptions are standard, and in fact necessary, for learning in security games. The conditions we use correspond to bit precision and non-degeneracy assumptions in previous work~\citep{letchford2009learning, blum2014, peng2019learning}.
First, we require a known \emph{slope bound} $C \geq 1$ such that, for all $0 \leq s < t \leq 1$ and $y \in \cY$:
\begin{equation*}
    \frac{1}{C} \leq \frac{v^y(s) - v^y(t)}{t-s} \leq C \quad \text{and} \quad 0 < \frac{u^y(t) - u^y(s)}{t-s} \leq C.
\end{equation*}
At a high level, this slope assumption bounds how quickly the defender and attacker utilities improve and degrade, respectively, with one extra unit of protection on an attacked target. When utility functions are linear, the upper bounds must be satisfied with $C=1$ to ensure payoffs in~$[0,1]$. When each $v^y$ is non-linear but continuously differentiable with derivative bounded away from 0, compactness of $[0,1]$ implies that these inequalities hold for sufficiently large $C$.

Second, we define for each target $y \in \cY$ the \emph{best response region} $K_y\subseteq\mathcal{X}$ as the set of principal strategies for which $y$ is a best response, i.e.,
$K_y \coloneqq \{ \bv{x} \in \cX : y \in \BestResp(\bv{x}) \}$. We require that non-empty best response regions $K_y$ have known minimum \emph{width} $W > 0$ along its target's dimension, i.e., 
\begin{equation*}
    \max_{x \in K_y} x_y - \min_{\bv{x} \in K_y} x_y = \max_{\bv{x} \in K_y} x_y \geq W \qquad \text{for all $y \in \cY$ with $K_y \neq \emptyset$},
\end{equation*}
where the equality uses that $\min_{\bv{x} \in K_y} x_y = 0$ by downward closure of $\cX$ and monotonicity of $v^y$. This assumption implies that any target that a best-responding attacker can be made to attack under some defense $x$ is a best response to a sufficiently substantial set of defenses.

\begin{remark}
The state-of-the art algorithm for learning SSGs  \citep{peng2019learning} imposes that utilities are linear with non-zero coefficients specified by $L$ bits, implying a slope bound of $C = 2^{L}$. Moreover, it requires a minimum volume $2^{-nL}$ for each non-empty region, implying a width bound of $W = 2^{-nL}$ as a region contained in $[0,1]^n$ has volume bounded by its width along any dimension.
\end{remark}

\subsection{Structural properties of SSG equilibria}\label{sec:ssg-structure}

To characterize equilibrium structure in SSGs underlying the analysis of \Clinch{}, we introduce the notion of \emph{conservative} strategies, where the principal wastes no defensive resources on targets not attacked by a best-responding agent. This property was originally defined for optimal strategies \citep{blum2014} but our generalization enables a helpful decoupling in our analysis.

\begin{definition}
A strategy $\bv{x}\in\cX$ is called \emph{conservative} if $x_y > 0$ only for $y \in \cY$ such that $y\in\BestResp(\bv{x})$.
\end{definition}

\begin{figure}[t]
\centering
\hspace{0mm}
\begin{minipage}{0.7\textwidth}
\centering
\includegraphics[width=1\textwidth]{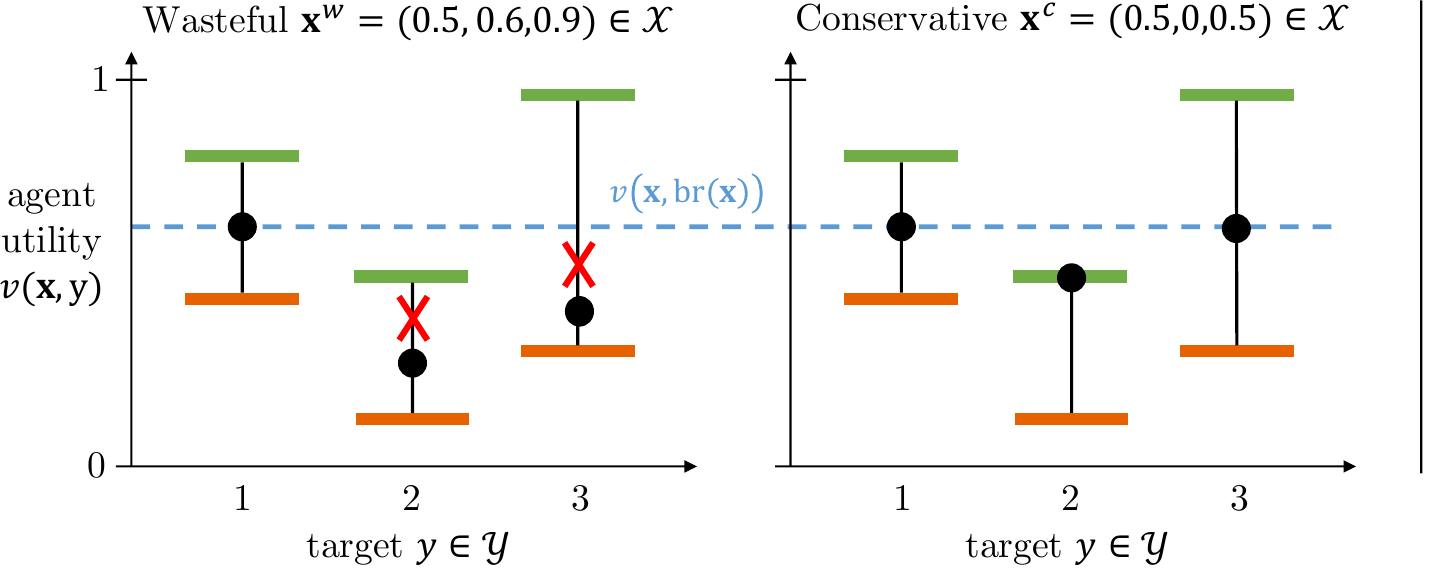}
\end{minipage}
\hspace{1mm}
\begin{minipage}{0.2\textwidth}
\vspace{-8mm}
{\fontsize{8}{12}\selectfont

$\cY = \{1,2,3\}, \, \cX = [0,1]^3$\vspace{1mm}

$v(\bv{x},1) = 0.8 - 0.35x_1$\vspace{-0.5mm}

$v(\bv{x},2) = 0.5 -0.5x_2$\vspace{-0.5mm}

$v(\bv{x},3) = 0.95 - 0.65x_3$\vspace{5mm}

\includegraphics[width=0.22\textwidth]{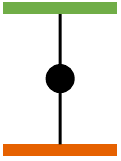}

\vspace{-10.6mm}
\hspace{7mm}
\begin{minipage}{0.75\textwidth}
    $v^y(0)$\\
    $v(\bv{x},y) = v^y(x_y)$\\
    $v^y(1)$
\end{minipage}
}
\end{minipage}\vspace{2mm}
\caption{
Agent utility profiles for a wasteful principal strategy $\bv{x}^w$ and a conservative strategy $\bv{x}^c$ for a 3-target SSG. The strategy $\bv{x}^w$ is wasteful because it allocates non-zero weight to targets 2 and 3, but $\BR(\bv{x}^w) = \{1\}$. The strategy $\bv{x}^c$ is conservative because $\BR(\bv{x}^c) = \{1,3\}$ and it allocates no weight to target 2.}
\label{fig:stability}
\vspace{-4mm}
\end{figure}

\Cref{fig:stability} compares agent utility profiles for a wasteful and a conservative principal strategy. The red crosses indicate wasteful coverage that must be eliminated to conserve resources while maintaining the same best response payoff $w = v(\bv{x},\brr(\bv{x}))$ for the agent. We now show that conservative strategies are uniquely determined by this payoff, which relates monotonically to their coordinates.

\begin{lemma}\label{lem:ssg-optimality}
Suppose $\bv{x}, \bv{x'}\in\cX$ are two conservative principal strategies, and define a best-responding agent's utilities as $V \coloneqq v(\bv{x}, \brr(\bv{x}))$ and $V' \coloneqq v(\bv{x'}, \brr(\bv{x'}))$. If $V = V'$, then $\bv{x} = \bv{x'}$. Otherwise, if $V < V'$, we have $\BestResp(\bv{x'}) \subseteq \BestResp(\bv{x})$ and $x_y \geq x'_y$ for all $y \in \cY$ with equality only if $x_y = x'_y = 0$.
\end{lemma}
\begin{proof}
To show that $V = V'$ implies $\bv{x} = \bv{x'}$, we prove that the equation $v(\bv{\tilde{x}},\br(\bv{\tilde{x}})) = V$ has at most one conservative solution $\bv{\tilde{x}} \in \cX$ (which must thus coincide with $\bv{x}$). For each target $y \in \cY$, we prove that $\tilde{x}_y$ is uniquely determined by $V$, with our analysis split into three cases based on $v^y(0)$. This is the utility the agent obtains by attacking $y$ when it is entirely undefended ($\tilde{x}_y = 0$), corresponding to the green upper bars in \Cref{fig:stability}. 
\begin{enumerate}[(i)]
    \item If $v^y(0) > V$  then $\tilde{x}_y > 0$; otherwise, we would have $v(\bv{\tilde{x}},\br(\bv{\tilde{x}})) \geq v(\bv{\tilde{x}},y) = v^y(0) > V = v(\bv{\tilde{x}},\br(\bv{\tilde{x}}))$, which is a contradiction. Since $\bv{\tilde{x}}$ is conservative, it further holds that $y\in\BestResp(\bv{\tilde{x}})$ and thus $v^y(\tilde{x}_y) = V$. Since $v^y(\tilde{x}_y)$ is strictly decreasing, $V$ uniquely determines $\tilde{x}_y$.
    \item If $v^y(0) < V$  then $v^y(\tilde{x}_y) \leq v^y(0) < V$. Hence $y \not\in \BestResp(\bv{\tilde{x}})$, and so $\tilde{x}_y = 0$ as $\bv{\tilde{x}}$ is conservative. 
    \item If $v^y(0) = V$, then either $\tilde{x}_y = 0$ or $\tilde{x}_y > 0$. The latter implies $y\not\in\BestResp(\bv{\tilde{x}})$, which contradicts conservativeness. Hence we must have $\tilde{x}_y = 0$ here as well.
\end{enumerate}

Next, we suppose that $V < V'$ and show that $\BR(\bv{x'}) \subseteq \BR(\bv{x})$.
For any $y \in \BestResp(\bv{x}')$, we have $v^y(x'_y) = V' > V \geq v^y(x_y)$, so monotonicity implies $x'_y < x_y$. Hence, $x_y > 0$ and so $y \in \BestResp(\bv{x})$ by con-servativeness. 
Lastly, we show that $x_y \geq x'_y$ for all $y$, with equality only if $x_y = x'_y = 0$. Indeed, if
$x'_y > 0$, $y \in \BestResp(\bv{x}')$ and so $x_y > x'_y$, as above. For $y$ with $x'_y = 0$, we have $x_y \geq 0 = x'_y$.
\end{proof}

Using this, we establish that the unique conservative $\bv{x^\star}\in \cX$ maximizing the principal's payoff also minimizes the agent's payoff. We exploit this for efficient optimization: while the principal's payoff $u(\bv{x}, \brr(\bv{x}))$ is difficult to directly optimize, as it may be discontinuous in $\bv{x}$, the equivalent objective $v(\bv{x}, \brr(\bv{x})) = \max_{y\in\cY} v^y(x_y)$ is quasi-convex in $x$, and even convex when the $v^y$ are linear.

\begin{proposition}\label{prop:ssg-solution}
There exists a unique \emph{conservative} strategy $\bv{x^\star}\in\argmax_{x\in\cX} u(\bv{x}, \brr(\bv{x}))$.
Moreover, this principal strategy $\bv{x^\star}$ is also the unique \emph{conservative} strategy in $\argmin_{\bv{x}\in\cX} v(\bv{x}, \brr(\bv{x}))$.
\end{proposition}
\begin{proof}
We first show that a maximizer exists. For fixed $y \in \cY$, the best response set $K_y = \{ \bv{x} \in \cX : v(\bv{x},y) \geq v(\bv{x},y') \: \forall \: y' \in \cY \}$ is compact, since each function $v^y - v^{y'}$ is continuous and $\cX$ is compact. Since $u^y$ is also continuous, there must exist a solution to $\max_{\bv{x} \in K_y} u^y(x_y)$, and since $\cY$ is finite, one of these solutions must achieve $\max_{\bv{x} \in \cX} u(\bv{x},\br(\bv{x})) = \max_{y \in \cY} \max_{\bv{x} \in K_y} u^y(x_y)$.

Let $\bv{\tilde{x}}$ be an optimal strategy. If $\bv{\tilde{x}}$ is not conservative, we transform it into a conservative $\bv{x^{\star}} \in\argmax_{x\in\cX} u(\bv{x}, \brr(\bv{x}))$. Given that $\bv{\tilde{x}}$ is not conservative, there exists some $y\in\cY$ such that $\bv{\tilde{x}}_y > 0$ but $y\not\in\BestResp(\bv{\tilde{x}})$. By the downward closure property of $\cX$, we can reduce $\tilde{x}_y$ until either $\tilde{x}_y = 0$ or $y$ becomes a best response. If $y$ does not become a best response, then $u(\bv{\tilde{x}}, \brr(\bv{\tilde{x}}))$ stays the same. Otherwise, since the other coordinates are unchanged and $\brr$ breaks ties in favor of the principal, adding $y$ to the best response set can only decrease $u(\bv{\tilde{x}}, \brr(\bv{\tilde{x}}))$. While $\bv{\tilde{x}}$ is not conservative, we iterate this procedure. Since either $x_y$ gets set to $0$ or $y$ gets added to $\BR(\bv{\tilde{x}})$, this procedure will terminate after $n$ iterations. Thus, there exists a conservative strategy $\bv{x^\star}\in\argmax_{\bv{x}\in\cX} u(\bv{x}, \brr(\bv{x}))$.

Next, we show that $\bv{x^\star}$ is unique. Suppose for the sake of contradiction there exists another conservative strategy $\bv{x'}\neq \bv{x^\star}$ in $\argmax_{\bv{x}\in\cX} u(\bv{x}, \brr(\bv{x}))$.
Define $V^\star\coloneqq v(\bv{x^\star}, \brr(\bv{x^\star}))$ and $V'\coloneqq v(\bv{x'}, \brr(\bv{x'}))$.
By \Cref{lem:ssg-optimality}, we must have $V^\star\neq V'$.
Assume without loss of generality that $V^\star < V'$. We claim $\bv{x'}$ cannot be optimal for the principal. Indeed, fix any $y' \in \BR(\bv{x'})$ and $y^\star \in \BR(\bv{x^\star})$.
Note that if $\bv{x}^\star_{y^\star} = 0$, then the best response region $K_{y^\star}$ has  width $\max_{\bv{x} \in K_{y^\star}} x_{y^\star} = x^\star_{y^\star} = 0$, which is forbidden by our regularity assumption.
Therefore, $x^\star_{y^\star} >0$ and, by \Cref{lem:ssg-optimality}, we have that $x'_{y^\star} < x^\star_{y^\star}$ and $y'\in \BestResp(\bv{x}^\star)$. In this case, monotonicity of the utilities and optimality of $y^\star$ within $\BR(\bv{x}^\star)$ imply that $u(\bv{x}^\star, y^\star) \geq u(\bv{x}^\star, y') > u(\bv{x}', y')$. This shows that $\bv{x}'$ is not an optimal strategy for the principal.

Having shown that $\bv{x}^\star$ is well-defined, we now prove that it also minimizes $v(\bv{x}, \brr(\bv{x}))$. We first show that there exists a conservative strategy $\bv{x}'$ minimizing $v(\bv{x}, \brr(\bv{x})) = \max_y v(\bv{x}, y)$ over $\cX$. As before, we start with any minimizer $\bv{\tilde{x}}$ of $v(\bv{x}, \brr(\bv{x}))$. Such a minimizer must exist because $v(\bv{x}, \brr(\bv{x})) = \max_y v^y(x_y) $ is continuous in $\bv{x}$ as the maximum of finitely many continuous functions, and $\cX$ is compact. Then, if $\bv{\tilde{x}}$ is not conservative at some $y$, we can decrease $\tilde{x}_y$ until either $\tilde{x}_y = 0$ or $y$ is a best response, leaving $v(\bv{\tilde{x}},\br(\bv{\tilde{x}}))$ unchanged. As before, this procedure will terminate after $n$ iterations to arrive a conservative minimizer $\bv{x}'$. By minimality, we know that $v(\bv{x}',\brr(\bv{x}')) \leq V^\star$. If this inequality were strict, then the argument in the previous paragraph would contradict $\bv{x}^\star \in \argmax_{\bv{x} \in \cX} u(\bv{x},\brr(\bv{x}))$. Hence we have equality and \Cref{lem:ssg-optimality} implies that $\bv{x}' = \bv{x}^\star$.
\end{proof}

\begin{remark}
Theorem 3.8 of \citet{korzhyk2011stackelberg} identifies that all optimal strategies for the principal minimize the utility of a best-responding agent. Moreover, assuming a homogeneity condition on $\cX$ satisfied in the simplex setting (but not for general SSGs), Theorem 3.10 of \citet{korzhyk2011stackelberg} implies that \Cref{prop:ssg-solution} holds without requiring conservativeness. To the best of our knowledge, these properties have not been previously exploited for learning SSGs.
\end{remark}

\subsection{Design and analysis of \Clinch{}}
\label{sec:ssg-results}

Our algorithm \Clinch{} (\Cref{alg:clinch}) estimates the unique conservative optimizer $\bv{x}^\star$ guaranteed by \Cref{prop:ssg-solution}, even with inexact best response feedback. Its main loop searches for an approximate minimizer of the agent's best response utility $v(\bv{x},\brr(\bv{x}))$, while the post-processing routine \ConserveMass{} (\Cref{alg:conserve-mass}) ensures that this strategy is nearly conservative. 
More precisely, \Clinch{} maintains an active search region $S$, determined by entry-wise lower and upper bounds $\underline{\bv{x}},\overline{\bv{x}}$ which may be initialized using prior knowledge of $\bv{x}^\star$.  Upon querying the centroid\footnote{We give our proof assuming that we can exactly compute the centroid $\E_{w \sim\operatorname{Unif}(S)}[w]$ of each search region $S$. Handling an approximate centroid is standard (see, e.g., \citet{bertsimas2004solving}), and we omit the details. (Moreover, sample complexity is not affected since we have full knowledge of the set $S$ at each iteration.)} 
$\bv{x}$ of $S$ and receiving feedback $y \in \cY$, we deduce that $\bv{x}^\star \geq x_y - C\eps$. By updating $\underline{x}_y$ accordingly, we ``clinch'' this progress and either shrink $S$ significantly or remove $y$ from the active set $\cR$, in which case $S$ is flattened along this dimension in the next round.
The termination condition at Step~\ref{step:clinch-terminate} ensures that the agent's utility in best response to the final query is sufficiently small.
After this, \Stabilize{} performs a binary search for each target $y$ that approximates the procedure described in the proof of \Cref{prop:ssg-solution}, reducing $x_y$ until it nears the threshold where $y$ becomes a best response.

\begin{center}
\begin{minipage}{0.495\textwidth}
\begin{algorithm}[H]
\caption{\Clinch}\label{alg:clinch}
\DontPrintSemicolon
\SetAlgoNoLine
\SetKwInOut{Input}{input}
\Input{accuracy $\delta \in (0,1]$,
entry-wise lower and upper bounds $\underline{\bv{x}},\overline{\bv{x}} \in \R^n$ for $\bv{x}^\star$, $\eps$-approximate best response oracle \Oracle{} with $\Oracle(\bv{x})\!\in\! \BestResp^\eps(\bv{x})$\;}
$S \gets \cX$, $\cR \gets \cY$, $y \gets 1$, $\lambda \gets \smash{\frac{\delta}{4C^2}}$\;
    \While{$y \in \cR$}{\label{step:clinch-terminate}
        \For{$y' \in \cR$ \emph{\textbf{with}} $\underline{\bv{x}} + \lambda \mathbf{e}_{y'} \not\in \cR$}{
            $\cR \gets \cR \setminus \{y'\}$, \: $\smash{\overline{x}_{y'} \gets \underline{x}_{y'}}$ \label{step:clinch-coordinate-locking}
        }
        $S \gets \{ \bv{x}' \in \cX : \underline{\bv{x}} \leq \bv{x}' \leq \overline{\bv{x}}
        \}$\; \label{step:clinch-search-region}
        $\bv{x} \gets \E_{\bv{x}' \sim \Unif(S)} [\bv{x}']$\; \label{step:clinch-oracle-query}
        $y \gets \Oracle(\bv{x})$\; \label{step:clinch-oracle-response}
        $\smash{\underline{x}_y \gets x_y - C\eps}$\; \label{step:clinch-lower-bound-update}
    }
    \Return $\ConserveMass(\bv{x},\lambda,\underline{\bv{x}})$
\end{algorithm}
\end{minipage}
\hfill
\begin{minipage}{0.475\textwidth}
\begin{algorithm}[H]
\caption{\ConserveMass}\label{alg:conserve-mass}
\DontPrintSemicolon
\SetAlgoNoLine
\SetKwInOut{Input}{input}
\Input{strategy $\bv{x} \in \cX$, accuracy $\lambda \in (0,1]$, entry-wise lower bound $\underline{\bv{x}} \in \R^n$ for $\bv{x}^\star$, approximate best response oracle $\Oracle$\;}
\For{$y \in \cY$}{ \label{step:conserve-mass-loop-start}
    $\ell \gets \underline{x}_y$, $u \gets x_y$\;
    \While{$u - \ell \geq \lambda$}{
        $m \gets (u + \ell)/2$\;
        \textbf{if} $\Oracle(\bv{x} + [m-x_y]\mathbf{e}_y) = y$\;
        \hspace*{5mm} $\ell \gets m$\;
        \textbf{else} \:$u \gets m$\;
    }
   $\hat{x}_y \gets \ell$\;
} \label{step:conserve-mass-loop-end}
\Return $\hat{\bv{x}}$
\end{algorithm}
\end{minipage}
\end{center}

\begin{lemma}[Minimize]
\label{lem:minimization}
Fix $\delta \in (0,1]$ and $\underline{\bv{x}},\overline{\bv{x}} \in \R^n$ with $\underline{\bv{x}} \leq \bv{x}^\star \leq \overline{\bv{x}}$. Then, after $\smash{O\bigl(n\log\frac{C^2\alpha n}{\delta}\bigr)}$ queries to a $\smash{\frac{\delta}{33C^3n}}$-approximate best response oracle, \Clinch{} passes a strategy $\bv{x} \in \cX$ to \ConserveMass{} with
$v(\bv{x},\brr(\bv{x})) \leq v(\bv{x}^\star,\brr(\bv{x}^\star)) + \frac{\delta}{2C}$, where $\alpha = \max_{y \in \cY}\overline{x}_y - \underline{x}_y$.
\end{lemma}

\begin{lemma}[Conserve]
\label{lem:stabilization}
Fix $\lambda\in (0,1]$, $\bv{x} \in \cX$, and $\underline{\bv{x}} \in \R^n$ with $\underline{\bv{x}} \leq x$. Then $\Stabilize$ returns $\hat{\bv{x}} \in \cX$ with $\hat{x}_y > \underline{x}_y$ only for $y \in \BestResp^{3C\lambda}(\hat{\bv{x}})$ and such that $v(\hat{\bv{x}},\brr(\hat{\bv{x}})) \leq v(\bv{x},\brr(\bv{x})) + 2C\lambda$, using $O(n \log\frac{\alpha}{\lambda})$ queries to a $\frac{\lambda}{C}$-approximate best response oracle, where $\alpha = \max_{y \in \cY} x_y - \underline{x}_y$.
\end{lemma}

To prove \Cref{lem:minimization} (\Cref{prf:minimization}), we show that the volume of $S$ decreases by a constant factor in each round unless a target is removed from $\cR$, in which case $S$ loses a dimension but still has lower-dimensional volume not too much larger than before. The volume decrease claim follows by an approximate version of Gr\"unbaum's inequality \citep{grunbaum1960partitions}; we show that a half-space nearly passing through the centroid of $S$ splits the region into roughly balanced halves. For \Cref{lem:stabilization}, the $n$ binary searches ensure that each coordinate $\hat{x}_y$ is within $O(\lambda)$ of the threshold $\sup \{ p \in [\underline{x}_y,x_y] : y \in \BestResp((x_1,\dots,x_{y-1},p,x_{y+1},\dots,x_n)) \}$. In \Cref{app:stabilize}, we use this to show that each target with substantial coverage under $\hat{\bv{x}}$ is an approximate best response, with the agent's utility in best response to $\hat{\bv{x}}$ nearly matching that of $\bv{x}$.

Equipped with these results, we prove that \Clinch{} finds a $\delta$-approximation for $\bv{x}^\star$ in $\widetilde{O}(n \log \frac{1}{\delta})$ queries, and strengthen the guarantee when the provided bounding box is small. We note that this complexity cannot be improved in general beyond logarithmic factors; when $\cX = [0,1]^n$, $n \log \frac{1}{\delta}$ bits are needed to specify $\bv{x}^\star$ up to $\ell_\infty$-precision $\delta$, but each query provides only $\log n$ bits of information.

\begin{theorem}
\label{thm:clinch}
Fix $0 < \delta \leq 1$ and $\underline{\bv{x}},\overline{\bv{x}} \in \R^n$ with $\underline{\bv{x}} \leq \bv{x}^\star \leq \overline{\bv{x}}$. Then \Clinch{} finds $\hat{\bv{x}} \in \cX$ with $\|\hat{\bv{x}}-\bv{x}^\star\|_\infty\leq\delta$ using $O\bigl(n\log\frac{C^2\alpha n}{\delta}\bigr)$ queries to a $\frac{\delta}{33C^3n}$-approximate best response oracle, where $\alpha=\max_{y \in \cY} \overline{x}_y-\underline{x}_y$. In particular, fixing $\underline{\bv{x}} = \bv{0}_n$ and $\overline{\bv{x}} = \bv{1}_n$ gives query complexity $O\big(n \log \frac{C n}{\delta}\big)$.
\end{theorem}
\begin{proof}
Combining \Cref{lem:minimization} and \Cref{lem:stabilization} (with $\lambda = \frac{\delta}{4C^2}$), we find that the strategy $\hat{\bv{x}} \in \cX$ returned by \Clinch{} satisfies $v(\hat{\bv{x}},\brr(\hat{\bv{x}})) \leq v(\bv{x}^\star,\brr(\bv{x}^\star)) + \delta/C$, and $\hat{x}_y > \underline{x}_y$ only for $y \in \smash{\BestResp^{3C\lambda}(\hat{x})}$, using $\smash{O(n\log\frac{C^4 \alpha n}{\delta})}$ queries. First, the approximate minimization guarantee requires that $\hat{x}_y \geq x^\star_y - \delta$ for all $y \in \cY$. Indeed, if $y \not\in \BestResp(\bv{x}^\star)$, then $\bv{x}^\star_y = 0$ and the claim holds trivially. Otherwise, we have
\begin{align*}
    v^y(\hat{x}_y) \leq v(\hat{\bv{x}},\brr(\hat{\bv{x}})) \leq v(\bv{x}^\star,\brr(\bv{x}^\star)) + \delta/C = v^y(\bv{x}^\star_y) + \delta/C.
\end{align*}
Monotonicity of $v^y$ and the lower slope bound of $1/C$ then imply that $\hat{x}_y \geq x^\star_y - \delta$.
Next, we show that approximate conservativeness implies $\hat{x}_y \leq x^\star_y + \delta$ for all $y \in \cY$. Indeed, if $y \in \smash{\BestResp^{3C\lambda}(\hat{\bv{x}})}$, 
\begin{align*}
\label{eq:separation}
    v^y(\hat{x}_y) \geq v(\hat{\bv{x}},\brr(\hat{\bv{x}})) - 3C\lambda \geq v(\bv{x}^\star,\brr(\bv{x}^\star)) - 3C\lambda \geq v^y(x^\star_y) - 3C\lambda,
\end{align*}
and so monotonicity and our slope bound imply that $\hat{x}_y \leq x^\star_y + 3C^2\lambda < x^\star_y + \delta$. Otherwise, we must have $\hat{x}_y = \underline{x}_y \leq x^\star_y$. All together, we have $\|\hat{\bv{x}} - \bv{x}^\star\|_\infty \leq \delta$, as desired.
\end{proof}

\begin{remark}
The minimization stage of \Clinch{} is connected to classic cutting-plane methods. Indeed, feedback $y \in \BestResp^\eps(\bv{x})$ for a query $\bv{x}$ implies $x^\star_y \geq x_y - C\eps$, so the cut $\{ \bv{z} \in \R^n : z_y = x_y \}$ nearly separates $\bv{x}$ from $\bv{x}^\star$. Treating the agent as a noisy separation oracle, \Clinch{} mirrors the query selection of center-of-gravity methods \citep{levin1965,newman1965}. Unlike this setting, we lack an evaluation oracle for the objective $v(\bv{x},\brr(\bv{x}))$ and apply modern adjustments.
Specifically, \Clinch{} adapts the multidimensional binary search \textsc{ProjectedVolume} algorithm  of \cite{lobel2018search} to our noisy axis-aligned setting, with coordinate locking at Step~\ref{step:clinch-coordinate-locking} serving as an analogue of their cylindrification procedure and preventing a quadratic dependence on $n$.
\end{remark}

\subsection{Comparison to prior work with best-responding agents}
\label{ssec:ssg-comparison}

\paragraph{The classical approach: multiple LPs.}
To compare \Clinch{} with prior approaches for SSGs, we recall the standard method for the full-information problem, which requires linear agent utility functions $v^y$. Each best response region $K_y$ is convex as the intersection of $\cX$ with $n$ half-spaces, and we may rewrite the Stackelberg benchmark as a set of $n$ optimization problems, with the $y$-th optimizing $u_y$ over $K_y$: 
\begin{equation*}
    \max_{\bv{x} \in \cX} u(\bv{x},\brr(\bv{x})) = \max_{y\in\cY}\max_{\substack{\bv{x} \in \cX \\ y\in \BestResp(\bv{x})}} u^y(x_y) = \max_{y\in\cY} \max_{\bv{x}\in K_y} u^y(x_y).
\end{equation*}
These inner problems can be solved efficiently, since each $K_y$ is convex and each objective $u^y$ is quasi-convex by monotonicity. This approach was originally developed by \citet{conitzer2006computing} for the case when $\cX$ is a polytope and each $u^y$ is linear, where it is termed ``multiple LPs''.

When $v$ is unknown but the agent is myopic, previous works \citep{letchford2009learning, blum2014, peng2019learning} use exact best response feedback to learn each $K_y$ and apply multiple LPs. Most recently, \cite{peng2019learning} provide an algorithm that finds an optimal strategy using $O(n^3 L)$ best response queries when agent utilities are linear with coefficients specified by $L$ bits each.

\paragraph{Our improvements.} In contrast, \Clinch{} applies in more general environments and provides stronger query complexity guarantees. Even for the full information problem, our method only requires monotonicity of agent utilities and works when these are non-linear, as depicted in \Cref{fig:ssg-non-linear-example}. On the other hand, linearity was necessary for inducing convex best response regions that were crucial for previous methods. 
For the learning problem, our results do not require the stronger $L$-bit precision assumption and extend seamlessly to the approximate best response regime needed to handle non-myopic agents. When specializing to the setting of prior work, we show in \Cref{app:ssg-exact-search} that \Clinch{} finds an \emph{exact} optimizer using $O(n^2L)$ best response queries, improving upon the state-of-the-art $O(n^3L)$ complexity. For $\delta$-\emph{approximate} search with $\delta^{-1} = \poly(n)$, our query complexity  of $\widetilde{O}(n L)$ improves quadratically over prior guarantees (for exact search).

\begin{figure}
  \begin{minipage}[c]{0.8\textwidth}
    \caption{
The principal strategy space $\cX = \{ (x_1,x_2)\in[0,1]^2 : x_1^2 + x_2^2 \leq 1 \}$ for a two-target game with non-convex best response regions $K_1$ and $K_2$, induced by agent payoffs $v^1(x_1) = 1 - \frac{3}{4}(1+e^{5-15x_1})^{-1}$ and $v^2(x_2) = 1-x_2$. Observe that the unique conservative minimizer $\bv{x}^\star$ lies in the intersection $K_1 \cap K_2$.} \label{fig:ssg-non-linear-example}
  \end{minipage}\hfill
    \begin{minipage}[c]{0.15\textwidth}
    \includegraphics[width=\textwidth]{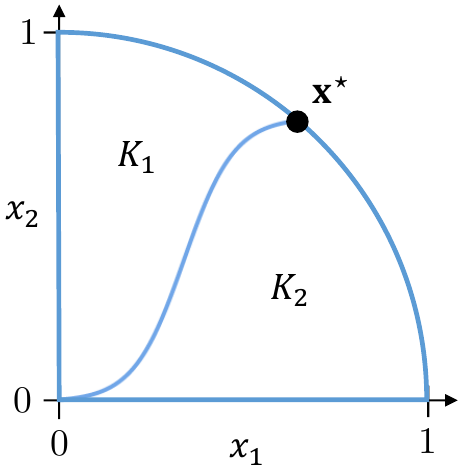}
  \end{minipage}\hfill
\end{figure}

We attribute our improved bounds to two main factors. First, previous works \citep{letchford2009learning, blum2014,peng2019learning} fix a target $y$ and use best responses primarily as a \emph{membership oracle} for the best response polytope $K_y$ (since $y \in \BestResp(\bv{x})$ if and only if $\bv{x} \in K_y$). In contrast, we use this feedback to simulate a \emph{separation oracle}, incorporating more information per query to obtain faster convergence. Since the objective $v(\bv{x},\brr(\bv{x}))$ is quasi-convex, as the maximum of monotonic functions in each coordinate, cutting plane methods can be adapted to obtain $\widetilde{O}(n\log\frac{1}{\delta})$ search complexity.
Second, existing methods solve an LP for each $K_y$, while our structural results imply that $\bv{x}^\star$ belongs to all non-empty best response regions. (Indeed, if $y \not\in \BestResp(\bv{x}^\star)$, then $x^\star_y = 0$ and $v^y(0) < v(\bv{x}^\star,\brr(\bv{x}^\star)) \leq v(\bv{x},\brr(\bv{x}))$ for all $\bv{x} \in \cX$, so $y$ is never a best response.) Hence, it suffices to solve $\max_{\bv{x} \in K_y}x_y$ for \emph{any} non-empty $K_y$ and then eliminate wasteful defensive resources, giving an immediate $n$-factor improvement.

To empirically demonstrate the effectiveness of \Clinch{}, we compare our approach to the \textsc{SecuritySearch} algorithm of \citet{peng2019learning} in Figure~\ref{fig:clinch_vs_securitysearch}. In two simulated settings with best-responding agents, we find that \Clinch{} is more efficient in terms of both the implicit constant factor and asymptotic scaling, outperforming the prior state-of-the-art by several orders of magnitude in sample complexity. In both settings, the principal's strategy space is the simplex and both players' utilities are linear. The first setting uses a fixed, symmetric choice of utility coefficients, while the second is averaged over three choices of coefficients selected uniformly at random. Full details are provided in Section~\ref{app:clinch-experiments}, and code for both implementations is available at \url{https://github.com/sbnietert/learning-stackelberg-games}.

\begin{figure}
    \centering
\subfloat{\includegraphics[width=2.7in]{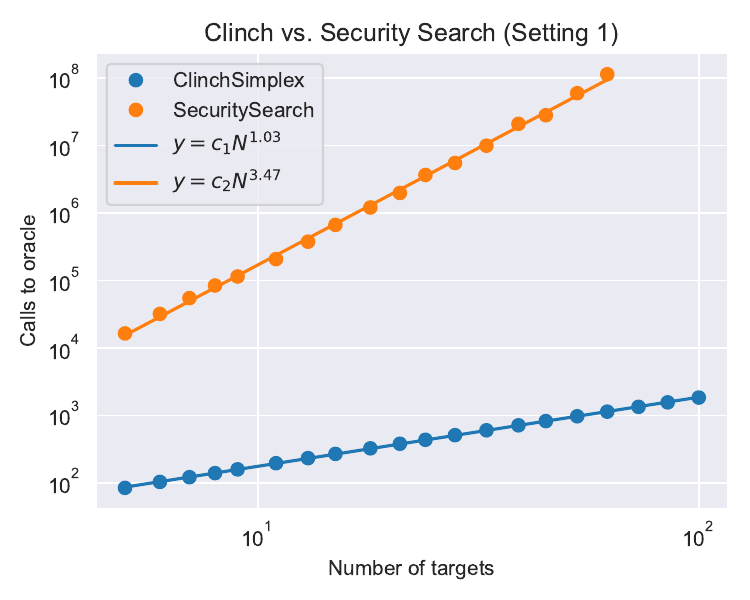}}
    \quad
\subfloat{\includegraphics[width=2.7in]{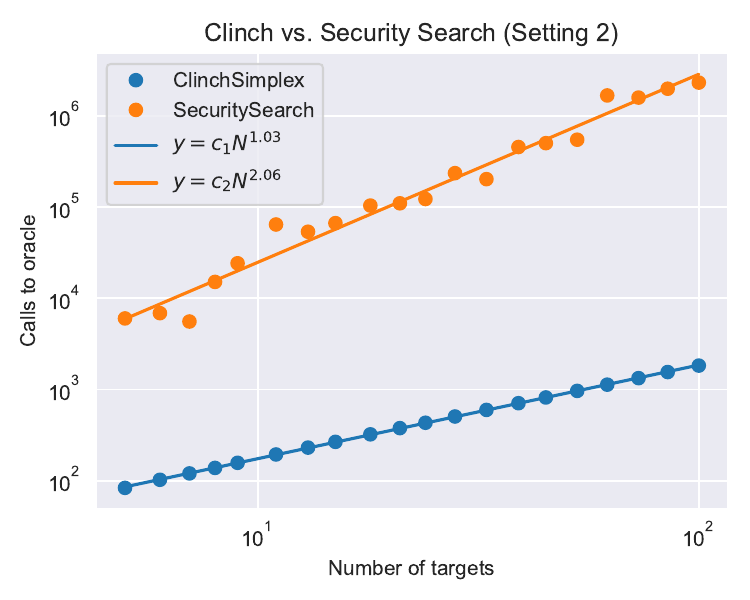}}
    \caption{Query complexity of \Clinch{} versus \SecuritySearch{} in two settings. The $y$-axis shows the number of calls the principal must make to the best response oracle. Both sets of axes are displayed on a log-log scale. Dashed lines depict power law fits of the scaling curves, obtained from log-log linear regression.}
\label{fig:clinch_vs_securitysearch}
\end{figure}

%% file: 4-non-myopic-clinch.tex
\section{Extending \Clinch{} to Non-Myopic Agents}
\label{sec:non-myopic-clinch}

For the non-myopic setting, we first provide an efficient batched algorithm, \BatchedClinch{}, that translates \Clinch{}'s query complexity bound into an $\widetilde{O}(n(\log T + T_\gamma))$ regret bound against $\gamma$-discounting agents (see \Cref{thm:batched-clinch}).
For unknown discount factor, we develop a $\gamma$-agnostic policy \MultiThreadedClinch{} with regret $\widetilde{O}(n \log T(\log T + T_\gamma))$. 
We test the performance of simplified variants of both policies via numerical simulations against a restricted class of non-myopic agents.

\subsection{A batched algorithm for known discount factor}
\label{subsec:ssg-nonmyopic}

The search guarantee of \Clinch{} translates to a single-round regret bound against approximately best-responding agents, after a small perturbation. Given any sufficiently precise estimate $\hat{\bv{x}}$ for $\bv{x}^\star$, we prove (\Cref{prf:perturbation}) that the principal can identify the true best response set $\BestResp(\bv{x}^\star)$ of $\bv{x}^\star$ and should slightly reduce the weight placed on a target $\hat{y}$ in this set which maximizes $u(\hat{\bv{x}},\hat{y})$.

\begin{lemma}
\label{lem:perturbation}
Fix $0 < \lambda \leq 1$ and suppose $\hat{\bv{x}} \in \cX$ with $\|\hat{\bv{x}} - \bv{x}^\star\|_\infty \leq \frac{W\lambda}{6C^2}$. Then the perturbed strategy $\tilde{\bv{x}} = \Perturb(\hat{\bv{x}},\lambda) \coloneqq \hat{\bv{x}} - \frac{W\lambda}{2} \mathbf{e}_{\hat{y}}$, where $\hat{y} \in \argmax_{y \in \cY:\hat x_y > W/2} u(\hat{\bv{x}}, y)$, belongs to $\cX$ and satisfies $u(\tilde{\bv{x}},y) \geq u(\bv{x}^\star,\brr(\bv{x}^\star)) - \lambda$ whenever $y \in \BestResp^\eps(\tilde{\bv{x}})$ for $\eps = \frac{W\lambda}{200C^5n}$.
\end{lemma}

For $\gamma$-discounting agents, fix $\lambda = nT_\gamma/T$ and consider the policy which runs \Clinch{} with accuracy $\delta = \frac{W\lambda}{6C^2}$, na\"ively repeating each query $\smash{D = T_\gamma \log \frac{T_\gamma}{\eps}}$ times to induce $\eps=\frac{W\lambda}{200C^5n}$-approximate best responses (only using feedback from each batch's first query). By Proposition~\ref{prop:delayed-feedback} and Theorem~\ref{thm:clinch}, the returned strategy $\hat{\bv{x}}$ satisfies $\|\hat{\bv{x}} - \bv{x}^\star\|_\infty \leq \frac{W\lambda}{6C^2}$. By Lemma~\ref{lem:perturbation}, we can commit to $\tilde{\bv{x}} = \Perturb(\hat{\bv{x}},\lambda)$ for the remaining rounds, incurring at most $O(Dn\log\frac{CT}{WT_\gamma} + T_\gamma n) = O(nT_\gamma \log^2 \frac{CT}{W})$ total regret. However, this approach is 
wasteful, as it does not exploit any knowledge gained until $\tilde{\Omega}(n T_\gamma )$ rounds have passed. Instead, we should use the extra delay rounds for exploitation.

Specifically, we introduce \BatchedClinch{} (\Cref{alg:batched-clinch}), which maintains a running estimate $\tilde{\bv{x}}$ of the current best strategy. For each epoch $\phi=1,2,\dots$, this algorithm runs \Clinch{} with accuracy $\lambda=2^{-\phi}$ in batches of size $O\big(T_\gamma \log \frac{T_\gamma C n}{W \lambda}\big)$, initialized with $\underline{\bv{x}}$ and $\overline{\bv{x}}$ set to bounds implied by the previous search. We update \Clinch{} during the first round of each batch, playing $\tilde{\bv{x}}$ for the others, and then setting $\tilde{\bv{x}}$ to the estimate returned by \Clinch{} (after perturbation). When all epochs have completed, the final estimate $\tilde{\bv{x}}$ is played for the remaining rounds. For ease of presentation, we identify entry-wise lower and upper bounds with their corresponding axis-aligned bounding box.\medskip

\begin{algorithm}[H]
\caption{\BatchedClinch{}}\label{alg:batched-clinch}
\DontPrintSemicolon
\SetAlgoNoLine
    $\tilde{\bv{x}} \gets (0,\dots,0) \in \R^n, \: B \gets [0,1]^n$\;
    \For{epoch $\phi=1,\dots,\lceil \log_2 T \rceil$}{
        Initialize \Clinch{} with accuracy $\delta =\frac{W\lambda}{6C^2}$ for $\lambda=2^{-\phi}$ and axis-aligned bounding box $B$\;
        \While{\Clinch{} has not terminated}{
            Simulate query/response for \Clinch{} with next $\bv{x}^{(t)},y_t$ pair \label{step:batched-clinch-explore} \tcp*{Explore}
            Play strategy $\tilde{\bv{x}}$ for next $\left\lceil T_\gamma \log \frac{200T_\gamma C^5 n}{W\lambda}\right\rceil$ rounds \label{step:batched-clinch-exploit} \tcp*{Exploit}
        }
        $\tilde{\bv{x}} \gets \Perturb(\hat{\bv{x}},\lambda)$ for $\hat{\bv{x}}$ returned by \Clinch{}\;
        $B \gets \big\{ \bv{x} \in \R^d : \|\bv{x} - \hat{\bv{x}}\|_\infty \leq \delta \big\}$ \label{step:batched-clinch-bound-updates} \tcp*{Update bounding box}
    }
    Play strategy $\tilde{\bv{x}}$ for remaining rounds \label{step:batched-clinch-commit-final}\tcp*{Exploit}
\end{algorithm}

\begin{theorem}
\label{thm:batched-clinch}
\BatchedClinch{} incurs regret $O\left( n\log\frac{Cn}{W}\log T + n T_\gamma \log^2 \frac{CnT_\gamma}{W}\right)$ against $\gamma$-discounting agents, for known discount factor $\gamma \in (0,1)$.
\end{theorem}

\begin{proof}
During and after epoch $\phi$, when $\delta =\frac{W \lambda}{6C^2}$ for $\lambda = 2^{-\phi}$, the principal's feedback delay incentivizes the agent to $\frac{W\lambda}{200C^5}$-approximately best respond, by \Cref{prop:delayed-feedback}. Moreover, the entry-wise lower and upper bounds are trivially valid at the start and remain valid due to the $\ell_\infty$ guarantee of \Cref{thm:clinch}.
By the same theorem, epoch $\phi=1$ terminates after $O\big(n \log \frac{C n}{W}\big)$ batches of size $\smash{O\big(T_\gamma \log \frac{T_\gamma C n}{W}\big)}$ and epoch $\phi > 1$ terminates after $O\big(n \log C n\big)$ batches of size $\smash{O\big(T_\gamma \log((T_\gamma C n 2^i)/W))}$. The resulting strategy $\tilde{x}$ incurs regret at most $2^{-i}$ when played in a future round (by the initial observation and \Cref{lem:perturbation}). We bound total regret obtained during exploration (Step~\ref{step:batched-clinch-explore}) by
     $O\left(n \log \frac{Cn}{W} + \log T \cdot n \log C n\right) = O\left(n\log\frac{Cn}{W}\log T\right)$.
The regret from exploitation rounds (Step~\ref{step:batched-clinch-exploit}) is at most
    $O\left(n T_\gamma \sum_{i=1}^{\lceil \log_2 T\rceil} i 2^{-i} \log \frac{C n}{W} \log \frac{T_\gamma C n}{W}\right)= O\left(n T_\gamma \log \frac{Cn}{W} \log\frac{CnT_\gamma}{W} \right)$,
where the equality uses that $\sum_{i=1}^m i2^{-i} = O(1)$. Finally, Step~\ref{step:batched-clinch-commit-final} contributes $O(1)$ regret.
\end{proof}

\subsection{A multi-threaded algorithm for unknown discount factor}

\BatchedClinch{} requires the principal to know the discount factor~$\gamma$. When $\gamma$ is unknown, we adapt the multi-layer approach of \cite{LykourisMirrokniPaesLeme18}, running $\log T$ copies of \Clinch{} in parallel threads, where thread $r$ experiences delay $2^r$ between queries. Each thread's delay corresponds to a guess $\hat{\gamma}$ for $\gamma$, with search guarantees only holding if $\hat{\gamma} \geq \gamma$ and sufficiently accurate best response feedback is induced. 
As thread $r$ performs a search to accuracy $O(1/T)$, it maintains a shrinking bounding box $B_\mathrm{priv}^{(r)}$ around $\bv{x}^\star$. To maintain appropriate feedback delay, this private state is only published to other threads, as $\smash{B_\mathrm{pub}^{(r)}}$, after $2^r$ rounds have passed. When search completes, it performs exploitation by perturbing a strategy in the intersection of the public boxes for threads $\geq q$, for the smallest $q$ such that this intersection is non-empty. By design, $\cap_{p \geq q} B^{(p)}_\mathrm{{pub}}$ contains $\bv{x}^\star$ and is contained by $B^{(r^\star)}_\mathrm{{pub}}$,
corresponding to the correct guess for $\gamma$. The resulting $\gamma$-agnostic policy \MultiThreadedClinch{} (\Cref{alg:multi-threaded-clinch}) achieves the following bound. 

\begin{theorem}
\label{thm:multi-threaded-clinch}
\MultiThreadedClinch{} incurs regret $O\big(n \log \frac{Cn}{W} \log^2 T + nT_\gamma \log \frac{Cn}{W} \log\frac{CnT_\gamma T}{W}\big)$ against $\gamma$-discounting agents, for unknown discount factor $\gamma \in (0,1)$
\end{theorem}

\begin{proof}
By design, thread $r$ runs on rounds $2^{r-1}(2k-1)$ for $k=1,2,\dots$ and does not update its public state until $2^r$ rounds have passed. Hence, the agent experiences feedback delay $2^r$ on rounds during which thread $r$ is active. Write $r^\star = \big\lceil \log_2 \big(T_\gamma \log \frac{400 C^5 n T_\gamma T}{W}\big) \big\rceil$ for the index of the first thread whose delay induces $\frac{W}{400 C^5 n T}$-approximate best responses, using \Cref{prop:delayed-feedback} (we can assume that $r^\star \leq \lfloor \log T \rfloor + 1$; otherwise the regret bound holds trivially). By \Cref{thm:clinch}, feedback for each thread $r \geq r^\star$ is sufficiently accurate so that the bounding boxes $B^{(r)}_\mathrm{pub}, B^{(r)}_\mathrm{priv}$ always contain $\bv{x}^\star$. Consequently, $q$ the confidence set $B$ selected at Step~\ref{step:multi-threaded-clinch-intersection} always contains $\bv{x}^\star$. 

To bound regret, we note that all threads terminate exploration after $O(n \log\frac{Cn}{W} \log T)$ updates, using the same bound as in the proof \Cref{thm:batched-clinch}; this gives a total exploration cost of $O(n \log\frac{Cn}{W} \log^2 T)$. For exploitation, we consider the separate runs of $\cA^{(r^\star)}$; by \Cref{thm:clinch}, the first run takes $O(n \log \frac{Cn}{W})$ queries, and the remaining $O(\log T)$ runs take $O(n \log Cn)$ queries each. While the $i$th run is in progress, corresponding to $\delta = \frac{W}{6C^2 2^i}$, \Cref{lem:perturbation} implies that no exploit round from \emph{any} thread can incur regret more than $2^{2-i}$. Indeed, $\bv{x}^\star \in B \subseteq B^{(r^\star)}_\mathrm{pub}$ and $B^{(r^\star)}_\mathrm{pub}$ has side width at most $4\delta$ during this run.  Consequently, we bound exploitation regret by \begin{equation*}
    O\left(2^{r^\star} \left[n \log \frac{Cn}{W} + \sum_{i=1}^{\log T} n \log(Cn) 2^{2-i}\right]\right) = O\left(n T_\gamma \log \frac{Cn}{W} \log \frac{Cn T_\gamma T}{W}\right).
    \qedhere
\end{equation*}
\end{proof}

\begin{algorithm}[t]
\caption{\MultiThreadedClinch}\label{alg:multi-threaded-clinch}
\DontPrintSemicolon
\SetAlgoNoLine
    \For{thread $r=1,\dots,\lfloor \log T \rfloor + 1$}{
        Initialize copy $\cA^{(r)}$ of \Clinch{} with accuracy $\delta^{(r)} \gets \frac{W}{12C^2}$\;
        Initialize public and private bounding boxes $B^{(r)}_\mathrm{pub}, B^{(r)}_\mathrm{priv} \gets [0,1]^n$\;
    }
    \For{round $t = 1,\dots,T$}{
        $r \gets \argmax\{k \in \mathbb{N}_{>0} : 2^{k-1} \text{ divides } t \}, \;B^{(r)}_\mathrm{pub} \gets B^{(r)}_\mathrm{priv}$
        
        \If(\tcp*[f]{Explore}){$\delta^{(r)} > \frac{W}{12C^2T}$}{
            Simulate query/response for $\cA^{(r)}$ using $\bv{x}^{(t)}, y_t$\;
            \If{$\cA^{(r)}$ has terminated, with output $\tilde{\bv{x}} \in \cX$}{
                Restart $\cA^{(r)}$ with $\delta^{(r)} \gets \delta^{(r)}/2$ and $B^{(r)}_\mathrm{priv}\!\gets\!\big\{\bv{x} \in \R^d : \|\bv{x} - \tilde{\bv{x}}\|_\infty \leq \delta^{(r)}\big\}$
            }
        }
        \Else(\tcp*[f]{Exploit}){
            $\hat{\bv{x}} \gets \E_{\bv{x} \sim \Unif(B)}[\bv{x}]$ where $B = \cap_{p \geq q} B^{(p)}_\mathrm{pub}$ for min $q$ such that intersection non-empty\; \label{step:multi-threaded-clinch-intersection}
            Play $\bv{x}^{(t)} \gets \Perturb(\hat{\bv{x}},1/T)$
        }
    }
\end{algorithm}

\begin{remark}
As we see in Sections \ref{ssec:demand} and \ref{ssec:finite-games}, this approach for handling unknown discount factors extends beyond SSGs to settings where we wish to estimate a \emph{ground truth} and there exists an algorithm for the myopic setting that maintains and aggressively shrinks confidence sets about this quantity. This is analogous to the settings where an unknown number of adversarial corruptions can be handled by the multi-threading technique \citep{LykourisMirrokniPaesLeme18,ChenKrishnamurthyWang23,LykourisSimchowitzSlivkinsSun23,KrishnamurthyLP21,GolrezaeiManSchSek22,ChenWang22}.
\end{remark}

\subsection{Experiments with simulated discounting agents}
\label{ssec:non-myopic-clinch-experiments}

We now empirically test the performance of these algorithms. Since exact simulation of $\gamma$-discounting agents is computationally intractable, we instead simulate a restricted class of discounting agents with bounded rationality. At each round $t$, our simulated agent picks target $y \in \cY$ maximizing its future $\gamma$-discounted utility if it plays $y$ at round $t$ and best responds in future rounds. While less complex, this agent model still poses major obstacles to any principal algorithm expecting best responses, and our regret bounds still apply (since Proposition~\ref{prop:delayed-feedback} holds). We also opt for slightly simplified implementations of \BatchedClinch{} and \MultiThreadedClinch{}, which are faster to simulate and still satisfy logarithmic regret bounds. Full details for the algorithms and agent model are provided in \cref{app:non-myopic-clinch-experiments}.

\begin{figure}
    \centering
    \includegraphics[width=0.9\linewidth]{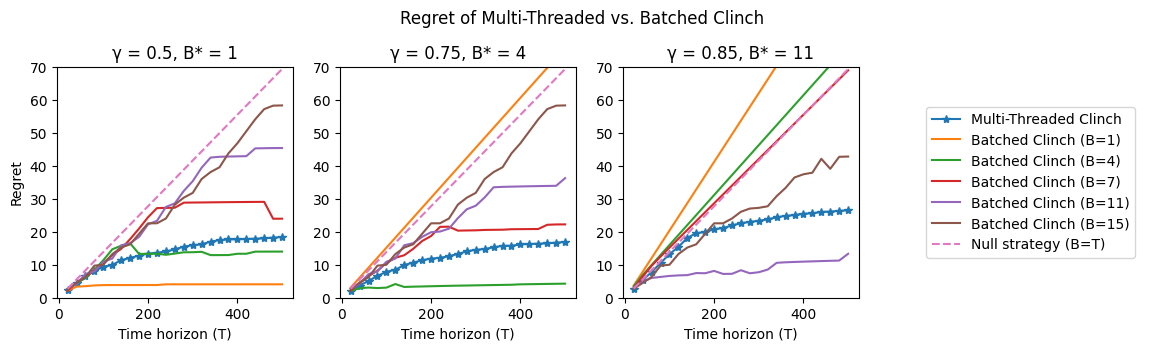}
    \caption{Regret achieved by batched and multi-threaded variants of \Clinch{} against a simulated $\gamma$-discounting agent on a random SSG instance. For each discount factor, we note the optimal batch size $B^\star$ at $T=500$. We note that $B=1$ corresponds to the baseline $\Clinch{}$ algorithm designed for myopic agents.}
    \label{fig:multi-threaded-simulations}
\end{figure}

The batched algorithm accepts a batch size $B$, which must be taken to scale with $T_\gamma$ to ensure sublinear regret.
In Figure~\ref{fig:multi-threaded-simulations}, we compare the performance of multi-threaded versus batched \Clinch{} with varied $B$ on a randomly sampled SSG simplex instance with $n=3$ targets. Each plot shows regret as a function of time horizon $T \in [20, 500]$ for a different choice of $\gamma \in \{0.5, 0.75, 0.85\}$. These results demonstrate that batched \Clinch{} can achieve sublinear regret, but that its performance deteriorates if $B$ is too small or too large. On the other hand, multi-threaded \Clinch{} incurs a mild overhead over the best batch size $B^\star = B^\star(\gamma)$ with no tuning. Appendix~\ref{ssec:non-myopic-clinch-extra-experiments-geometric} includes full details and plots for four more random SSGs, all of which exhibit these trends. We also provide plots for an extended set of simulations with $n=10$. Here, the overhead of multi-threading is more pronounced in some regimes but still distinctly sub-linear, compared to the linear regret suffered by small batch sizes. Code is provided at \url{https://github.com/sbnietert/learning-stackelberg-games}.

\begin{remark}[Non-geometric discounting]
\label{rem:hyperbolic-discounting}
Our multi-threaded approach does not rely on the specifics of geometric discounting; we only require that delays of sufficiently high-indexed threads induce approximate best responses from the agent. In Appendix~\ref{ssec:non-myopic-clinch-extra-experiments-hyperbolic}, we corroborate this empirically with an additional set of experiments against a simulated hyperbolically-discounting agent.
\end{remark}

%% file: 5-beyond-SSGs.tex
\section{Applications beyond SSGs}
\label{sec:beyond-ssgs}

We now apply our framework to settings beyond SSGs. We treat demand learning in Section~\ref{ssec:demand}, general finite Stackelberg games in Section~\ref{ssec:finite-games}, and strategic classification in Section~\ref{ssec:strategic-classification}.

\subsection{Pricing with an unknown demand curve}\label{ssec:demand}
In the demand learning problem \citep{kleinberg03value}, a price-setting principal seeks to maximize revenue from selling a good to a returning buyer with demand curve induced by an unknown value distribution. Each round, the principal posts a price and the buyer decides whether to purchase based on their realized value. This problem was among the first examined with non-myopic agents \citep{amin2013learning,mohri2014}, as it arises naturally in settings like online advertising where strategic buyers may try to trick the seller into providing low prices.

When the buyer's value is fixed, our learning task mirrors binary search, and we adapt the batching method of \BatchedClinch{} to obtain a policy \BatchedBinarySearch{} with regret $\widetilde{O}(\log T + T_\gamma)$, improving upon the state-of-the-art when $T_\gamma = \Omega(\log T)$.
For stochastic values, the problem reduces naturally to an instance of stochastic multi-armed bandits---typically solved by adaptive exploration rather than explore-then-commit---and so a different approach is required. Fortunately, a classic policy \SE{} extends seamlessly to the delayed feedback setting and exhibits natural robustness to bounded adversarial perturbations. This policy achieves regret $\widetilde{O}(\sqrt{T} + T_\gamma)$, with only an additive overhead in $T_\gamma$ compared to the standard $\widetilde{O}(\sqrt{T})$ bound.

\paragraph{Model and preliminaries.}
A \emph{posted-price single-buyer auction} is a Stackelberg game where the principal (``seller'') sets a price $p \in [0,1]$ of a single good and the agent (``buyer'') decides whether to buy ($a=1$) or not ($a=0$) at the posted price. The buyer has value $v \in [0,1]$ for the good and receives payoff $a(v-p)$, while the seller receives $pa$. In the stochastic setting ($v$ sampled from a distribution $\Dist$), the seller's expected revenue for posting price $p$ is $f(p) = p d(p)$, where $d(p) = \Pr_{v \sim \Dist}(v \geq p)$ is the buyer's \emph{demand curve}. 

We consider the repeated game where values $v_1,\dots,v_T$ of a returning buyer are sampled i.i.d.\ from $\Dist$, unknown to the seller, where $\Dist$ either (i) is supported on a single value or (ii) satisfies mild regularity assumptions described below.
Denoting the game's history by $\{(p_t,a_t)\}_{t=1}^T$, the seller seeks to maximize revenue $\sum_{t=1}^T p_t a_t$, while the buyer maximizes discounted profit $\sum_{t=1}^T \gamma^t a_t(v_t - p_t)$. Stackelberg regret for the seller is $T \max_p p d(p)  - \E\left[\sum_{t=1}^T p_t a_t\right]$, since a myopic agent buys the good when their value exceeds the posted price. As before, we write $\BestResp^\eps_t(p) = \big\{ a \in \{0,1\} : a (v_t-p)\ge \max\{v_t-p,0\}- \eps \big\}$ for the $\eps$-approximate best response set at time $t$.

\paragraph{Connection to multi-armed bandits.}
In case (ii), \citet{kleinberg03value} reduce this task to a \emph{stochastic multi-armed bandits} problem, a setting which we recall briefly (see \cite{slivkins2019introduction} for a textbook treatment). In this model, a principal interacts with a set of $K$ ``arms'' over $T$ rounds. The arms are indexed by $i\in [K]$, and each arm has an associated reward distribution $\Dist_i$ supported on $[0, 1]$. Write $\mu_i = \E_{r\sim\Dist_i}(r)$ and $\Delta_i = \max_{i'} \mu_{i'} - \mu_i$. During the $t$-th round, the principal must pull an arm $i_t\in [K]$; after doing so, they observe a reward $r_t\sim \Dist_i$. The performance of a policy is benchmarked against the best arm in expectation, with regret defined as $\E\left[\sum_{t=1}^T \Delta_{i_t}\right]$. To view pricing through this lens, \cite{kleinberg03value} discretize the space of possible prices into the set $\big\{\frac {i}{K} : 1\le i\le K\big\}$. Each of these $K$ possible prices can then be thought of as a bandit arm that the seller can pull, where pulling the $i$-th arm corresponds to posting a price of $\frac {i}{K}$, and the two notions of regret coincide up to a small difference in benchmarks due to discretization.

\paragraph{Regularity assumptions.}
For case (ii), we assume that the demand curve $d(p) = \Pr_{v \sim \Dist}(v \geq p)$ is $L$-Lipschitz, which is standard for this setting (see, e.g., \cite{amin2013learning}). Moreover, following \cite{kleinberg03value}, we assume that $f(p)$ achieves its maximum for a unique $p^\star \in (0,1)$ with $f''(p^\star) < 0$. This implies that there exist constants $C_1,C_2$ with $C_1(p^\star - p)^2 \leq f(p^\star) - f(p) \leq C_2(p^\star - p)^2$. Our work requires knowledge of $L$ (this can be relaxed, see Remark~\ref{rem:unknown-params-pricing}) but not $C_1$ and $C_2$. \cite{kleinberg03value} show the following (corollary 3.13 and theorem 3.14 therein).

\begin{lemma}
\label{lemma:discretization}
The discretization error $f(p^\star) - \max_i f(i/K)$ is at most $\frac{C_2}{K^2}$.
\end{lemma}
\begin{lemma}
\label{lemma:gaps}
The sum of inverse gaps $\sum_{\Delta_i > 0} \frac{1}{\Delta_i}$ is at most $\frac{6K^2}{C_1}$.
\end{lemma}

\paragraph{Pricing with a fixed value.}
Returning to case (i) where $\Dist$ is concentrated on an unknown $v \in [0,1]$, feedback $\mathds{1}\{v \geq p_t\}$ from a myopic agent is sufficient to perform binary search, implying an explore-then-commit $O(\log T)$ regret bound. Looking closely, we can reinterpret the problem as a security game with two targets representing the buyer's purchase choices. Indeed, we can set $\cY = [2]$, map $a$ to $y = a+1$, and take $\cX = \Delta_1^{\leq}$, where price $p$ is mapped to $\bv{x} = (1-p,p)$. The principal's utility $ap$ is then given by $u(\bv{x},y) = u(x_y,y) = (y-1)x_y$, and the agent's by $(y-1)(v-x_y)$; both are appropriately monotonic in $x_y$ and satisfy the slope condition with $C = 1$. Finally, the optimal principal strategy is $\bv{x}_\star = (1-v,v)$, so the SSG regret coincides with that for dynamic pricing.

This motivates us to implement delays via the batching approach of \BatchedClinch{}, lengthening batches as search progresses and committing to prices with low regret during non-exploration rounds. We simplify this approach to a policy \BatchedBinarySearch{} for the present setting.\medskip

\begin{algorithm}[H]
\DontPrintSemicolon
\SetAlgoNoLine
\caption{\BatchedBinarySearch}\label{alg:batched-binary-search}
$\ell \gets 0, u \gets 1, \hat{v} \gets 0$\;
\While{$u - \ell > 1/T$}{
    Set $\eps = (u - \ell)/4$ and post price $p = (\ell+u)/2$\;
    \textbf{if} agent buys good \textbf{then} $\ell \gets p - \eps$ \textbf{else} $u \gets p + \eps$\;
    Post price $\hat{v}$ for next $T_\gamma \log \frac{T_\gamma}{\eps}$ rounds\; \label{step:batched-binary-search-delay}
    $\hat{v} \gets \max\{\ell - \eps,0\}$
}
Post price $\hat{v}$ for remaining rounds \label{step:batched-binary-search-commit}
\end{algorithm}

\begin{theorem}
\BatchedBinarySearch{} incurs regret $O(\log T + T_\gamma \log T_\gamma)$ against $\gamma$-discounting agents with a fixed value.
\end{theorem}
\begin{proof}
First, we note that $v$ always lies in the interval $[\ell,u]$, and that this interval shrinks by a factor of 3/4 between iterations. Indeed, when this interval has width $4\eps$, our feedback delay ensures that the buyer’s decision is a best response for some perturbed value $v'$ with
$|v'-v| \leq \eps$ (by \Cref{prop:delayed-feedback}), and so the updates are sound. Consequently, we always have $\hat{v} < v$ (unless $v = 0$, in which case any policy suffices), and so the buyer purchases the good at Steps~\ref{step:batched-binary-search-delay} and \ref{step:batched-binary-search-commit}, since there is no incentive to deviate from best response during these rounds.
Hence, we incur at most $4 \cdot \frac{4}{3}\eps$ regret for each round of Step~\ref{step:batched-binary-search-delay} and at most regret 1 after search concludes, giving a total bound of
\begin{equation*}
    \sum_{i=1}^{\left\lceil\log_{4/3} T\right\rceil} \left(1 + 4 \cdot 0.75^{i-1} T_\gamma \log \left(T_\gamma 1.4^i \right)\right) = O(\log T + T_\gamma \log T_\gamma).\qedhere
\end{equation*}
\endproof
\end{proof}

\begin{remark}[Comparison to prior work]
Noting that over-pricing is costlier than under-pricing for the seller, \cite{kleinberg03value} beat binary search with a policy attaining regret $O(\log\log T)$ for the myopic setting. With non-myopic agents, a line of work \citep{amin2013learning,mohri2014,drutsa2017horizon} has brought regret down to $O(T_\gamma \log T_\gamma \log\log T)$ via a delayed search policy of \cite{drutsa2017horizon}, compared to a lower bound of $\Omega(\log\log T + T_\gamma)$ implied by \cite{kleinberg03value} and \cite{amin2013learning}. When $T_\gamma = O(1)$, this regret is optimal; however, for $T_\gamma = \Omega(\log T)$, our bound of $O(T_\gamma \log T_\gamma)$ is a $\log \log T$ improvement.
\end{remark}

\paragraph{Pricing with stochastic values.}
To address case (ii), we first consider stochastic bandits with perturbed and delayed feedback. Formally, we say that (potentially adversarially adaptive) feedback $r_1,\dots,r_T$ is \emph{$\delta$-perturbed} from that of the original bandits problem if, conditioned on any arm sequence $i_1, \dots, i_T$ with positive probability, there exist independent random intervals $\{[\ell_t,u_t]\}_{t=1}^T$ such that each $r_t$ lies in $[\ell_t,u_t]$ almost surely and that $\mu_{i_t} - \delta \leq \E[\ell_t] \leq \E[u_t] \leq \mu_{i_t} + \delta$. (Note that this definition of perturbations via couplings strictly generalizes the setting where each reward is shifted by $\pm \delta$ prior to observation.) Recall that classic algorithms for standard stochastic bandits like \UCB{} \citep{auer02finite} and \SE{} \citep{evendar06action} achieve regret $O\big(\sum_{\Delta_i > 0} \frac{\log T}{\Delta_i}\big)$. Here, we apply a simple variant \SEDelayed{} (\Cref{alg:SE-delayed}) of \SE{}, first analyzed by \cite{lancewicki2021stochastic} for a broader class of delays and without perturbations. Each phase of this policy pulls all arms, updates confidence intervals for reward means based on $D$-delayed feedback, and removes suboptimal arms. Our regret bound incurs overhead $\delta T$ from perturbations and $D \log K$ from delays.

\begin{lemma}
\label{prop:robust-delayed-bandits}
For $K$-armed stochastic bandits with $\delta$-perturbed rewards and $D$-delayed feedback, \SEDelayed{} achieves regret $O\left(\sum_{\Delta_i > 0} \frac{\log T}{\Delta_i} + \delta T + D \log K\right)$.
\end{lemma}

\let\oldnl\nl%
\newcommand{\nonl}{\renewcommand{\nl}{\let\nl\oldnl}}%

\begin{center}
\begin{minipage}{0.52\textwidth}
\begin{algorithm}[H]
\caption{\SEDelayed{}}\label{alg:SE-delayed}
\DontPrintSemicolon
\SetAlgoNoLine
\SetKwInOut{Input}{input}
\Input{arm count $K$, delay $D$, error bound $\delta$}
    $S\gets\{1,\dots,K\}$; $t \gets 1$\;
    \While{$t < T$}{
        Pull each arm $i \in S$ and observe feedback\;
        $t \gets t + |S|$\;
        $\SEUpdate{}(S,t-D,\delta)$\;
        $S \gets \{ i \in S : \mathrm{UCB}_i \geq \mathrm{LCB}_j \text{ for all } j \in S\}$\; \label{step:SE-eliminate}
    }
\end{algorithm}
\end{minipage}
\hfill
\begin{minipage}{0.47\textwidth}
\begin{algorithm}[H]
\caption{\SEUpdate}\label{alg:se-update}
\DontPrintSemicolon
\SetAlgoNoLine
\SetKwInOut{Input}{input}
\SetKwInOut{Output}{output}
\nonl\textbf{input:} arm set $S$, time $t$, error bound $\delta$\;
\Output{$\mathrm{LCB}_i$ and $\mathrm{UCB}_i \:\forall i \in S$}
    \For{arm $i\in S$}{
        $n \gets \max\{\sum_{\tau=1}^{t}\mathds{1}\{i_\tau = i\},1\}$\;
        $\hat{\mu}_i \gets \frac{1}{n}\sum_{\tau=1}^{t} \mathds{1}\{i_\tau = i\} r_\tau$\;
        $\mathrm{LCB}_i \gets \hat{\mu} - \sqrt{2\log(T)/n} - \delta$\;
        $\mathrm{UCB}_i \gets \hat{\mu} + \sqrt{2\log(T)/n} + \delta$
    }
\end{algorithm}
\end{minipage}
\end{center}

\begin{proof}[Proof Sketch]
In \citep{lancewicki2021stochastic}, a $O\left(\sum_{\Delta_i > 0} \frac{\log T}{\Delta_i} + D \log K\right)$ regret bound is given for this policy with unperturbed feedback. They observe that if $m$ arms remain after an iteration where an arm would have been eliminated without delays, then this arm is pulled at most $O(D/m)$ extra times before elimination (since we round robin over remaining arms). Summing over all arms, the delay overhead is at most $D \sum_{i=1}^K \frac{1}{i} = O(D\log K)$. We prove in \Cref{prf:robust-delayed-bandits} that at most $\delta T$ additional regret is incurred due to the potential $\delta$ inaccuracy of the confidence bounds.
\end{proof}

To apply this policy to pricing with $\eps$-approximately best-responding agents, we use the following lemma, whose proof in \Cref{prf:revenue-error-bd} is an immediate consequence of a standard error bound ($a \in \BestResp_t^\eps(p)$ implies $a \in \BestResp_t(p')$ with $|p' - p_t| \leq \eps$) and $L$-Lipschitzness of the demand curve $d$.

\begin{lemma}
\label{lem:revenue-error-bd}
Let $\ell = \mathds{1}\{v_t > p - \eps\}$ and $u = \mathds{1}\{v_t \geq p + \eps\}$ for some round $t$ and $\eps \geq 0$. If $a \in \BestResp_t^\eps(p)$, then $\ell \leq a \leq u$, and, if $v_t \sim \Dist$, then $f(p) - L\eps \leq p\E[\ell] \leq p\E[u] \leq f(p) + L\eps$.
\end{lemma}

Finally, we apply \Cref{prop:delayed-feedback} to obtain a $\widetilde{O}(\sqrt{T} + T_\gamma)$ regret guarantee for demand learning.

\begin{theorem}
\label{thm:stochastic-pricing}
In stochastic demand learning, \SEDelayed{}
$(T^{1/4}, T_\gamma \log (L T_\gamma T),T^{-1})$ incurs regret $O\left((C_2 + C_1^{-1})\sqrt{T \log T} + T_\gamma \log^2 (L T_\gamma T)\right)$ against $\gamma$-discounting agents.
\end{theorem}
\begin{proof}
By \Cref{lem:revenue-error-bd}, we see that the bandits problem is $L\eps$-perturbed from the myopic setting if the agent is $\eps$-approximately best responding. Combining \Cref{lemma:discretization,prop:robust-delayed-bandits}, we then find that \SEDelayed{} with $K$ arms, delay $D$, and error bound $\delta = T^{-1}$ achieves regret
\begin{equation*}
    O\left(\sum_{i=1}^K \frac{\log T}{\Delta_i} + 1 + D \log K + \frac{C_2 T}{K^2}\right)
\end{equation*}
for demand learning with $\eps$-approximately best-responding agents if $\eps \leq (LT)^{-1}$. Taking $D = T_\gamma \log (L T_\gamma T)$, \Cref{prop:delayed-feedback} gives the same bound for $\gamma$-discounting agents. Controlling the first term with \Cref{lemma:gaps} and fixing $K = (T/\log(T))^{1/4}$, we bound total regret by
\begin{equation*}
    O\left(\sqrt{T \log T} / C_1 + T_\gamma \log (L T_\gamma T) \log T + C_2 \sqrt{T \log T} \right) = O\left((C_2 + C_1^{-1}) \sqrt{T \log T} + T_\gamma \log^2 (L T_\gamma T) \right).\qedhere
\end{equation*}
\end{proof}

\begin{remark}[Comparison to prior work]
\cite{kleinberg03value} address this setting with myopic agents, where they obtain regret $\widetilde{O}(\sqrt{T})$. \cite{amin2013learning} examine the non-myopic setting but assume a finite
price set without bounding discretization error. \cite{mohri2014} establish a regret bound of $\widetilde{O}(\sqrt{T} + T^{1/4}T_\gamma)$ using a variant of \UCB{} for a weaker class of $\eps$-strategic agents. In comparison, our bound decomposes $T_\gamma$ from polynomial dependence on $T$.
\end{remark}

\begin{remark}[Unknown $L$ or $\gamma$]
\label{rem:unknown-params-pricing}
The delay $D$ in Theorem~\ref{thm:stochastic-pricing} depends on both $L$ and $\gamma$. If either are unknown, we can apply the approach of \cite{LykourisMirrokniPaesLeme18} as in \Cref{subsec:ssg-nonmyopic} to obtain matching regret $\widetilde{O}(\sqrt{T} + T_\gamma)$ for stochastic values, since the overhead is logarithmic in $T$. Each of the $\log T$ threads (associated with a guess for $D$) performs arm elimination at Step~\ref{step:SE-eliminate} according to an intersection of confidence bounds over multiple threads which mirrors Step~\ref{step:multi-threaded-clinch-intersection} of \MultiThreadedClinch{}. For fixed values, we inherit the $\widetilde{O}(\log T(T_\gamma + \log T))$ regret bound for two-target security games.
\end{remark}

\subsection{Finite Stackelberg games}
\label{ssec:finite-games}

Finally, we treat general finite Stackelberg games, which encompass standard linear-utility security games \citep{letchford2009learning,peng2019learning,blum2014}. Lacking the monotonic structure of Section~\ref{sec:clinch}, we take a general approach that gives weaker guarantees for security games but applies quite broadly. When the principal and a $\gamma$-discounting agent have $m$ and $n$ actions, respectively, we apply convex optimization with membership queries to achieve regret $\widetilde{O}\big(T_\gamma  (V^{-1}\sqrt{m} + nm^{2.5}) \log^{4}(T)\big)$, where $V$ is the volume of a ball contained within a certain best response region.

\paragraph{Model and preliminaries.}

Let $(\cX_0,\cY,u_0,v_0)$ be a base Stackelberg game where $\cX_0\!=\!\{ 1, \dots, m \}\!=\![m]$ and $\cY\!=\![n]$ are finite action sets for the principal and agent, respectively, and $u_0,v_0 \in [0,1]^{\cX_0 \times \cY}$ are arbitrary payoff matrices. We consider the mixed strategy game $(\cX,\cY,u,v)$ where the principal commits to a distribution $\bv{x} \in \cX = \Delta_{m-1} = \{ \bv{x} \in \R^m : \bv{x}^\top \mathbf{1}_m = 1, x_i \geq 0 \: \forall i \in [m] \}$, the agent responds with an action $y \in \cY$, and expected payoffs are given by $u(\bv{x},y) = \E_{i \sim \bv{x}}[u_0(i,y)]$ and $v(\bv{x},y) = \E_{i \sim \bv{x}}[v_0(i,y)]$. We define the best response function $\BR$ as well as Stackelberg regret for the corresponding repeated game in the usual way.
Finally, for each action $y \in \cY$, we let $K_y \coloneqq \{\bv{x} \in \cX : y \in \BestResp(\bv{x})\}$ denote the corresponding best response polytope.

\paragraph{Regularity assumptions.}
We require that there exists $y^\star \in \argmax_{y \in \cY : K_y \neq \emptyset} \max_{\bv{x} \in K_y}u(\bv{x},y)$ such that its best response region $K_{y^\star} \subseteq \cX$ contains an $\ell_2$-ball of radius $2r$, where $r > 0$ is known to the principal. This guarantees that the best response polytope associated with $y^\star$ is substantial enough to be found via sampling and is assumed even in the myopic setting \citep{blum2014}. Next, for $y \in \cY$, we define the centered agent utility profile $\bar{\bv{v}}^{(y)} \in \R^m$ by $\bar{v}^{(y)}_i \coloneqq v(i,y) - \frac{1}{m} \sum_{j=1}^m v(j,y)$, and denote the minimum distance between profiles by $\Delta \coloneqq \min_{y \neq y'} \|\bar{\bv{v}}^{(y)} - \bar{\bv{v}}^{(y')}\|_2$. Our regret bound scales logarithmically in $\Delta^{-1}$, but $\Delta$ need not be known to the principal. This minimum distance characterizes the stability of the best response regions to perturbations of the agent's utilities, generalizing the lower slope bound for linear-utility security games in Section~\ref{sec:clinch} (see Remark~\ref{remark:interpreting-delta}). 

\paragraph{Multiple LPs.} 

The principal's optimal utility in the single-round game is $\max_{\bv{x} \in \cX} u(\bv{x},\brr(\bv{x})) = \max_{y \in \cY} \max_{\bv{x} \in K_y} u(\bv{x},y)$. For fixed $y \in \cY$, the objective $u(\bv{x},y)$ is linear in $\bv{x} \in \cX$, so previous works find an optimal strategy by solving the $n$ inner LPs (one for each $y$), originally for known~$v$ \citep{conitzer2006computing} and later for unknown $v$ using best response queries \citep{letchford2009learning}. For the special case of security games,
\cite{blum2014} observe that $\mathds{1}\{y_t = y\} = \mathds{1}\{\bv{x}_t \in K_y\}$ for any query $\bv{x}_t$ and best response $y_t$, up to tie-breaking on the boundary of $K_y$. That is, responses from a myopic agent act as feedback from a \emph{membership oracle} for $K_y$. They then maximize $u(\cdot,y)$ over $K_y$ using linear optimization with membership queries \citep{kalai2006annealing}.

\paragraph{Our extension to non-myopic agents.}
To treat non-myopic agents, we first focus on the setting of $\eps$-approximate best responses (later, this will be implemented via delayed feedback).
We face two main challenges not present in the myopic case. 
First, agent feedback $\mathds{1}\{y_t = y\}$ now only simulates an \emph{approximate} membership oracle for $K_y$. Thus, our optimization approach must be robust to inexact responses near the boundary of $K_y$.
Relatedly, we cannot safely commit to an obtained strategy unless it has sufficient margin within $K_y$. To avoid falsely concluding that $\bv{x}_t \in K_y$, we play several small perturbations of $\bv{x}_t$ and check that the agent always responds with $y$.

Our robust search algorithm \FiniteAlg{} begins by sampling $O(V^{-1} \log T)$ initial points uniformly from $\cX$, where $V$ is the probability that a single sample lies within an $\ell_2$-ball of radius $r$ contained in $K_{y^\star}$ (guaranteed to exist by our assumptions). With high probability, some point lies within this ball. We then play each point $\bv{x}$, observe the agent's response $y$, and test whether $\bv{x}$ lies firmly within $K_y$ using multiple perturbed queries, as above. After filtering out all points for which this test fails and keeping only one point per action, we are left with a substantially smaller set $\{\bv{x}^{(y)}\}_{y \in \cY_0}$ such that $\cY_0 \subseteq \cY$ contains $y^\star$ and that each $\bv{x}^{(y)}$ is well-centered within $K_{y}$.

Next, for each $y \in \cY_0$, we apply a robust convex optimization algorithm of \cite{lee2018efficient} to find a near maximizer $\hat{\bv{x}}^{(y)}$ of $u(\cdot,y)$ over $K_y$ using approximate membership queries. We simulate each call to the membership oracle using repeated perturbed best response queries, as above. The resulting set of maximizers is guaranteed to contain an approximately optimal strategy $\hat{\bv{x}}^{(y^\star)}$.
Finally, using a similar multi-threaded approach to that of \MultiThreadedClinch{}, we translate this search guarantee to a policy \MultiThreadedFiniteAlg{} for learning against $\gamma$-discounting agents with unknown $\gamma$. 
This argument is formalized in \Cref{app:finite-games}.

\begin{theorem}
\label{thm:finite-games}
\MultiThreadedFiniteAlg{} incurs regret at most $\widetilde{O}\big(T_\gamma V^{-1}\sqrt{m} \log^3(T) \log \tfrac{1}{r\Delta} \!+\! T_\gamma nm^{2.5} \log^{4} \bigl(\frac{T}{r}\bigr)\log\frac{1}{\Delta}\big)$ against $\gamma$-discounting agents.
\end{theorem}

\begin{remark}[Dependence on $V^{-1}$]
We note that our bound scales linearly with the inverse volume $V^{-1}$. Although the algorithms of \cite{blum2014} and \cite{peng2019learning} exhibit improved $\log(V^{-1})$ dependence, the former work is specific to security games and the latter appears to heavily rely on exact feedback. Moreover, the guarantees of \cite{peng2019learning} scale with $\binom{m+n}{m}$ in general, reducing to $\poly(m,n)$ only under certain structural assumptions satisfied by security games.
\end{remark}

\begin{remark}[Interpreting $\Delta$ and $r$]
\label{remark:interpreting-delta}
To see how $\Delta$ controls the sensitivity of approximate best response polytopes, fix any $y \in \cY$. The best response polytope $K_y$ can be written as the intersection of the non-negative orthant $\R^m_{\geq 0}$, the affine subspace $A \coloneqq \{ \bv{x} \in \R^m : \bv{x}^\top \mathbf{1}_m = 1\}$, and the half-spaces $H_{y'} \coloneqq \{ \bv{x} \in A : \sum_{i=1}^m x_i v_0(i,y) \geq \sum_{i=1}^m x_i v(i,y') \}$, for each $y' \in \cY \setminus \{y\}$. Indeed, $\cX = \R^m_{\geq 0} \cap A$ and $\cX \cap H_{y'} = \{\bv{x} \in \cX : v(\bv{x},y) \geq v(\bv{x},y') \}$. For any perturbation $\eps \geq 0$, the $\eps$-approximate best response polytope $K_y^\eps$ can be written as the intersection of $\cX$ with the perturbed half-spaces $H_{y'}^\eps = \{ \bv{x} \in A : \sum_{i=1}^m x_i v_0(i,y) \geq \sum_{i=1}^m x_i v(i,y') - \eps \}$. The $\ell_2$-distance between $H_{y'}$ and $H_{y'}^\eps$ is $\eps/\|\bar{\bv{v}}^{(y)} - \bar{\bv{v}}^{(y')}\|_2 \leq \eps/\Delta$. For the special case of linear-utility SSGs induced by the standard combinatorial setup, i.e.\ $v_0(i,y) = \mathds{1}\{i = y\} v_1(y) + \mathds{1}\{i \neq y\} v_2(y)$ with $0 \leq v_1(y) \leq v_2(y) \leq 1$, it holds that $C_\star^{-1} \leq \Delta_\star \leq 2C_\star^{-1}$, where $C_\star = \max_y [v_2(y) - v_1(y)]^{-1}$ is a tight choice of the slope bound $C$ (Section~\ref{sec:clinch}). The minimum radius $r$ plays a similar role to the minimum width $W$ in Section~\ref{sec:clinch}, with $W \geq r$ for these linear-utility SSGs. However, $r$ can be arbitrarily smaller than $W$ in the worst case.
\end{remark}

\begin{remark}[Knowledge of $r$] 
Of the three problem parameters (discount factor $\gamma$, minimum distance $\Delta$, radius $r$), our algorithm only requires knowledge of the radius $r$.
If $r$ is unknown, one may run \MultiThreadedFiniteAlg{} with a guess $\hat{r}$ for $r$. In this case, the algorithm's performance will be competitive compared to the principal strategies which lie in best response regions with radius at least $\hat{r}$ (which may or may not include a Stackelberg equilibrium strategy).
\end{remark}

\subsection{Strategic classification}
\label{ssec:strategic-classification}

Finally, we address the strategic classification environment of \cite{dong2018}, where a learner collects data from agents that may manipulate their features to obtain a desired classification outcome.
In the myopic setting, \cite{dong2018} reduce this problem to bandit convex optimization and obtain regret $O(\sqrt{d}T^{3/4})$ via \emph{gradient descent without a gradient} \citep{flaxman2005}, or \textsc{GDwoG}. Here, we show that this algorithm is inherently robust to the perturbations arising from approximate best responses and apply a simple delay procedure to achieve non-myopic regret $\widetilde{O}(T_\gamma^{1/4} \sqrt{d} T^{3/4})$ for $d$-dimensional features, compared to $O(\sqrt{d} T^{3/4})$ for the original myopic setting.

\paragraph{Model.} In a single round of classification, the agent is described by a tuple $a = (\bv{x},y,\mathsf{d})$, where $\bv{x} \in \R^d$ is their original feature vector, $y \in \{-1,1\}$ is their label, and $\mathsf{d}:\R^d \times \R^d \to \R$ is a distance function describing the cost $\mathsf{d}(\bv{x},\hat{\bv{x}})$ of changing their feature vector from $\bv{x}$ to $\hat{\bv{x}}$. During the round, (i) the principal commits to a linear classifier parameterized by $\bm{\theta} \in \Theta \subset \R^d$, (ii) the agent responds with a manipulated feature vector $\hat{\bv{x}} \in \R^d$, and (iii) the label $y$ is revealed to the principal.
Following \citep{dong2018}, agent payoff is given by $v_a(\bm{\theta},\hat{\bv{x}}) = \bm{\theta}^\top \hat{\bv{x}} - \mathsf{d}(\hat{\bv{x}},\bv{x})$ when $y = -1$, and, when $y=1$, the agent is assumed to be non-strategic with $v_a(\bm{\theta},\hat{\bv{x}}) = -\infty$ for $\hat{\bv{x}} \neq \bv{x}$.
Hence, we write $\BestResp(\bm{\theta}) = \argmax_{\bv{x}'} v(\bm{\theta},\bv{x}')$ when $y = -1$ and $\BestResp(\bm{\theta}) = \{ \bv{x} \}$ otherwise.
The principal's payoff is given by $-\ell(\bm{\theta},\hat{\bv{x}},y)$, where $\ell$ is either logistic loss $\ell_\mathrm{log}(\bm{\theta},\hat{\bv{x}},y) = \log(1 + e^{-y \bm{\theta}^\top \hat{\bv{x}}})$, corresponding to logistic regression, or hinge loss $\ell_\mathrm{h}(\bm{\theta},\hat{\bm{x}},y) = \max\{0,1-y\bm{\theta}^\top\hat{\bm{x}} \rangle\}$, for a support vector machine.

We consider a repeated game determined by the sequence of types $\{a_t = (\bv{x}_t,y_t,\mathsf{d}_t)\}_{t=1}^T$ for a returning agent. Define $\BestResp_t$ and $\BestResp^\eps_t$ in the usual way and let $\brr_t(\theta)$ denote a representative from $\BestResp_t(\theta)$, breaking ties in favor of the principal. Denoting the game's history by $\{(\bm{\theta}_t,\hat{\bv{x}}_t)\}_{t=1}^T$, the agent seeks to maximize their expected $\gamma$-discounted utility $\E\big[\sum_{t=1}^T \gamma^t v_{a_t}(\theta_t,\hat{\bv{x}}_t)\big]$, while the principal seeks to minimize Stackelberg regret $\E\big[\sum_{t=1}^T \ell(\bm{\theta}_t, \hat{\bv{x}}_t, y_t)\big] - \min_{\bm{\theta} \in \Theta} \E\big[\sum_{t=1}^T\ell(\bm{\theta},\brr_t(\bm{\theta}),y_t)\big]$. 

\begin{remark}[Choice of model]
At a high level, strategic classification arises in many real-world settings where individuals derive utility from their classification outcomes and can manipulate their features at some cost (e.g., spam classifiers, tax reporting, and college admissions). There are a variety of existing models for strategic classification \citep{hardt2016strategic, dong2018, chen2020learning} and, relatedly, recommendation with strategic users \citep{haupt2023recommending}.Users in these settings can be long-lived and aware of the future impact of their actions (see, e.g., the user testimonials in \citealp{haupt2023recommending}), motivating extensions to the non-myopic case. To illustrate how our framework applies to strategic classification and to address the distinct paradigm of bandit convex optimization, we employ the  model of \cite{dong2018}.
\end{remark}

\paragraph{Regularity assumptions.}
We mirror requirements of \cite{dong2018}, assuming that (i) $\Theta$ is convex,
and contains the unit $\ell_2$-ball $\mathbb{B}$, (ii) each feature vector $\bv{x}_t$ and classifier $\bm{\theta} \in \Theta$ lie within $R\mathbb{B}$, and (iii) each distance function is of the form $d_t(\hat{\bv{x}},\bv{x}) = f_t(\hat{\bv{x}}- \bv{x})$, where $f_t$ is $\alpha$-strongly convex and positive homogeneous of degree 2. (Strong, rather than standard, convexity is an extra assumption used for non-myopic learning.)\smallskip

Under these assumptions, we prove in \cref{app:classification} that, with best response feedback, each loss function $\ell_t(\bm{\theta}) = \ell(\bm{\theta},\brr_t(\bm{\theta}),y_t)$ is convex, Lipschitz, and bounded. Moreover, with $\eps$-approximate best responses, this still holds up to additive error $O(\sqrt{\eps})$. Thus, we can successfully apply techniques from (robust) bandit convex optimization. In particular, we note that \textsc{GDwoG} is automatically robust to such corruptions. To obtain guarantees against non-myopic agents, we cycle through multiple copies of this algorithm to achieve the needed feedback delay. The resulting policy, Cycled Gradient Descent without a Gradient (\textsc{CGDwoG}, see \cref{app:classification}), achieves the following.

\begin{theorem}\label{thm:strategic-classification}
\textsc{CGDwoG} achieves regret $\widetilde{O}\big(T_\gamma^{1/4}\sqrt{d}T^{3/4} + d^2\big)$ for strategic classification against $\gamma$-discounting agents when $R,\alpha^{-1} = \mathrm{polylog}(T,d)$.
\end{theorem}

%% file: 6-conclusion.tex
\section{Conclusion}

In this work, we developed a framework for learning in Stackelberg games with non-myopic agents, reducing this problem to the design of minimally reactive and robust bandit algorithms. For each application (SSGs, demand learning, finite games, strategic classification), we identified a robust and delayed learning policy with low Stackelberg regret against $\gamma$-discounting agents. Our work opens up several interesting avenues of future research, three of which we now highlight.

Our first question regards the optimality and generality of our reduction framework. While we primarily considered the feedback-delay information screen to inspire new algorithm design principles, other forms of information screens may be also effective for learning in the presence of non-myopic agents. Examples of such screens may include using random delays, delays that involve non-monotone release of information, or differentially private information screens. Are these information screens inherently different in the power they provide algorithms? In particular, are there settings for which algorithms that conform to one of these channels outperform those conforming to others? Are there universally optimal information screens that achieve optimal regret for a wide range of discount factors across all principal-agent games?

Second, we ask whether a regret overhead of $T_\gamma \log T$ can be avoided when the discount factor $\gamma$ is unknown, and whether algorithms which do not maintain confidence sets can be adapted to this setting. While our algorithms achieve an additive dependence on $T_\gamma$ for SSGs and demand learning, our extensions to their agnostic analogues have a multiplicative dependence arising from our multi-threading approach (see, e.g., \Cref{thm:multi-threaded-clinch}). Is this gap necessary to derive $\gamma$-agnostic algorithms for these settings, or can one achieve an additive dependence with different methods? Moreover, the multi-threading approach to $\gamma$-agnostic learning relies heavily on algorithms for known $\gamma$ that maintain shrinking confidence sets (see Remark 5). Is there a technique that applies to algorithms beyond this class? For adversarial corruptions, this is often achieved via gradient-descent-based algorithms \citep{ZimmertSeldin21,ChenWang22,KrishnamurthyLP21}.

Our last question concerns non-myopic agent learning beyond repeated game settings. Specifically, many repeated interactions with agents do not fit the framework of a repeated game, e.g., our demand learning setting when the principal (the seller) has a limited inventory of items that they may distribute over the $T$ rounds \citep{besbes2009dynamic}. Such learning settings can be modeled as bandits with knapsacks \citep{badanidiyuru2018}, in which exploration time is not the only limited resource. More broadly, learning in the context of state and repeated interaction is of interest in reinforcement learning. In such settings---in which state is carried over between interactions, what principles apply to the design of effective learning algorithm?

%% file: A1-framework.tex
\section{Supplementary Material for our Reduction (Section~\texorpdfstring{\ref{sec:framework}}{2})}\label{app:framework}

\subsection{Reduction to robust learning with delays (proof of Proposition~\texorpdfstring{\ref{prop:delayed-feedback}}{2.1})}
\label{app:robust-learning-with-delays}

\begin{proof}[Proof of \Cref{prop:delayed-feedback}.]
Given a $D$-delayed policy $\cA$, we show by contradiction that the policy $\cB$ of a $\gamma$-discounting agent satisfies $\mathcal{B}(H_{t - 1}, x_t)\in\BestResp[\tau](x_t)$
for any pair $(H_{t - 1}, x_t)$ that occurs with positive probability, where $\tau = \frac{1}{1-\gamma}\gamma^{D}$. Our choice of $D = \lceil T_\gamma \log(T_\gamma/\eps) \rceil$ ensures that $\gamma^D \leq \eps/T_\gamma$, and so $\tau = T_\gamma \gamma^D \leq \eps$, as desired. If $\mathcal{B}(H_{t-1}, x_t)\not\in\BestResp[\tau](x_t)$ for a pair $(H_{t-1}, x_t)$ that occurs with positive probability, we can
construct a modified agent policy $\mathcal{B}'$ with strictly higher expected payoff. Define $\mathcal{B}'$ so that $\cB'(H_{t - 1}, x_t)\in \BestResp(x_t)$ and $\mathcal{B}'(H', x') = \cB(H', x')$ for all other pairs $(H', x')$. Conditioned on history $H_{t - 1}$, a $D$-delayed policy $\mathcal{A}$ plays the same sequence of actions $x_{t+1}, \dots x_{t+D-1}$ under both $\mathcal{B}$ and $\mathcal{B'}$. Therefore, conditioned on playing history $H_{t - 1}$ and observing action $x_t$, the agent loses at most $\sum_{s=D+t}^\infty \gamma^s = \gamma^{D+t} / (1 - \gamma)$ in discounted future payoff by switching to $\cB'$ because the principal's policy is $D$-delayed. Moreover, the agent gains more than $\gamma^t\tau = \gamma^{D+t} / (1 - \gamma)$ payoff at time $t$. Thus, switching from $\cB$ to $\cB'$ yields a strictly positive gain in expectation for the agent.
\end{proof}

\subsection{A batch-delay equivalence (proof of Proposition~\texorpdfstring{\ref{prop:batch-delay-reduction}}{2.2})}\label{app:batch-delay}

Before proving the result, we define a general framework for bandit problems that generalizes the Stackelberg setting.
We consider bandit problems over $T$ rounds. An \emph{abstract bandit problem} is defined by a tuple $(\mathcal{X}, \mathcal{Y}, r)$, where $\mathcal{X}$ is the principal's action set, $\mathcal{Y}$ describes possible unknown states, and $r$ is a regret function $r\colon\mathcal{H}\to\mathbb{R}$ mapping the set $\mathcal{H}\coloneqq\bigcup_{t\ge 0} (\mathcal{X}\times\mathcal{Y})^t$ of histories to regret values. We assume $r$ is \emph{subadditive}: if a history $H\in\mathcal{H}$ is partitioned into two complementary subsequences $H', H''\in\mathcal{H}$, then $r(H)\le r(H') + r(H'')$. Subadditivity is satisfied by common notions of regret: in stochastic settings, regret is simply the sum of regrets over individual rounds; in adversarial settings, regret is subadditive. We further distinguish a subset $\mathcal{H}^*\subseteq\mathcal{H}$ of \emph{feasible} histories, and assume that any subsequence of a feasible history is also feasible. In our setting, feasible histories correspond to restrictions on the agent's behavior, e.g., the agent plays an $\varepsilon$-approximate best response at each round. The principal's policy is a map $\mathcal{A}\colon\mathcal{H}\to\mathcal{X}$. During the $t$-th round, the principal plays action $x_t = \mathcal{A}(H_{t-1})$ (where $H_{t-1}$ is the history up to the start of round $t$) and observes $y_t$ (which may be chosen randomly and adaptively based on $x_t$). We say that $\mathcal{A}$ satisfies the regret bound $R_{\mathcal{A}}(T)$ if, for each history $H\in\mathcal{H}^*$ of length $T$ such that $x_t = \mathcal{A}(H_{t-1})$, then $r(H)\le\Regret_{\mathcal{A}}(T)$.

Given any abstract bandit problem $(\mathcal{X}, \mathcal{Y}, r)$, we define learning with delayed feedback and batched queries as follows. As before, $\mathcal{A}$ is \emph{$D$-delayed} if $\mathcal{A}(H_{t-1})$ depends only on the prefix $H_{t-D}$, and $\mathcal{A}$ is \emph{$B$-batched} if $\mathcal{A}(H_{t-1})$ depends only on the prefix $H_{B\lfloor (t - 1)/B\rfloor}$. 

To cast our principal-agent learning setting as an abstract bandit problem, we let $\mathcal{X}$ be the set of principal actions, $\mathcal{Y}$ be the set of agent actions, and regret be Stackelberg regret. Note that Stackelberg regret \eqref{eq:stackelberg-regret} is subadditive because $\max_{x} (f(x) + g(x)) \le \max_{x}(f(x)) + \max_{x}(g(x))$. Finally, we take the set of feasible histories to be those where the agent policy belongs to a class~$\mathfrak{B}$.

We now present a proof of the batch-delay equivalence (\Cref{prop:batch-delay-reduction}), which relies on the following lemma. It states that a $1$-delayed policy $\mathcal{A}$ can be converted into a $2$-delayed policy by instantiating two independent copies of $\mathcal{A}$ and following them on alternating rounds.

\begin{lemma}\label{lemma:1-delay}
  Let $\cA$ be a policy with regret bound $R_{\mathcal{A}}$. Consider the policy $\mathcal{A}'$ that instantiates two independent copies $\mathcal{A}_{0}$ and $\mathcal{A}_{1}$ of $\mathcal{A}$, and on round $t$ plays $x_t = \mathcal{A}_{r}(((x_{t'},y_{t'}))_{t'\equiv r\pmod 2, t' < t})$, where $t \equiv r \pmod 2$. Then $\mathcal{A}'$ is $2$-delayed and satisfies a regret bound of $R_{\mathcal{A}'}(T)\le 2R_{\mathcal{A}}(T)$.
\end{lemma}
\begin{proof}
  To bound regret, note that $\mathcal{A}_0$ is run on the history $((x_{t'}, y_{t'}))_{t'\equiv 0\pmod 2, t'\le T}$, which is by definition feasible. Thus, it incurs a total of $R_{\mathcal{A}}(\lceil T / 2\rceil)$ regret on this subsequence of the actual history. Likewise, $\mathcal{A}_1$ incurs at most $R_{\mathcal{A}}(\lceil T / 2\rceil)$ regret as well. Therefore, by the subadditivity axiom, the total regret is at most $2R_{\mathcal{A}}(\lceil T / 2\rceil)\le 2R_{\mathcal{A}}(T)$, since regret is monotonic in the delay length.
  The lemma now follows, since by definition, $\mathcal{A}'$ is $2$-delayed.
\end{proof}

\begin{proof}[Proof of \Cref{prop:batch-delay-reduction}.]
The first claim follows from the definition of a batched algorithm, since $t-D\le D\lfloor (t - 1) / D)\rfloor$. The second claim follows from an application of \Cref{lemma:1-delay} to the ``$B$-batched'' $(\mathcal{X}^B, \mathcal{Y}^B, r^B)$ bandit problem, where $\cX^B$ and $\cY^B$ be are the $B$-fold products of $\cX$ and $\cY$, respectively, and $r^B$ is given by evaluating $r$ on the history given by the concatenations of actions $((x_1,\ldots,x_B), (y_1,\ldots,y_B))$. That is, in this new bandit problem, the principal simply chooses $B$-tuples of actions and receives feedback on these $B$-tuples at once, with regret measured according to the original $(\cX, \cY, r)$ bandit problem. This new problem is by definition equivalent to our $B$-batched bandit problem defined above. By \Cref{lemma:1-delay}, an algorithm $\mathcal{A}$ for this equivalent problem can be converted into an algorithm $\mathcal{A'}$ for the $2$-delayed version of this problem achieving $2R_{\mathcal{A'}}(T)$ regret. Forgetting about the batch structure, we see that this algorithm $\mathcal{A'}$ is $B$ delayed, since any batch starting at time $t = kB + 1$ depends on the history $H_{<k(B-1) + 1}$. Hence the second claim is proven.
\end{proof}

%% file: A2-clinch.tex
\section{Supplementary Material for Theoretical Results on SSGs (Sections~\texorpdfstring{\ref{sec:clinch}-\ref{sec:non-myopic-clinch}}{3-4})}\label{app:clinch}

\subsection{Simplifying design \& analysis of \texorpdfstring{\Clinch{}}{Clinch} when \texorpdfstring{$\cX = \Delta_{n-1}$}{X=Δn-1} (Remark~\texorpdfstring{\ref{rem:clinch-simplex}}{3.1})}
\label{app:security-simplex}
When $\cX = \Delta^\leq_{n-1}\coloneqq\{x : \|\bv{x}\|_1 \leq 1 \land x_y\ge 0\:\forall y\}$, we may as well restrict to $\cX = \Delta_{n-1}$. Indeed, since each agent utility function $v^y$ is continuous and strictly decreasing, we can increase coverage probabilities of any $\bv{x} \in \cX \setminus \Delta_{n-1}$ to obtain $\bv{x}' \in \Delta_{n-1}$ with $\BestResp(\bv{x}') = \BestResp(\bv{x})$ and $v(\bv{x}',\brr(\bv{x}')) < v(\bv{x}',\brr(\bv{x}'))$. This argument also implies that the optimal stable strategy $\bv{x}^\star$ guaranteed by \Cref{prop:ssg-solution} belongs to $\Delta_{n-1}$. From now on, we thus fix $\cX = \Delta_{n-1}$.

For this setting, we present a simplified (simplexified) algorithm \ClinchSimplex{} that achieves the same query complexity as \Cref{thm:clinch} but admits a simpler analysis. Similarly to \Clinch{}, \ClinchSimplex{} maintains an (approximate) entry-wise lower bound $\underline{\bv{x}}$ for $\bv{x}^\star$ initialized to the 0 vector. This time, however, we envision the remaining mass $1 - \|\underline{\bv{x}}\|_1$ as a potential that is decreased with each step. To ensure a significant reduction, we query $\bv{x}$ which distributes this remaining mass evenly across the coordinates of $\underline{\bv{x}}$ and update $\underline{x}_y \gets x_y$ for the attacked target $y \in \cY$. Finally, we normalize $\underline{\bv{x}}$ so that it lies on the simplex and apply the same perturbation used by \Clinch{}.

\begin{proposition}
\label{prop:clinch-simplex}
Fix $0 < \lambda \leq 1$.
Then \ClinchSimplex{} returns a $\lambda$-approximate equilibrium strategy using $O(n\log\frac{C}{W\lambda})$ queries to an $\eps$-approximate best response oracle with $\eps \leq \frac{W\lambda}{12C^3n}$.
\end{proposition}

\begin{algorithm}[H]
    \caption{\ClinchSimplex{}: a robust algorithm for learning SSGs when $\cX = \Delta_n$}\label{alg:ssg-search}
    \DontPrintSemicolon
    \SetAlgoNoLine
    \SetKwInOut{Input}{input}
    \SetKwInOut{Output}{output}
    \Input{target accuracy $\lambda \in (0,1]$, best response oracle $\Oracle$ with $\Oracle(x)\in\BestResp^\eps(x)$}
    \Output{$\lambda$-approximate equilibrium strategy}
    $\underline{\bv{x}} \gets (0, 0, \ldots, 0)\in\R^n, \delta \gets \frac{W\lambda}{6C^2}$\;
    \For{$i=1, 2, \ldots, \lceil n \ln\frac 4\delta \rceil$}{
        $\bv{x} \gets \underline{\bv{x}} + (1, 1, \ldots, 1)\cdot\frac{1}{n}(1 - \|\underline{\bv{x}}\|_1)$\;
        $y \gets \Oracle(\bv{x})$\;
        $\underline{x}_y \gets x_y$\;
    }
    $\hat{\bv{x}} \gets \underline{\bv{x}}/\|\underline{\bv{x}}\|_1$\;
    $\hat{y} \gets \argmax_{y \in \cY:\hat x_y > W/2} u(\hat x, y)$\;
    \Return $\hat{\bv{x}} - \frac{W\lambda}{2} \mathbf{\bv{x}}_{\hat{y}}$\;
\end{algorithm}

\begin{proof}
First, we note that $\underline{x}_y \leq x^\star_y + C\eps$ for each $y \in \cY$, by the same argument applied in the proof of \Cref{thm:clinch}. Next, we analyze convergence. Notice that the quantity $1 - \|\underline{\bv{x}}\|_1$ decreases by a factor of $1 - \frac{1}{n}$ after each iteration. Since $1 - \|\underline{\bv{x}}\|_1$ is initially $1$, after $\lceil n\ln\frac 4\delta\rceil$ iterations, it holds that
\[ 1 - \|\underline{\bv{x}}\|_1 \le \left(1 - \frac{1}{n}\right)^{n\ln\frac 4\delta}\le\frac{\delta}{4}. \]
We may thus conclude, for $\hat{\bv{x}}\coloneqq \underline{\bv{x}} / \|\underline{\bv{x}}\|_1$, that
\begin{equation*}
    \|\hat{\bv{x}} - \bv{x}^\star\|_\infty \leq \|\hat{\bv{x}} - \bv{x}^\star\|_1 \leq \left(\|\underline{\bv{x}}\|_1^{-1} - 1\right)\|\underline{\bv{x}}\|_1 + \|\underline{\bv{x}} - \bv{x}^\star\|_1 \leq  1 - \|\underline{\bv{x}}\|_1 + \|\underline{\bv{x}} - \bv{x}^\star\|_1,
\end{equation*}
by the triangle inequality. By the entry-wise lower bound property of $\underline{x}$, we further note that
\begin{equation*}
\|\underline{\bv{x}} - \bv{x}^\star\|_1 \leq 1 - \|\underline{\bv{x}}\|_1 + Cn\eps \leq 1 - \|\underline{\bv{x}}\|_1 + \frac{\delta}{2},
\end{equation*}
and so $\|\hat{\bv{x}}-\bv{x}^\star\|_\infty \leq \delta$. Finally, \Cref{lem:perturbation} gives that the returned point is a $\lambda$-approximate Stackelberg equilibrium strategy, as desired.
\end{proof}

\subsection{Minimizing agent best response utility (proof of Lemma~\texorpdfstring{\ref{lem:minimization}}{3.7})}
\label{prf:minimization}

To show that \Clinch{} makes continual progress, we require an approximate version of Gr\"unbaum's inequality \citep{grunbaum1960partitions}. First, we state the classic result.

\begin{lemma}[Theorem 2 in \cite{grunbaum1960partitions}]
\label{lem:grunbaums-inequality}
If $K$ is a non-empty compact convex set in $\R^d$, then for any halfspace $H$ containing its centroid $\bv{x} = \E_{\bv{z} \sim \Unif(K)}[\bv{z}]$, we have $\vol_d(H \cap K) \geq \frac{1}{e}\vol_d(K)$.
\end{lemma}

For our approximate case, we use the following, which implies that each update to $\underline{\bv{x}}$ will sufficiently shrink the volume of the active search region $S$ (so long as no target is removed from~$\cR$).

\begin{lemma}
\label{lem:volume-blowup}
Let $K \subseteq [0,1]^d$ be convex and downward closed with centroid $\bv{x} = \E_{\bv{z} \sim \Unif(K)}[\bv{z}]$, and write $\alpha = \sup_{\bv{z} \in K}z_1$. Then for $\beta \geq 0$, we have $\vol_d(\{\bv{z} \in K : z_1 \geq x_1 - \beta\}) < (1 + e \beta/\alpha)^d(1 - 1/e) \vol_d(K)$.
\end{lemma}

\begin{proof}[Proof of \Cref{lem:volume-blowup}.]
For convenience, assume that $K$ is closed; this does not affect the volumes. Similarly replace $\beta$ with $\min\{\beta,x_1\}$ so that $x_1-\beta \geq 0$. Proving the result with this update implies the original result.

Now define the halfspace $H \coloneqq \{ \bv{z} \in \R^d : z_1 \geq x_1 - \beta \}$. First, we observe that the related halfspace $H_0 = \{ \bv{z} \in \R^d : z_1 \geq x_1 \}$ satisfies $\frac{1}{e}\vol_d(K) \leq \vol_d(H_0 \cap K) < \left(1 - \frac{1}{e}\right)\vol_d(K)$ by Gr\"unbaum's inequality (\Cref{lem:grunbaums-inequality}). By downward closure, we have $\alpha\mathbf{e}_1 \in K$ and deduce that $\vol_d(H_0 \cap K) \leq \left(1 - \frac{x_1}{\alpha}\right)\vol_d(K)$. This requires that $x_1 \leq \bigl(1-\frac{1}{e}\bigr)\alpha$ to avoid violating the first inequality. 

Next, consider the $(d-1)$-dimensional intersection of $K$ with the hyperplane defining $H_0$, denoted by $L_0 \coloneqq \{ \bv{z} \in K : z_1 = x_1 \}$ for $H_0$. Convexity requires that $H_0 \cap K$ contain the convex hull of $L_0$ and $\alpha\mathbf{e}_1$, a cone we denote by $A$ with $\vol_d(A) = \frac{\alpha-x_1}{d!}\vol_{d-1}(L_0)$. Moreover, every point in $H \cap K \setminus H_0$ is outside of $A$ and connected to $\alpha\mathbf{e}_1$ by a line segment contained in $H$ and passing through $L_0$. Thus, $H \cap K \setminus H_0$ is disjoint from the cone $A$ but contained by the cone $B$ obtained by intersecting $H$ with the union of all rays emitted from $\alpha\mathbf{e}_1$ and passing through $L_0$. Similarly to $A$, we compute the volume of $B$ to be $\frac{\alpha-x_1+\beta}{d!}\bigl(\frac{\alpha-x_1 + \beta}{\alpha-x_1}\bigr)^{d-1}\vol_{d-1}(L_0)$. Consequently, we have
\begin{align*}
    \vol_d(H \cap K \setminus H_0) &\leq \vol_d(B) - \vol_d(A)\\
    &= \left[(\alpha-x_1+\beta)\left(\frac{\alpha-x_1 + \beta}{\alpha-x_1}\right)^{d-1} - (\alpha-x_1)\right]\frac{1}{d!}\vol_{d-1}(L_0)\\
    &\leq \left[\left(\frac{\alpha-x_1 + \beta}{\alpha-x_1}\right)^{d} -1 \right]\vol_{d}(H_0 \cap K)\\
    &= \left[\left(1 + \frac{\beta}{\alpha-x_1}\right)^{d} -1 \right]\vol_{d}(H_0 \cap K)\\
    &\leq \left[\left(1 + \frac{e\beta}{\alpha}\right)^{d} -1 \right]\vol_{d}(H_0 \cap K).
\end{align*}
Finally, we can bound
\begin{align*}
    \vol_d(H \cap K) &= \vol_d(H_0 \cap K) + \vol_d(H \cap K \setminus H_0)\\
    &\leq \left(1 + \frac{e\beta}{\alpha}\right)^{d} \vol_{d}(H_0 \cap K)\\
    &< \left(1 + \frac{e\beta}{\alpha}\right)^{d}\left(1 - \frac{1}{e}\right)\vol_d(K),
\end{align*}
as desired.
\end{proof}

We now prove the guarantee for the primary stage of \Clinch{}.

\begin{proof}[Proof of \Cref{lem:minimization}.]
First, we observe that $\underline{\bv{x}}$ always approximately lower bounds $\bv{x}^\star$ in each entry. Note that whenever $\underline x_y$ gets updated, we set $\underline x_y = x_y - C\eps$ for some $\bv{x}$ such that $y\in\BestResp^\eps(x)$. By monotonicity of $v^y$ and our slope bound, this implies $x^\star_y \geq x_y - C\eps = \underline{x}_y$, as desired.
  
Next, we will show that the termination condition at Step~\ref{step:clinch-terminate} is satisfied after at most $O(n\log\frac{n}{\delta})$ rounds, recalling that $0 < \delta \leq 1$ is our desired accuracy. To start, we establish a bit of notation to keep track of variables between iterations. For each round $i = 1,2,\dots$ before termination, we write $\cR_i$ for the remaining targets and $S_i$ for the active search region after Step~\ref{step:clinch-search-region}, $\bv{x}^{(i)}$ for the queried point at Step~\ref{step:clinch-oracle-query}, $y_i$ for the oracle response at Step~\ref{step:clinch-oracle-response}, and $\underline{\bv{x}}^{(i)}$ for the value of $\underline{\bv{x}}$ after Step~\ref{step:clinch-lower-bound-update}. Set $n_i = \dim(S_i)$, defined as the minimum dimension of the subspace spanned by $S_i - \bv{w}$ over some $\bv{w} \in S_i$. Finally, write $\lambda = \frac{\delta}{4C^2}$ for the threshold used to flatten $S$. Now, we fix ourselves at some round $i$ and consider two cases.
  
\paragraph{Case 1: $\cR_{i+1}=\cR_i$.}
In this case, no targets are removed from $\cR_i$ and $n_{i+1} = n_i$. Since we selected $\bv{x}^{(i)}$ as the centroid of $S_i$, we can apply \Cref{lem:volume-blowup}. Indeed, $S_i$ is convex, and its translation $K = S_i - \underline{\bv{x}}^{(i-1)}$ is downward closed. Moreover, $\sup_{\bv{z} \in K} z_{y_i} = \sup_{\bv{z} \in S_i} z_{y_i} - \underline{x}^{(i-1)}_{y_i} \geq \lambda$ (otherwise, the target $y_i$ would have been removed from $\cR_i$ to obtain $\cR_{i+1}$). Consequently, we have
\begin{align*}
     \vol_{n_{i+1}}(S_{i+1}) =  \vol_{n_{i}}(S_{i+1}) &= \vol_{n_i}\left(\left\{ \bv{z} \in S_i : z_{y_i} \geq x^{(i)}_{y_i} - C\eps\right\}\right)\\
     &\leq \left(1 + \frac{e \cdot C \eps}{\lambda}\right)^n \left(1-\frac{1}{e}\right) \vol_d(S_i)\\
     &\leq e^{1/3} \left(1-\frac{1}{e}\right) \vol_d(S_i) < \frac{9}{10}\vol_d(S_i),
\end{align*}
where the penultimate inequality uses that $\eps \leq \frac{\lambda}{3Cen} = \frac{\delta}{12C^3 en}$.

\paragraph{Case 2: $\cR_{i+1} \subset \cR_i$.} 
In this case, $n_{i+1} - n_i > 0$ targets are removed from $\cR_i$ in the next step. Writing $K_i = \{ \bv{x}' \in S_i : x_z' = \underline{x}^{(i-1)}_z \: \forall z \not\in R_{i+1}\}$ for the region which enforces the locked coordinates for the next step --- but not the updated lower envelope --- we (loosely) bound
\begin{equation*}
    \vol_{n_{i+1}}(S_{i+1}) \leq \vol_{n_{i+1}}(K_i) \leq \left(\frac{n}{\lambda}\right)^{n_{i}-n_{i+1}} \vol_{n_i}(S_{i}).
\end{equation*}
The first inequality uses that $S_{i+1} \subseteq K_i$. For the second, convexity requires that $S_i$ contains the convex hull of $K_i$ and the points $\{ \underline{x}^{(i-1)} + \lambda \mathbf{e}_z : z \in \cR_i \setminus \cR_{i+1}\}$, which has volume loosely bounded from below by $(\lambda/n)^{|n_{i+1}-n_i|}\vol_{n_{i+1}}(K_i)$. 

Combining these cases inductively, we deduce that
\begin{equation}
\label{eq:clinch-volume-reduction}
    \vol_{n_{i}}(S_{i}) < \left(\frac{n}{\lambda}\right)^n \left(\frac{9}{10}\right)^{i-n}\alpha^n.
\end{equation}
On the other hand, once $\vol_{n_{i}}(S_{i}) < \lambda^n / n!$, every coordinate must have slack less than $\lambda$, and the termination condition at Step~\ref{step:clinch-terminate} will be satisfied. Consequently, we compute that the outer loop must terminate after at most $15n\log\frac{2\alpha n}{\lambda} \leq 15n\log\frac{8C^2 \alpha n}{\delta}$ iterations. At this point, we have $y \in \BestResp^\eps(\bv{x})$ for $\bv{x} \in \cX$ with $\bv{x} \geq \underline{\bv{x}}$ and either $\underline{\bv{x}} + \lambda \mathbf{e}_y \not\in \cX$ or $\underline{x}_y + \lambda > \overline{x}_y$. In the former case, downward closure of $\cX$ implies $x^\star_y \leq \underline{x}_y + \lambda \leq x_y + \lambda$, and the same relations hold for the latter, since $x^\star_{y} \leq \overline{x}_y$ at this point from the input guarantee. Hence, we obtain
\begin{align*}
    v(\bv{x}^\star, \brr(\bv{x}_\star)) \geq v^y(\bv{x}^\star_y) &\geq v^y(x_y + \lambda)\\
    &\geq v^y(x_y) - C \lambda\\
    &\geq v(\bv{x},\brr(\bv{x})) - 2 C\lambda\\
    &= v(\bv{x},\brr(\bv{x})) - \frac{\delta}{2C}.\qedhere
\end{align*}
\end{proof}

\subsection{Mass conservation (proof of Lemma~\texorpdfstring{\ref{lem:stabilization}}{3.8})}
\label{app:stabilize}
\begin{proof}[Proof of \Cref{lem:stabilization}.]
Fixing $y \in \cY$, define the thresholds
\begin{align*}
    r_y &= \sup\left\{p \in [\underline{x}_y,x_y] : \BestResp^{\lambda/C}(\bv{x} + [p-x_y]\mathbf{e}_y) = \{y\}\right\},\\
    s_y &= \sup\left\{p \in [\underline{x}_y,x_y] : \BestResp(\bv{x} + [p-x_y]\mathbf{e}_y) = \{y\} \right\},\\
    t_y &=  \sup\left\{p \in [\underline{x}_y,x_y] : y \in \BestResp^{\lambda/C}(\bv{x} + [p-x_y]\mathbf{e}_y) \right\},
\end{align*}
where we define each to be $\underline{x}_y$ if the corresponding set is empty. (Note that the set for $s_y$ can contain at most one point by strict monotonicity of $v^y$.) By monotonicity of $v^y$ and our slope bound, we have $s_y - \lambda \leq r_y \leq s_y \leq t_y \leq s_y + \lambda$. By our choice of binary search, either $\hat{x}_y > x_y - \lambda$, or $\hat{x}_y > m - \lambda$ at some iteration for which $\Oracle{}(\bv{x} - [x_y - m]\mathbf{e}_y) \neq y$. In the latter case, monotonicity of $v^y$ and the slope bound require that $\hat{x}_y \geq r_y$, while, for the former, we have $\hat{x}_y > r_y - \lambda$.
Similarly, either $\hat{x}_y = \underline{x}_y$, or $\hat{x}_y < m$ for a search iteration during which $\Oracle{}(\bv{x} - [x_y - m]\mathbf{e}_y) = y$. In the latter case, monotonicity of $v^y$ and the slope bound require that $\hat{x}_y < t_y$, while, for the former, we have $\hat{x}_y \leq t_y$. Combining, we have that $\hat{x}_y \in (s_y - 2\lambda, s_y + \lambda]$.

By definition of $s_y$, we must have $v^y(s_y) \leq v(\bv{x},\brr(\bv{x}))$ (with equality unless $s_y=\underline{x}_y$). Hence, $v^y(\hat{x}_y) < v^y(s_y) + 2C\lambda \leq v(\bv{x},\brr(\bv{x})) + 2C\lambda$. Since this holds for all $y \in \cY$, we have $v(\hat{\bv{x}},\brr(\hat{\bv{x}})) \leq v(\bv{x},\brr(\bv{x})) + 2C\lambda$, proving the second part of the claim. Now, if $\hat{x}_y > \underline{x}_y$, we must have $t_y > \hat{x}_y > \underline{x}_y$, and so $y \in \BestResp^{\lambda/C}(\bv{x} + [\hat{x}_y - x_y]\mathbf{e}_y)$. The previous result then implies that $y \in \BestResp^{\lambda/C + 2C\lambda}(\bv{x} + [\hat{x}_y - x_y]\mathbf{e}_y)$. We bound $\lambda/C + 2C\delta \leq 3C\lambda$ for conciseness, and note that each binary search use $O\big(\log \frac{\alpha}{\lambda}\big)$ queries.
\end{proof}

\subsection{Perturbing estimates of \texorpdfstring{$\bv{x}^\star$}{x*} (proof of \texorpdfstring{Lemma~\ref{lem:perturbation}}{4.1})}
\label{prf:perturbation}

\begin{proof}[Proof of \Cref{lem:perturbation}.]
The error bound implies that $\hat{x}_y > W/2$ only if $x^\star_y > 0$. On the other hand, if $x^\star_y > 0$, then our regularity width assumption requires that $x^\star_y \geq W$, and so $\hat{x}_y \geq W - \frac{W\lambda}{6C^2} > W/2$. As noted in the proof of \Cref{prop:ssg-solution}, there are no $y \in \BestResp(\bv{x}^\star)$ with $x^\star_y = 0$, and so
\begin{equation*}
    \BestResp(\bv{x}^\star) = \{ y \in \cY : x^\star_y > 0 \} = \{ y \in \cY : \hat{x}_y > W/2 \}.
\end{equation*}

Now fix $\hat{y}$ as defined in Step 8, and consider the returned strategy $\tilde{\bv{x}} \coloneqq \hat{\bv{x}} - \frac{W\lambda}{2} \mathbf{e}_{\hat{y}}$. 
We know that $\tilde{\bv{x}} \in \cX$ since $\hat{x}_y > W/2$ and $\cX$ is downward closed.
We claim that $\BestResp(\tilde{\bv{x}}) = \{\hat{y}\}$.  Indeed, we have $\tilde{x}_{\hat{y}} < x^\star_{\hat{y}} - (\frac{W\lambda}{2} - \frac{W\lambda}{6C^2}) \leq x^\star_{\hat{y}} - W\lambda/3$, and so our lower slope bound requires that
\begin{equation*}
    v^{\hat{y}}(\tilde{x}_{\hat{y}}) > v^{\hat{y}}(x^\star_{\hat{y}}) + \frac{W\lambda}{4C} = v(\bv{x}^\star,\brr(\bv{x}^\star)) + \frac{W\lambda}{3C}.
\end{equation*}
For $y \neq \hat{y}$, we have $\tilde{x}_y > x^\star_y - \frac{W\lambda}{6C^2}$, and so our upper slope bound requires that
\begin{equation*}
    v^y(\tilde{x}_y) < v^y(x^\star_y) + C \frac{W\lambda}{6C^2} \leq v(\bv{x}^\star,\brr(\bv{x}^\star)) + \frac{W\lambda}{6C} \leq v(\bv{x}^\star,\brr(\bv{x}^\star)) + \frac{W\lambda}{3C} - \eps.
\end{equation*}
Consequently, we have $\BestResp^\eps(\tilde{\bv{x}}) = \{ \hat{y} \}$. Finally, we compute
\begin{align*}
    u(\tilde{\bv{x}},\brr(\tilde{\bv{x}})) &= u^{\hat{y}}(\tilde{x}_{\hat{y}})\\
    &\geq u^{\hat{y}}(x^\star_{\hat{y}}) - C|x^\star_{\hat{y}} - \tilde{x}_{\hat{y}}|\\
    &\geq u(\bv{x}^\star,\brr(\bv{x}^\star)) - \frac{W \lambda}{6C}\\
    &> u(x^\star,\brr(\bv{x}^\star)) - \lambda,
\end{align*}
verifying that $\tilde{\bv{x}}$ is indeed a $\lambda$-approximate Stackelberg equilibrium strategy for the principal.
\end{proof}

\subsection{Exact search with bounded bit precision (discussion in Section~\texorpdfstring{\ref{ssec:ssg-comparison}}{3.4})}
\label{app:ssg-exact-search}
We now analyze \Clinch{} imposing the additional regularity assumptions of \cite{peng2019learning}. 

\begin{assumption}
\label{as:bit-precision}
Agent utilities are linear with rational coefficients whose denominators are at most $2^L$, and that each non-empty best response region has volume at least $2^{-nL}$. Moreover, $\cX$ is a polytope represented as the intersection of a finite set of half-spaces, each of the form $\{ \bv{x} \in [0,1]^n : \bv{x}^\top a \leq b \}$ where $b \in \R$ and each entry of $a \in \R^n$ are rational with numerators and denominators at most $2^L$.
\end{assumption}

By the discussion of our regularity assumptions in \Cref{sec:security-games-model}, it suffices to take $C = 2^L$. With these settings, \Cref{thm:clinch} states that \Clinch{} terminates in $O(n L + n\log \frac{1}{\delta})$ oracle queries, and the returned strategy $\hat{\bv{x}} \in \cX$ satisfies $\|\hat{\bv{x}} - \bv{x}^\star\|_\infty < \delta$. A more careful analysis can eliminate dependence on $W$---yielding query complexity $O(nL + n\log\frac{1}{\lambda})$---by avoiding the final perturbation step of \Clinch{}. However this improvement will not impact our final result.

Next, we bound the bit complexity of $\bv{x}^\star$.

\begin{lemma}
If agent utilities are linear with rational coefficients whose denominators are at most $2^L$, then the entries of $\bv{x}^\star$ are rational with denominators at most $2^{8Ln}$.
\end{lemma}
\begin{proof}
Fix any $y \in \BestResp(x^\star)$, and write $v^y(t) = d_y - c_y t$, where $c_y, d_y \in (0,1]$ are rational with denominators at most $2^L$. As noted in the proof of \Cref{prop:ssg-solution}, we must have $x^\star_y > 0$. 
Writing $w^\star = v(\bv{x}^\star,\brr(\bv{x}^\star)) = v^y(x^\star_y)$, we can solve for $x^\star_y = (d_y - w^\star)/c_y$. For $y \not\in \BestResp(x^\star)$, we have $x^\star_y = 0$.

Now, since $\bv{x}^\star$ minimizes the agent's best response utility, we know that $w^\star$ is as small as possible so that the $\bv{x}^\star$ as determined above lies in $\cX$. In other words, $\bv{x}^\star$ must lie on a face of $\cX$, represented as $\{ \bv{x} \in \R^n : a^\top \bv{x} = b \}$. Thus, we have $\sum_{y \in \BestResp(\bv{x}^\star)} a_y (d_y - w^\star)/c_y = b$ and can compute
\begin{equation*}
    w^\star = \frac{\sum_{y \in \BR(\bv{x}^\star)} a_yd_y/c_y - b}{\sum_{y \in \BR(\bv{x}^\star)} a_y/c_y}.
\end{equation*}
Our bit precision assumptions imply that $w^\star \in (0,1]$ is rational with denominator at most $2^{5Ln+L}$, and so each $x^\star_y \in [0,1]$ must also be rational with denominator at most $2^{5Ln+3L}$.
\end{proof}

Hence, rounding appropriately, we have the following.

\begin{proposition}
Under Assumption~\ref{as:bit-precision}, running \Clinch{} with $\delta = \frac{1}{3}2^{-8Ln}$ and rounding each entry of the result to the nearest multiple of $2^{-8Ln}$ gives $\bv{x}^\star$ using $O(n^2 L)$ best response queries.
\end{proposition}

%% file: A3-clinch-experiments.tex
\section{Full Details for Myopic Numerical Simulations (Section~\texorpdfstring{\ref{ssec:ssg-comparison}}{3.4})}\label{app:clinch-experiments}

In \Cref{fig:clinch_vs_securitysearch} of Section~\ref{ssec:ssg-comparison}, we compare the performance of \Clinch{} to that of the previous state-of-the-art, \SecuritySearch{} \citep{peng2019learning}, given queries to a best response oracle.  Here, we provide further details on these experiments. Recall that code for the algorithm implementations and plots is available at \url{https://github.com/sbnietert/learning-stackelberg-games}.

\subsection{Implementation details}
We implemented each algorithm in Python and NumPy according to their respective specifications. The best response oracle was straightforward to implement for the SSGs describe below, since their utilities are linear. We note that \Clinch{} is fully-specified without knowledge of $C$ since the best response oracle is exact (i.e., $\eps = 0$).

\subsection{Experimental setup}
\label{ssec:clinch-experimental-setup}

\Cref{fig:clinch_vs_securitysearch} depicts the query complexity of \Clinch{} and \SecuritySearch{} on a sequence of SSGs with number of targets $n$ ranging from $5$ to $100$. We examine two settings:

\textbf{Setting 1.} Given $n$, we consider the SSG where the defender's strategy space is the unit simplex $\Delta_{n-1}$ with payoffs such that the attacker (resp.\ defender) receives value $1$ if they successfully attack (resp.\ defend) and value $0$ otherwise. While it is clear that the optimal defender strategy is to mix uniformly over all $n$ targets, this problem specification is unknown to the algorithms (and thus they must learn this from scratch).

In practice, we find that \SecuritySearch{} suffers from severe numerical stability issues due to the symmetry between the targets of the problem described above. To alleviate stability problems of \SecuritySearch{} and obtain a fairer comparison, we run this algorithm on a perturbed version of this problem where the payoff for successfully attacking or defending each target is slightly perturbed, by a uniformly random quantity between $0$ and $0.0005$.

\textbf{Setting 2.} Given $n$, we sample an SSG where the defender's strategy space is the unit simplex $\Delta_{n-1}$ with payoffs as follows: for each target, the attacker and defender have independent values for a successful attack or defense, each sampled independently and uniformly from $[0, 1]$; furthermore, each agent receives payoff $0$ if they unsuccessfully attack or defend. Note that these games are non-zero sum, since the attacker and defender have independent valuations. 

For each value of $n$, we sample $3$ games as described above and report the averaged number of oracle queries across these three instantiations. In each setting, we run the algorithms until they solve for the equilibrium nearly exactly, up to an accuracy of $10^{-8}$ in each coordinate.

\subsection{Results}

To illustrate the asymptotic scaling of sample complexity clearly, \Cref{fig:clinch_vs_securitysearch} depicts our results on a log-log scale ($n$ versus query count). In addition, to estimate the scaling rate, we plot a best linear fits of the log-transformed variables for each curve.

\Cref{fig:clinch_vs_securitysearch} shows that, in Setting 1, \Clinch{} requires fewer than $100$ samples when $n = 5$ and fewer than $2000$ samples when $n = 100$, whereas \SecuritySearch{} requires over $10^4$ samples when $n = 5$ and over $10^8$ samples when $n = 100$.  \Cref{fig:clinch_vs_securitysearch} also shows that, in Setting 2, \Clinch{} requires fewer than $100$ samples when $n = 5$ and fewer than $2000$ samples when $n = 100$, whereas \SecuritySearch{} requires over $4000$ samples when $n = 5$ and over $2\cdot 10^6$ samples when $n = 100$.

\subsection{Discussion}

We find \Clinch{} outperforms \SecuritySearch{} both in the constant factor hidden by the big-$O$ and the asymptotic query complexity in $n$. The empirical complexities match theory, with the cost of \SecuritySearch{} scaling roughly as $n^{3}$ and the cost of \Clinch{} scaling roughly as $n$. (We note that the \SecuritySearch{} scales with exponent around $n^2$ on the random instantiations, suggesting better average- than worst-case performance; however, it is still outperformed by the linear scaling of \Clinch{}.) \Clinch{} is even efficient for small $n$, improving over \SecuritySearch{} by two orders of magnitude in the query complexity.

In summary, we find that \Clinch{} runs efficiently, being both asymptotically optimal and having a small constant factor in practice, while \SecuritySearch{} struggles even in these simple settings.

%% file: A4-non-myopic-clinch-experiments.tex
\section{Full Details for Non-myopic Numerical Simulations (Section~\texorpdfstring{\ref{ssec:non-myopic-clinch-experiments}}{4.3})}\label{app:non-myopic-clinch-experiments}

In \Cref{fig:multi-threaded-simulations} of Section~\ref{ssec:non-myopic-clinch-experiments}, we compare the performance of multi-threaded and batched \Clinch{} against simulated non-myopic agents. Here, we provide further details on these experiments. Recall that code for the algorithm implementations and plots is available at \url{https://github.com/sbnietert/learning-stackelberg-games}.

\subsection{Implementation details}

Both batched and multi-threaded versions of \Clinch{} were implemented using Python and NumPy. They were structured to advance their state one round of agent interaction at a time, so that the agent can copy their state and use it to simulate several potential future trajectories. The batched variant uses the na\"ive repetition approach described at the beginning of Section~\ref{subsec:ssg-nonmyopic}. Given a batch size $B$, it sets accuracy $\lambda = nB/T$ and runs \Clinch{} with $\delta=\frac{W\lambda}{6C^2}$ and $\eps = \frac{W\lambda}{200C^5n}$ until some $\hat{x}$ is returned, na\"ively repeating each query $B$ times. Then $\tilde{x} = \Perturb(\hat{x},\lambda)$ is played for the remaining rounds. Our multi-threaded algorithm runs $\log T$ threads in parallel, as in \MultiThreadedClinch{}. However, the exploration phase for thread with delay $B$ simply performs the batched search described above, instead of the full algorithm's series of searches. During each thread's exploit phase, we play the perturbed result of the highest-indexed thread which has entered the exploit phase. This variant is faster to simulate than the full algorithm and still achieves $\tilde{O}\bigl(n T_\gamma \log^{O(1)}(TC/W)\bigr)$ regret.

Our agent is defined by a discount function $\nu$ mapping delay $\tau$ to discount level $\nu(\tau)$, taken to be $\gamma^\tau$ for the geometric discounting plots in \Cref{fig:multi-threaded-simulations}. Then, at each round $t$, it selects target $y \in \cY$ maximizing $v(x_t,y) + \sum_{\tau = 1}^T \nu(\tau) v(x_{t+\tau},\brr(x_{t+\tau}))$, where the future $x_{t+\tau}$ strategies are obtained by simulating the principal's algorithm forward with best responses.

\begin{figure}
    \centering
    \includegraphics[width=\linewidth]{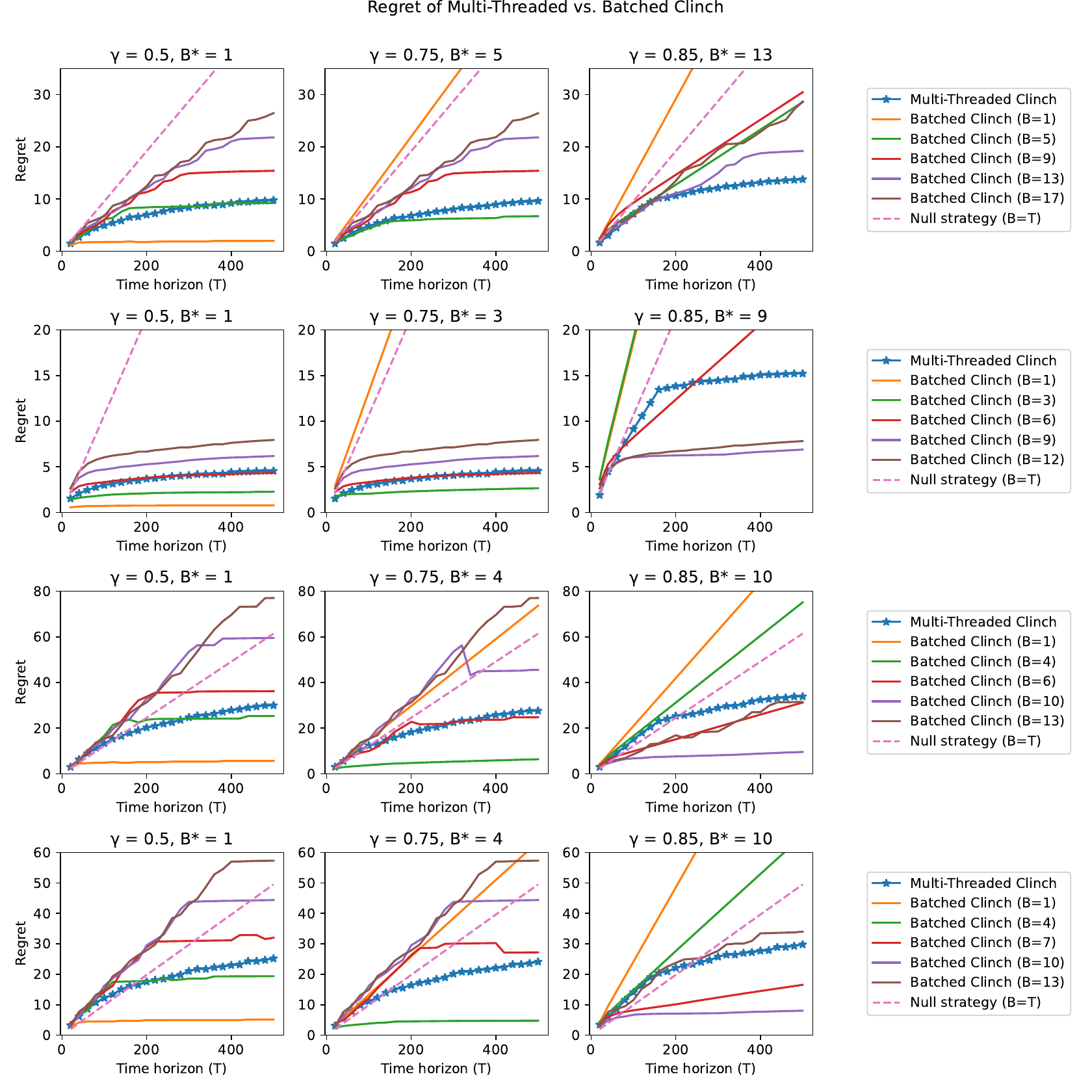}
    \caption{Regret achieved by batched and multi-threaded variants of \Clinch{} against a simulated $\gamma$-discounting agent on four random SSG instances with $n=3$. For each instance and discount factor, we note the optimal batch size $B^\star$ at $T=500$.}
    \label{fig:multi-threaded-extra-simulations}
\end{figure}

\begin{figure}
    \centering
\includegraphics[width=\linewidth]{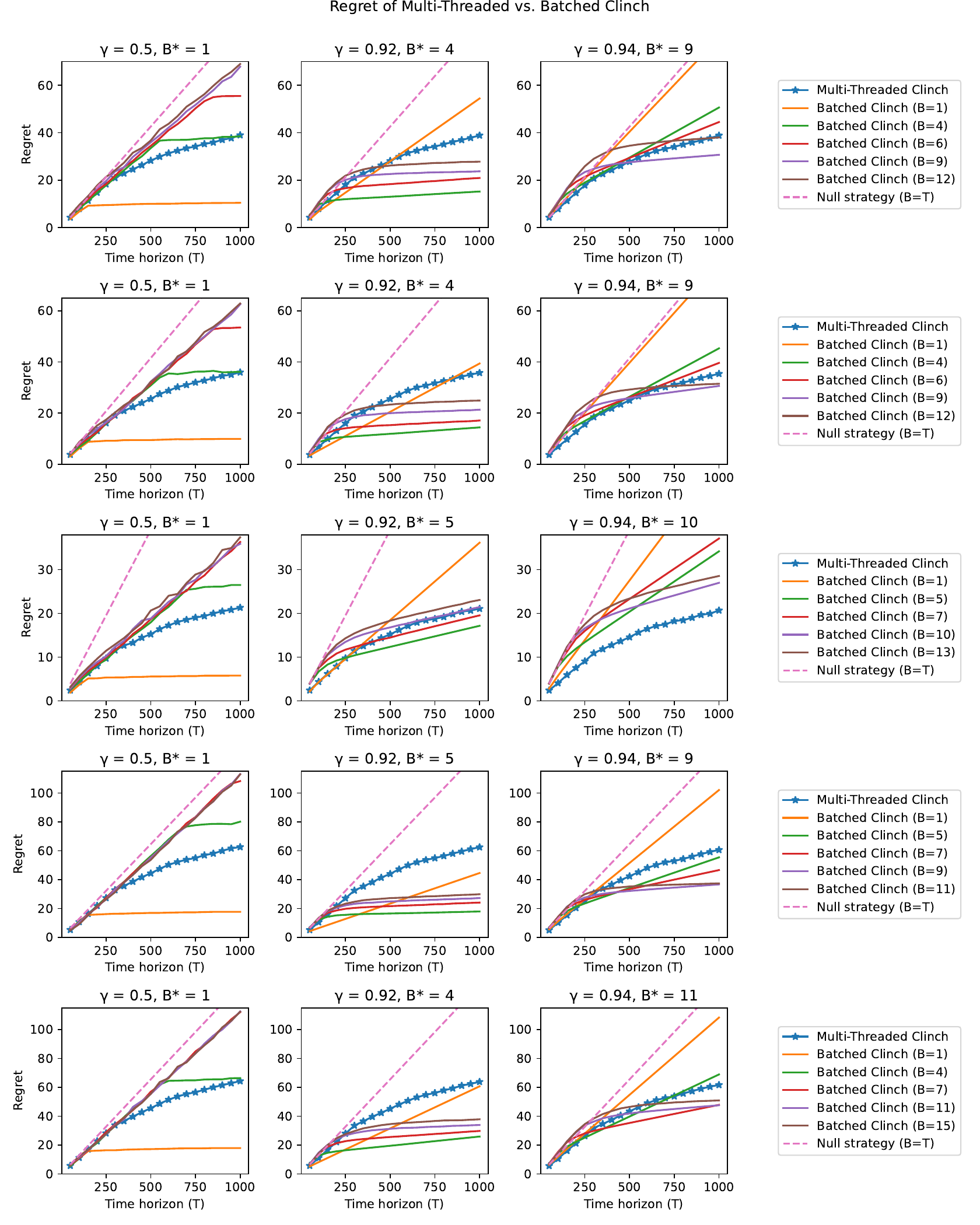}
    \caption{Regret achieved by batched and multi-threaded variants of \Clinch{} against a simulated $\gamma$-discounting agent on four random SSG instances with $n=10$. For each instance and discount factor, we note the optimal batch size $B^\star$ at $T=1000$.}
    \label{fig:multi-threaded-extra-simulations-n10}
\end{figure}

\subsection{Experimental setup}
\Cref{fig:multi-threaded-simulations} uses a simplex SSG with random linear utilities and $n=3$ targets, nearly as described in Setting 2 of \Cref{ssec:clinch-experimental-setup}. The only modification is that utilities are sampled randomly between $0.25$ and $0.75$ instead of $0$ and $1$, to ensure that the minimum width $W$ is not too small. For the principal's algorithms, we take $W = 0.25/(0.25 + (n-1)\cdot 0.75)$ and $C=1$. The slope bound is always valid since the coefficients are less than one, and the minimum width bound is valid when $n=2$ and empirically worked well for larger $n$. Although $W$ can feasible be much smaller for $n>2$, we achieved strong performance with no additional tuning. \Cref{fig:multi-threaded-simulations} depicts results for a single random SSG instance, though qualitatively similar results are observed for additional random instances.

\subsection{Additional results with geometric discounting}
\label{ssec:non-myopic-clinch-extra-experiments-geometric}

For Figure~\ref{fig:multi-threaded-extra-simulations}, we repeated the experiments for \Cref{fig:multi-threaded-simulations} with four additional random SSG instances. Observe that the multi-threaded algorithm always achieves sublinear regret, while any fixed batch size performs poorly if the discount factor is too large. For each instance the set of batch sizes displayed is selected to include the best batch size for each discount factor at $T=500$ (computed via brute-force search), along with an intermediate and larger batch size.

Finally, in Figure~\ref{fig:multi-threaded-extra-simulations-n10}, we present results for an extended set of experiments with $n = 10$. The maximum time horizon was extended from $T=500$ to $T=1000$, so that enough rounds pass by for searches to complete. Additionally, the set of discount factors was updated to $\{0.5, 0.92, 0.94\}$. The latter two were increased so that the simulated non-myopic agent has significant incentive to deviate from best-response behavior. We observe qualitatively similar results to Figure~\ref{fig:multi-threaded-extra-simulations}, although the overhead of multi-threading compared to selecting the optimal batch size, or slightly larger, is more pronounced in some regimes (particularly when $\gamma = 0.92$). Importantly, the regret of multi-threaded \Clinch{} is still sublinear, in contrast to the linear regret suffered when the batch size is too small.

\subsection{Additional results with hyperbolic discounting}
\label{ssec:non-myopic-clinch-extra-experiments-hyperbolic}
In Figures~\ref{fig:multi-threaded-hyperbolic-simulations} and \ref{fig:multi-threaded-hyperbolic-simulations-n10}, we repeat the experiments above for the alternative choice of hyperbolic discounting, taking $\nu(\tau) = 1/(1+k\tau)$ for varied $k$. Here, the optimal batch size for fixed $k$ is very sensitive to the time horizon $T$, so we use a fixed set of batch sizes for each choice of $n \in \{3,10\}$. Interestingly, the multi-threaded algorithm occasionally outperforms all of batched algorithms. To understand why this is possible, note that the number of future rounds which the agent can impact with their current action is substantially fewer with multi-threading, since the principal never commits to any fixed strategy for very long. Moreover, it is natural that this phenomenon is more pronounced with the less-aggressive, hyperbolic discounting. Indeed, a collection of many future rounds can impact the agent's discounted utility far more than any single future round (whereas they are within a factor of $T_\gamma$ under geometric discounting).

\begin{figure}
    \centering
    \includegraphics[width=\linewidth]{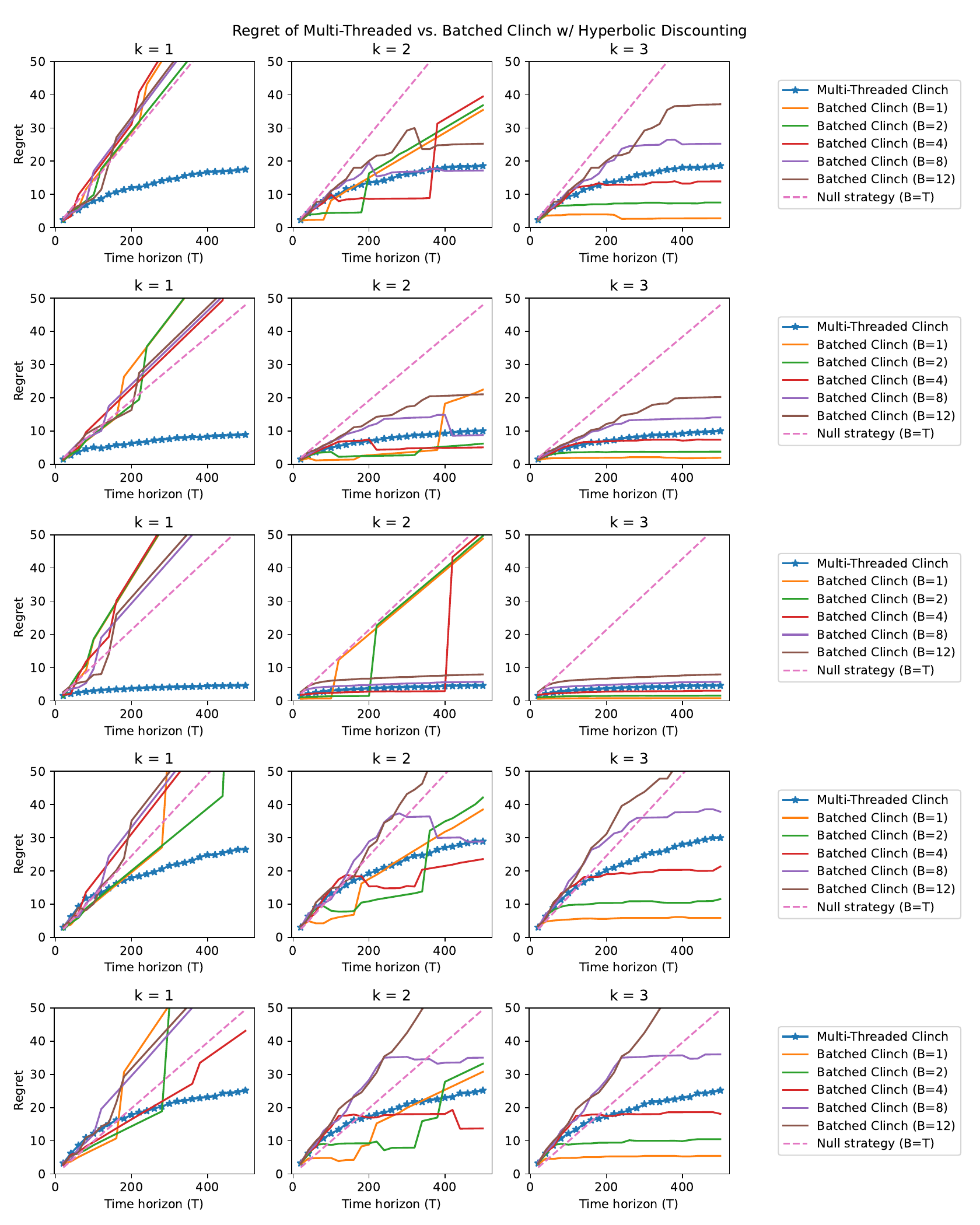}
    \caption{Regret achieved by batched and multi-threaded variants of \Clinch{} against a simulated hyperbolic discounting agent on five random SSG instances with $n=3$.}
    \label{fig:multi-threaded-hyperbolic-simulations}
\end{figure}

\begin{figure}
    \centering
    \includegraphics[width=\linewidth]{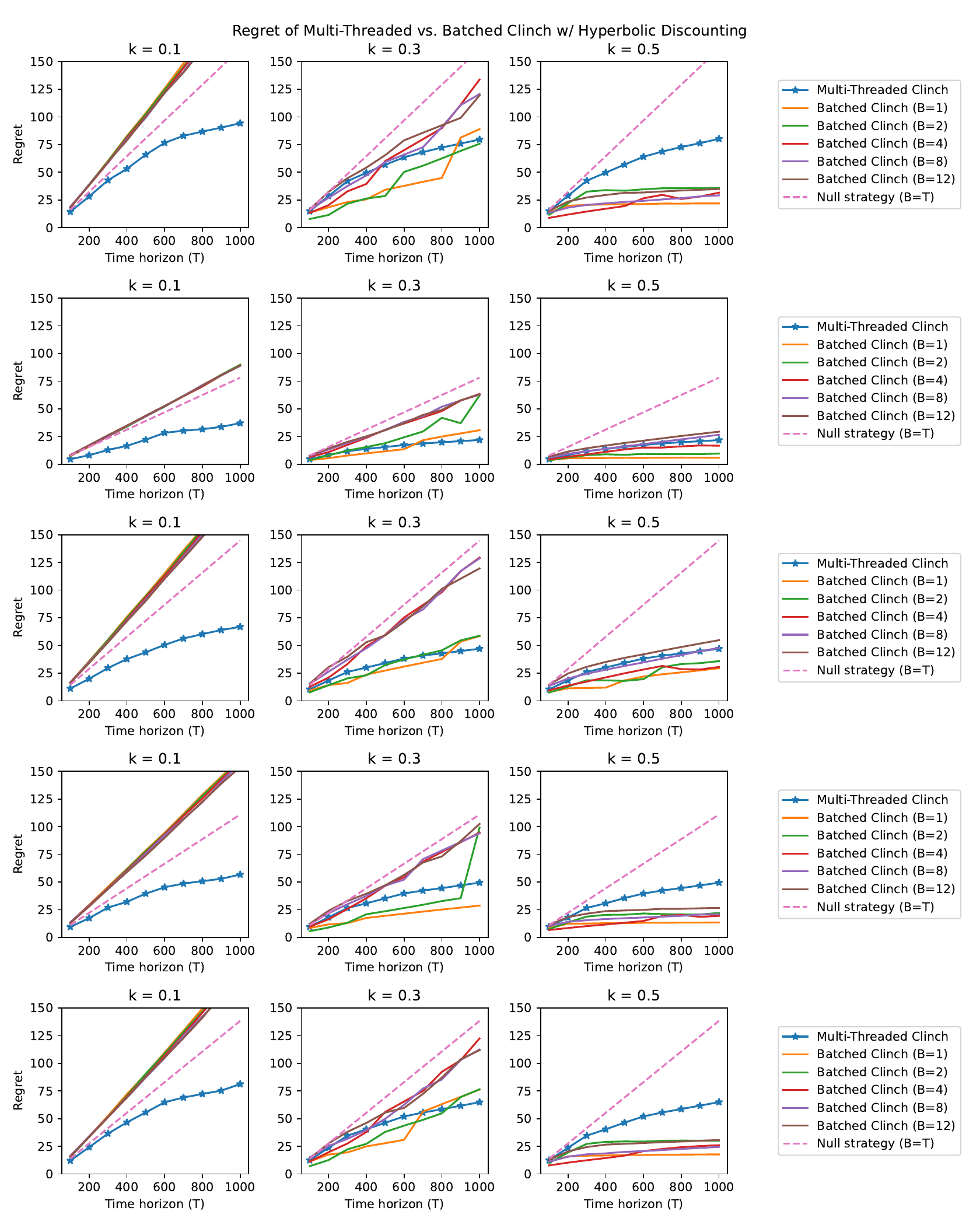}
    \caption{Regret achieved by batched and multi-threaded variants of \Clinch{} against a simulated hyperbolic discounting agent on five random SSG instances with $n=10$.}
    \label{fig:multi-threaded-hyperbolic-simulations-n10}
\end{figure}

%% file: A5-demand-learning.tex
\section{Supplementary Material for Demand Learning (Section~\texorpdfstring{\ref{ssec:demand}}{5.1})}\label{app:demand}

\subsection{Stochastic bandits with delays and perturbations (proof of Lemma~\texorpdfstring{\ref{prop:robust-delayed-bandits}}{5.5})}
\label{prf:robust-delayed-bandits}
Without loss of generality, we assume that each random interval $[\ell_t,u_t]$ is always contained within $[0,1]$. For analysis, it will be convenient to define empirical counts, means, and confidence bounds for all arms $i$ and rounds $t$ as
\begin{align*}
    &n_i(t) = \max\left\{\sum_{\tau=1}^{t}\mathds{1}\{i_\tau = i\},1\right\}, \quad \hat{\mu}_i(t) = \frac{1}{n_i(t)}\sum_{\tau=1}^{t} \mathds{1}\{i_\tau = i\} r_\tau\\
    &\mathrm{LCB}_i(t) = \hat{\mu}_i(t) - \sqrt{2\log(T)/n_i(t)} - \delta, \quad
    \mathrm{UCB}_i(t) = \hat{\mu}_i(t) + \sqrt{2\log(T)/n_i(t)} + \delta
\end{align*}
To start, we show that the confidence intervals are valid with high probability. 

\begin{lemma}
\label{lem:demand-learning-confidence-intervals}
With probability $1\!-\!\frac{2}{T^3}$, we have $\mathrm{LCB}_i(t)\!\leq\!\mu_i(t)\!\leq\!\mathrm{UCB}_i(t)$ for each arm $i$ and round $t$.
\end{lemma}
\begin{proof}
If arm $i$ has not been pulled by time $t$, the confidence bound $[\mathrm{LCB}_i(t),\mathrm{UCB}_i(t)]$ is trivially valid. Otherwise, conditioning on any arm pulls $i_1,\dots,i_t$ and considering the intervals $[\ell_t,u_t]$ guaranteed by the perturbation bound, Hoeffding's inequality implies that the corresponding empirical mean $\hat{\mu}_i(t)$ satisfies
\begin{equation*}
    \hat{\mu}_i(t) = \frac{1}{n_i(t)} \sum_{\substack{\tau \leq t\\i_\tau = i}} r_\tau \leq \frac{1}{n_i(t)} \sum_{\substack{\tau \leq t\\i_\tau = i}} u_t \leq \mu_i + \delta +  \sqrt{\frac{2 \log T}{n_i(t)}}
\end{equation*}
with probability at least $1 - \frac{1}{T^4}$. Likewise, we have $\hat{\mu}_i \geq \mu_i - \delta - \sqrt{2\log (T)/n_i(t)}$ with the same probability. Taking a union bound gives $\mu_i \in [\mathrm{LCB}_i(t), \mathrm{UCB}_i(t)]$ with probability at least $1 - \frac{2}{T^4}$. Since one confidence interval is modified per round, a union bound over rounds gives the lemma.
\end{proof}

Next, we note an arm $i$ can only contribute $O(\delta)$ regret in a given round if $\Delta_i = O(\delta)$. Hence, we call an arm \emph{acceptable} if $\Delta_i < 8\delta$ and \emph{unacceptable} otherwise. Conditioned on the ``clean'' event above, we show that the number of unacceptable arm pulls is bounded, extending the analysis from Theorem 2 of \cite{lancewicki2021stochastic} to the perturbed setting.

\begin{lemma}\label{lemma:unacceptable-bd}
Conditioned on the event from \Cref{lem:demand-learning-confidence-intervals}, no unacceptable arm $i$ is pulled more than $128 \log T / \Delta_i^2 + D/m + 2$ times, where $m$ is the number of remaining arms when it is pulled last.
\end{lemma}
\begin{proof}
To simplify analysis, we split the history of the algorithm into epochs, where epoch $\ell = 1,2,\dots$ denotes the $\ell$-th iteration of \SEDelayed{}'s main while loop. With this convention, after $\ell$ epochs, the remaining arms in $S$ have been pulled exactly $\ell$ times.

Now fix any unacceptable arm $i$, and consider the first epoch $\ell$ such that running \SEUpdate{} with the full (non-delayed) history through epoch $\ell$ would eliminate arm $i$ after the corresponding update to $S$. Then $i$ must be truly eliminated after $D + m$ additional rounds have passed, where $m$ is the number of arms remaining when $i$ is pulled for the last time. During these extra rounds, $i$ is pulled at most $D/m + 1$ times due to the round-robin nature of arm pulls; this is the overhead from delayed feedback.

Now if $\ell \leq 1$, we are done. Otherwise, fix $S$ as the set of arms remaining after the final round $t$ of epoch $\ell-1$, and let $\tilde{S} = \{ j \in S : \mathrm{UCB}_j(t) \geq \mathrm{LCB}_k(t) \text{ for all } k \in S \}$ denote the hypothetical update to $S$ based on non-delayed data. By the minimality of $\ell$, we know that $i \in \tilde{S}$, and so
\begin{align*}
    \hat{\mu}_i(t) \geq \max_{j \in S} \hat{\mu}_j(t) - 2\sqrt{\frac{2\log T}{\ell - 1}} - 2\delta \geq \max_{j} \mu_j - 3\sqrt{\frac{2\log T}{\ell - 1}} - 3\delta,
\end{align*}
where the second inequality follows by conditioning (noting in particular that the optimal arm is not eliminated). On the other hand, we have
\begin{align*}
    \hat{\mu}_i(t) \leq \mu_i + \sqrt{\frac{2\log T}{\ell - 1}} + \delta = \max_j \mu_j - \Delta_i + \sqrt{\frac{2\log T}{\ell - 1}} + \delta.
\end{align*}
Combining, we find that
\begin{equation*}
    \ell \leq \frac{32 \log T}{(\Delta_i - 4\delta)^2} + 1 \leq \frac{128 \log T}{\Delta_i^2} + 1.
\end{equation*}
Adding this upper bound to the overhead from delays gives the lemma.
\end{proof}

Now we are equipped to prove the main result.

\begin{proof}[Proof of \Cref{prop:robust-delayed-bandits}.]
By \Cref{lemma:unacceptable-bd}, we control regret by
\begin{equation*}
    \sum_{\Delta_i \geq 8\delta} n_T(i) \Delta_i + 8\delta T \leq 128 \sum_{\Delta_i > 0}\left( \frac{\log T}{\Delta_i} + \frac{D}{m_i} + 1\right) + 8\delta T,
\end{equation*}
where $m_i$ is the number of remaining arms when arm $i$ is pulled last. Bounding $\sum_i \frac{1}{m_i} \leq \sum_{i=1}^K \frac{1}{i} \leq \log(K) + 2$, we obtain a final bound of $128\sum_{\Delta_i > 0} \frac{\log (3T)}{\Delta_i} + 128 D \log (K) + 8\delta T$.
\end{proof}

\subsection{Perturbation bound for stochastic values (proof of Lemma~\texorpdfstring{\ref{lem:revenue-error-bd}}{5.6})}
\label{prf:revenue-error-bd}
\begin{proof}[Proof of \Cref{lem:revenue-error-bd}.]
If $a = 1$, then $v_t \geq p - \eps$ and $a = \mathds{1}\{v_t \geq p - \eps\} = u$, while, if $a = 0$, then $v_t \leq p + \eps$ and $a = \mathds{1}\{v_t > p + \eps\} = \ell$. Moreover, we have
\begin{align*}
    p\E[u] = p \Pr(v_t \geq p - \eps) \leq p d(p) + p L \eps \leq f(p) + L \eps,
\end{align*}
using the definitions of $d$ and $f$, the Lipschitz property of $d$, and that $p \in [0,1]$. Likewise, we bound
\begin{equation*}
    p\E[\ell] = p \Pr(v_t > p + \eps) = p \Pr(v_t \geq p + \eps) \geq f(p) - L \eps.\qedhere
\end{equation*}
\end{proof}

%% file: A6-finite-stackelberg.tex
\section{Supplementary Material for Finite Stackelberg Games (Section~\texorpdfstring{\ref{ssec:finite-games}}{5.2})}
\label{app:finite-games}

This section provides full details and analysis for \MultiThreadedFiniteAlg{} (\Cref{alg:multi-threaded-finite-games}) to prove \Cref{thm:finite-games}. We introduce this algorithm and its principal subroutine \FiniteAlg{} (\Cref{alg:finite-games}) in \Cref{app:finite-games-algs}. In \Cref{app:finite-games-search-proof}, we state a search guarantee for \FiniteAlg{}, \Cref{lem:finite-games-search}, and use it to prove the theorem. In \Cref{app:finite-games-search-proof}, we state three lemmas, pertaining to polytope conditioning bounds and robust convex optimization with membership queries, and use them to prove \Cref{lem:finite-games-search}. We prove the remaining lemmas in Sections~\ref{prf:conservative-membership-oracle} and~\ref{prf:optimization-with-membership-queries}. Throughout, we make use of the constants $r$, $\Delta$, and $V$ defined in \Cref{ssec:finite-games}.

\subsection{Algorithm definitions and discussion}
\label{app:finite-games-algs}

We first present \FiniteAlg{} (\Cref{alg:finite-games}), a procedure for learning in finite games with $\eps$-approximate best responses.
Formally, the algorithm takes as input a desired search accuracy $\delta$ and an approximate best response oracle \Oracle{} which, given query $\bv{x} \in \cX$, returns $\Oracle(\bv{x}) \in \BR^\eps(\bv{x})$ for some $\eps \geq 0$. 
For meaningful guarantees, we require $\eps \leq \left(\frac{\delta r \Delta}{nm}\right)^{O(1)}$. This will later be implemented against discounting agents via delayed feedback. \FiniteAlg{} outputs a $\delta$-approximate Stackelberg equilibrium pair (see \Cref{app:finite-games-thm-proof} for a formal statement).

The algorithm initially samples $\widetilde{O}(V^{-1})$ points from $\cX$ uniformly at random. For each sampled point $\bv{x}$, we obtain an approximate best response $y$ from \Oracle{} and run \textsc{ConservativeBestResponse} (\Cref{alg:oracle-adjustment} in \Cref{prf:conservative-membership-oracle}). This subroutine tests whether $\bv{x}$ is robustly within $K_y$ using multiple queries to \Oracle{} in the neighborhood of $\bv{x}$. We are left with a collection $\{\bv{x}^{(y)}\}_{y \in \cY_0}$ of sampled points which passed this test. Under the regularity assumptions, we prove that $y^\star \in \mathcal{Y}_0$ and that each $\bv{x}^{(y)}$ is well centered within $K_{y}$ with high probability.

Next, for each $y \in \cY_0$, we run an optimization procedure \textsc{MembershipOpt} (\Cref{alg:opt-with-membership-queries} in \Cref{prf:optimization-with-membership-queries}) to find an $\delta$-approximate maximizer $\hat{\bv{x}}^{(y)}$ for $u(\cdot,y)$ over $K_y$, starting at initial point $\bv{x}^{(y)}$. This subroutine applies convex optimization with membership queries, noting that $\mathds{1}\{\Oracle(\bv{x}) = y\} \approx \mathds{1}\{\bv{x} \in K_y\}$. \FiniteAlg{} then returns the strategy $\hat{\bv{x}}^{(y)}$ maximizing $u(\hat{\bv{x}}^{(y)},y)$.\medskip

\begin{algorithm}[H]
\caption{\FiniteAlg{}: robust learning for finite Stackelberg games}\label{alg:finite-games}
\DontPrintSemicolon
\SetAlgoNoLine
\SetKwInOut{Input}{input}
\SetKwInOut{Output}{output}
\Input{search accuracy $\delta \geq 0$, approximate best response oracle \Oracle{}}
\Output{$\delta$-optimal principal strategy $\hat{\bv{x}} \in \cX$}
$\cY_0 \gets \emptyset$, $\eta \gets \delta (3 \lceil V^{-1}\log  \frac{3}{\delta} \rceil)^{-1}$\;
\For{$i=1,\dots,\lceil V^{-1}\log \frac{3}{\delta} \rceil$}{\label{step:finite-games-sample}
    $y \gets \Oracle(\bv{x})$ for $\bv{x}$ sampled uniformly at random from $\cX$\;
    \If{$y \not\in \cY_0$ and $\textsc{ConservativeBestResponse}(y,\bv{x},r/2,\eta,\Oracle{}) = \textsc{True}$\label{step:finite-games-conservative-best-response-check}}{$\cY_0 \gets \cY_0 \cup \{y\}$, $\bv{x}^{(y)} \gets \bv{x}$}\label{step:finite-games-populate-set}
}
\textbf{for} $y \in \cY_0$ \textbf{do} $\hat{\bv{x}}^{(y)} \gets \textsc{MembershipOpt}(y,\bv{x}^{(y)},\frac{\delta}{3n},\Oracle{})$\label{step:finite-games-LSV}\;
\Return{$\hat{\bv{x}}^{(\hat{y})}$ for $\hat{y} \in \argmax_{y \in \cY_0} u(\hat{\bv{x}}^{(y)}, y)$}
\end{algorithm}\medskip

Finally, we present \MultiThreadedFiniteAlg{} (\Cref{alg:multi-threaded-finite-games}), a policy for the repeated game with $\gamma$-discounting agents (for unknown $\gamma$) that mirrors the multi-threaded approach of \MultiThreadedClinch{}. As before, each of $O(\log T)$ parallel threads runs a separate instance of \FiniteAlg{}, with thread $k$ experiencing delay $2^k$. Once a copy of \FiniteAlg{} terminates, its thread always plays the strategy returned by the largest eligible thread, where thread $k$ becomes eligible $2^k$ rounds after termination.\medskip

\begin{algorithm}[H]
\caption{\MultiThreadedFiniteAlg{}}\label{alg:multi-threaded-finite-games}
\DontPrintSemicolon
\SetAlgoNoLine
    \For{thread $k=1,\dots,\lfloor \log T \rfloor + 1$}{
        Initialize copy $\cA^{(k)}$ of $\FiniteAlg$ with $\delta = T^{-1}$
    }
    \For{round $t = 1,\dots,T$}{
        $k \gets \argmax\{\ell \in \mathbb{N}_{>0} : 2^{\ell-1} \text{ divides } t \}$ \tcp*{Identify current thread}
        \If{$\cA^{(k)}$ has not terminated}{
            Simulate oracle query/response for $\cA^{(k)}$ using $\bv{x}^{(t)}, y_t$\;
            \textbf{if} $\mathcal{A}^{(k)}$ \textit{terminates with output} $\hat{\bv{x}}$ \textbf{then} $\hat{\bv{x}}^{(k)} \gets \hat{x}$\;
        }
        \textbf{else} Play $\bv{x}^{(t)} \gets \hat{\bv{x}}^{(\bar{k})}$, where $\bar{k} = \max\{ \ell : \text{thread $\ell$ terminated by round $t - 2^\ell$}\}$ \label{step:multi-threaded-finite-alg-exploit}\;
    }
\end{algorithm}

\subsection{Learning against \texorpdfstring{$\gamma$}{γ}-discounting agents (proof of Theorem~\texorpdfstring{\ref{thm:finite-games}}{5.10})}
\label{app:finite-games-thm-proof}

We now provide a formal guarantee for \FiniteAlg{}, deferring the proof to \Cref{app:finite-games-search-proof}.

\begin{lemma}
\label{lem:finite-games-search}
Fix $\delta \in (0,1)$, and let \Oracle{} be an $\eps$-approximate best response oracle for some $\eps \geq 0$. Then $\FiniteAlg(\delta,\Oracle{})$ terminates after at most $100 V^{-1} \sqrt{m}\log^2 \left(\frac{3}{\delta}\right) \log V^{-1} + 10^7 m^{2.5} n \log^3\left(\frac{10 mn}{\delta r}\right)$  oracle calls. If $\eps \leq \left(\frac{\delta r}{200nm}\right)^{20}\Delta$, then, with probability at least $1-\delta$, the returned strategy $\hat{\bv{x}} \in \cX$ satisfies $u(\hat{\bv{x}},y) \geq \max_{\bv{x} \in \cX} u(\bv{x},\br(\bv{x})) - \delta$ for all $y \in \BR^\eps(\hat{\bv{x}})$.
\end{lemma}

We are now equipped to prove the main theorem.

\begin{proof}[Proof of \Cref{thm:finite-games}]
Denote by $Q = 100 V^{-1} \sqrt{m}\log^2 (3T) \log V^{-1} + 10^7 m^{2.5} n \log^3(10 mnT/r)$ and $\eps = \left(\frac{r}{200nmT}\right)^{20}\Delta$ the query complexity and oracle accuracy required by \Cref{lem:finite-games-search} when $\delta = T^{-1}$. As with \MultiThreadedClinch{}, thread $k$ runs on rounds $2^{k-1}(2\ell-1)$ for $\ell=1,2,\dots$ and hence experiences delay $2^k$ while copy $\cA^{(k)}$ of \FiniteAlg{} is running. For this copy's final oracle call, the selection rule at Step~\ref{step:multi-threaded-finite-alg-exploit} ensures that the $2^k$ round delay is maintained. We term the rounds up to this point for thread $k$ its ``exploration phase,'' since it is learning from agent feedback. Once $\cA^{(k)}$ has terminated, thread $k$ enters an exploitation phase and ignores agent feedback (i.e., infinite feedback delay).

We now establish several consequences of \Cref{lem:finite-games-search}, applied with our choice of $\delta = T^{-1}$. First, we can bound the total number of exploration rounds by $(\lfloor \log T \rfloor + 1) Q$. Next,  let $k^\star = \log_2 \left\lceil T_\gamma \log \frac{T_\gamma}{\eps}\right\rceil$ be the index of the first thread whose delay during exploration induces $\eps$-approximate best responses by \Cref{prop:delayed-feedback} (we can assume that $k^\star \leq \lfloor \log T \rfloor + 1$ is a valid thread index; otherwise the regret bound holds trivially). We freely condition on the event that the search of $\cA^{(k)}$ terminates successfully for all $k \geq k^\star$, since the complement has probability at most $O(T^{-1} \log T)$ by a union bound over threads. Starting at time $2^{k^\star}(Q+1)$, once copy $\cA^{(k^\star)}$ has terminated and a delay of $2^{k^\star}$ has passed, the strategy $\hat{\bv{x}}^{(\bar{k})}$ played at Step~\ref{step:multi-threaded-finite-alg-exploit} incurs regret at most $T^{-1}$, since $\bar{k} \geq k^\star$. Combining the above, we bound the regret by
\begin{align*}
    (\lfloor \log T \rfloor + 1) Q + 2^{k^\star}(Q+1) + 1 &= (\lfloor \log T \rfloor + 1) Q +  \left\lceil T_\gamma \log \frac{T_\gamma}{\eps}\right\rceil(Q+1) + 1\\
    &= \tilde{O}\left(\left(\log T + T_\gamma \log \frac{1}{r\Delta}\right)\left(V^{-1}\sqrt{m} \log^2 T + m^{2.5}n\log^3 \frac{T}{r}\right)\right)\\
    &= O\left(T_\gamma\left(V^{-1}\sqrt{m} \log^3 (T)  \log \frac{1}{r\Delta}+ m^{2.5}n\log^4 \frac{T}{r} \log \frac{1}{\Delta}\right)\right).
\end{align*}
Substituting the given values for $Q$ and $\eps$ gives the desired bound.
\end{proof}

\subsection{Robust search with \texorpdfstring{$\eps$}{ε}-approximate best responses (proof of Lemma~\texorpdfstring{\ref{lem:finite-games-search}}{F.1})}
\label{app:finite-games-search-proof}

Our search guarantee for \FiniteAlg{} relies on several lemmas, which we state below after establishing notation. For ease of presentation, we assume in what follows that the principal has $|\cX_0| = m+1$ actions, rather than $m$, so that $\cX$ can be identified with its isometric embedding into the ball $B(\sqrt{2})\!\subset\!\R^{m}$ (this poses no difficulties as the bulk of our analysis is coordinate-free). We write $B(A,r) \subset \R^m$ for the Minkowski sum of a set or point $A$ in $\R^m$ with the $\ell_2$-ball of radius $r$, with $B(r) \coloneqq B(\mathbf{0}_m,r)$, and, for a set $A \subseteq \R^m$, write $B(A,-r) \coloneqq \{ \bv{x} \in A : B(\bv{x},r) \subseteq A \}$. 

Finally, for any (potentially negative) $\eps \in \R$ and $\bv{x} \in \cX$, we define $\BR^\eps(\bv{x}) \coloneqq \{ y \in \cY : v(\bv{x},y) \geq v(\bv{x},y') - \eps \: \forall y' \in \cY \setminus \{y\} \}$. For each $y \in \cY$, let $K_y^\eps \coloneqq \{ \bv{x} \in \cX : y \in \BR^\eps(\bv{x}) \} \subseteq \R^m$ and set $K_y \coloneqq K_y^0$. Negative values of $\eps$ are relevant because they control the extent to which neighboring $|\eps|$-approximate best response regions can overlap with $K_y$. In particular, if $\eps \geq 0$ and $x \in K_y^{-2\eps}$, then $\BR^{\eps}(\bv{x}) = \{y\}$ (with the constant of two taken to avoid reliance on tie-breaking).

We first provide a correctness guarantee for \textsc{ConservativeBestResponse}.

\begin{lemma}
\label{lem:conservative-membership-oracle}
Fix $\bv{x} \in \cX$, $y \in \cY$, margin $\lambda \geq 0$, and failure probability $\delta \in (0,1)$. Let \Oracle{} be an $\eps$-approximate best response oracle for $\eps \geq 0$. Then, $\textsc{ConservativeBestResponse}(y,\bv{x},\lambda,\delta,\allowbreak\Oracle{})$ terminates after $12 \sqrt{m}\log \delta^{-1}$ oracle calls. If further $\eps \leq \frac{\lambda\Delta}{6\sqrt{m}}$, then, with probability at least $1-\delta$, the subroutine returns \textsc{True} only if $\bv{x} \in B\bigl(K_{y}^{-2\eps},-\frac{\lambda}{2\sqrt{m}}\bigr)$ and \textsc{False} only if $\bv{x} \not\in B(K_{y}^{-2\eps}, -\lambda)$.
\end{lemma}

Within the proof of \Cref{lem:conservative-membership-oracle}, we show the following useful fact.

\begin{lemma}
\label{lem:best-response-margin}
For each $y \in \cY$ and $\eps \geq 0$, we have $B\bigl(K_y,-\eps/\Delta\bigr) \subseteq K_y^{-\eps}$.   
\end{lemma}

This result translates a distance margin of $\eps/\Delta$ from the boundary of $K_y$ to a utility margin of $\eps$. Next, we give an optimization guarantee for \textsc{MembershipOpt}.

\begin{lemma}
\label{lem:optimization-with-membership-queries}
Fix $y \in \cY$, accuracy $\delta \in (0,1)$, initial point $\bv{x}_0 \in \cX$, and radius $\rho > 0$. Let \Oracle{} be an $\eps$-approximate best response oracle for $\eps \geq 0$. Then $\textsc{MembershipOpt}(y, \delta, \bv{x}_0, \rho, \Oracle{})$ terminates after at most $10^5 m^{2.5} \log^3\bigl(\frac{120m}{\delta \rho}\bigr)$ oracle calls. If further $\eps \leq \left(\frac{\delta\rho}{140 m}\right)^{13}\Delta $ and $B(x_0,\rho) \subseteq K_y$, then, with probability at least $1-\delta$, $\textsc{MembershipOpt}$ returns $\hat{\bv{x}} \in \cX$ such that $\BR^\eps(\hat{\bv{x}}) = \{y\}$ and $u(\hat{x},y) \geq \max_{\bv{x} \in K_y} u(\bv{x},y) - \delta$.
\end{lemma}

Together, these suffice to prove the search guarantee.

\begin{proof}[Proof of \Cref{lem:finite-games-search}]
Since the initial sampling loop at Steps~\ref{step:finite-games-sample}-\ref{step:finite-games-populate-set} calls \textsc{ConservativeBestResponse} $N = \lceil V^{-1}\log  \frac{3}{\delta} \rceil$ times with margin $\lambda = r/2$, and $\eps \leq \frac{\lambda\Delta}{6\sqrt{m}}$, the accuracy guarantee of \Cref{lem:conservative-membership-oracle} holds for all calls within the loop with probability at least $1 - \delta/3$. Conditioned on this event, we analyze the sampling loop. Let $S \subseteq K_{y^\star}$ denote the ball of radius $2r$ guaranteed by the regularity assumptions. The sample count $N$ is taken sufficiently large such that some sampled point $\bv{x}$ will lie inside $B(S,-r) \subseteq B(K_{y^\star},-r)$ with probability at least $1 - \delta/3$; indeed, the probability that any single point lies inside $B(S,-r)$ is $V$. Condition further on this event.

Since $\eps \leq r\Delta/4$, \Cref{lem:best-response-margin} implies that $\bv{x} \in B(K_{y^\star}^{-2\eps},-r/2)$, and so \Oracle{} will return $y^\star$ when $\bv{x}$ is queried. Moreover, since \textsc{ConservativeBestResponse} is run with margin $r/2$, \Cref{lem:conservative-membership-oracle} implies that $\bv{x}$ will pass the check at Step~\ref{step:finite-games-conservative-best-response-check} unless $\cY_0$ already contains $y^\star$. Thus, the final set $\cY_0$ will contain $y^\star$. Moreover, for each $y \in \cY_0$, \Cref{lem:conservative-membership-oracle} requires that the accepted strategy $\bv{x}^{(y)}$ have margin at least $\frac{r}{4\sqrt{m}}$ within $K_{y}^{\eps}$. Again by \Cref{lem:conservative-membership-oracle}, this sampling loop terminates within $N \cdot (1+12\sqrt{m} \log \frac{3N}{\delta})$ queries to \Oracle{}.

Finally, we examine the search loop at Step~\ref{step:finite-games-LSV}, where each call to \textsc{MembershipOpt} is run with accuracy $\gamma = \frac{\delta}{3n}$ and radius $\rho = \frac{r}{4\sqrt{m}}$. Since $\eps \leq \left(\frac{\gamma\rho}{140 m}\right)^{13}\Delta$, \Cref{lem:optimization-with-membership-queries} that implies that the search over $K_y$ will terminate with a $\gamma$-approximate maximizer $x^{(y)}$ within $10^5 m^{2.5} \log^3\bigl(\frac{120m}{\gamma \rho}\bigr)$ queries to \Oracle{}, with total success probability over all regions at least $1-\delta/3$. Taking a union bound over the three high probability events, we see that \FiniteAlg{} returns the desired maximizer with probability at least $1-\delta$. Indeed, conditioned on these good events, the returned strategy $\hat{\bv{x}}^{(\hat{y})}$, with $\hat{y} \in \argmax_{y \in \cY_0} u(\bv{x}^{(y)},y)$ satisfies
\begin{align*}
    u(\hat{\bv{x}}^{(\hat{y})},\hat{y}) = \max_{y \in \cY_0} u\bigl(\bv{x}^{(y)},y\bigr) 
    \geq \max_{\bv{x} \in K_{y^\star}} u(\bv{x},y^\star) - \gamma
    = \max_{x \in \cX} u(\bv{x},\br(\bv{x})) - \gamma
    > \max_{x \in \cX} u(\bv{x},\br(\bv{x})) - \delta,
\end{align*}
and $\BR^{\eps}(\hat{\bv{x}}^{(\hat{y})}) = \{ \hat{y} \}$, as desired.

For the final query complexity, we bound
\begin{align*}
    &\,N \cdot \left(1+12\sqrt{m} \log \frac{3N}{\delta}\right) + n \cdot 10^5 m^{2.5} \log^3\left(\frac{120m}{\gamma \rho}\right)\\
    \leq & \, 2 V^{-1} \log \frac{3}{\delta} \cdot \left(1+12\sqrt{m} \log \frac{3V^{-1} \log \frac{3}{\delta}}{\delta}\right) + n \cdot 10^5 m^{2.5} \log^3\left(\frac{120\cdot 12 m^{3.5}n}{\delta r}\right)\\
    < & \,100 V^{-1} \sqrt{m}\log^2 \left(\frac{3}{\delta}\right) \log V^{-1} + 10^7 m^{2.5} n \log^3\left(\frac{10 mn}{\delta r}\right),
\end{align*}
as desired.
\end{proof}

\subsection{Conservative best response data (proof of Lemma~\texorpdfstring{\ref{lem:conservative-membership-oracle}}{F.2})}
\label{prf:conservative-membership-oracle}

To ensure that the principal may safely commit to a strategy despite inexact best response feedback, we introduce \textsc{ConservativeBestResponse} (\Cref{alg:oracle-adjustment}) to determine whether a strategy $\bv{x}$ lies robustly within the best response polytope $K_y$. Specifically, given $\bv{x} \in \cX$, we query the best response oracle at $O(\sqrt{m} \log \delta^{-1})$ small perturbations of $\bv{x}$, only returning \textsc{True} if the oracle always responds with the fixed action $y$. If $K_y$ is sufficiently well-conditioned, a \textsc{True} output indicates that $\bv{x}$ lies firmly within the interior of $K_y$ with high probability, despite inexact best responses.\medskip

\begin{algorithm}[H]
\caption{\textsc{ConservativeBestResponse}}\label{alg:oracle-adjustment}
\DontPrintSemicolon
\SetAlgoNoLine
\SetKwInOut{Input}{input}
\SetKwInOut{Output}{output}
\Input{action $y \in \cY$, query $\bv{x} \in \R^m$, margin $\lambda \geq 0$, failure probability $\delta \in (0,1)$, approximate best response oracle \Oracle{}}
\Output{conservative estimate of $\mathds{1}\{\bv{x} \in K_y\}$}
\For{$i=1$ to $\smash{\lceil 6\sqrt{m}\log \delta^{-1}\rceil}$}{
    $\bv{w}_i \gets \bv{x} + \lambda \bv{S}_i$, where $\bv{S}_i$ is sampled uniformly at random from $\mathds{S}^{m-1}$\label{step:conservative-best-response-sampling}\; 
    \textbf{if} $\textsc{Oracle}(\bv{w}_i) \neq y$ \textit{or} $\bv{w}_i \not\in \cX$ \textbf{then} \Return{\textsc{False}}\;
    }
\Return{\textsc{True}}
\end{algorithm}\medskip

Our accuracy guarantee for \textsc{ConservativeBestResponse} relies on the following lemma, which provides a certain conditioning bound on each best response polytope $K_y^{\eps}$ as $\eps$ varies.

\begin{lemma}
\label{lem:best-response-polytope-margin}
If $y \in \cY$, $\lambda \geq 0$, and $\eps_1,\eps_2 \in \R$ with $\eps_1 \leq \eps_2$, then any $\bv{x} \in K_y^{\eps_2}$ with margin $\lambda + (\eps_2 - \eps_1)/\Delta$ has margin $\lambda$ within $K_y^{\eps_1}$. That is, we have $B\bigl(K_y^{\eps_2},-(\lambda + (\eps_2 - \eps_1)/\Delta)\bigr) \subseteq B(K_y^{\eps_1},-\lambda)$.
\end{lemma}

\begin{proof}
For this proof, we will view $\cX$ and the best response polytopes as subsets of the affine subspace $A \coloneqq \{ \bv{w} \in \R^{m+1}: \sum_{i=1}^{m+1} w_i = 1 \}$ in the natural way (rather than their isometric embeddings into $\R^m$). With this change, the sets of the form $B(S,-r)$ in the statement should be updated to $B_A(S,r) \coloneqq \{ \bv{x} \in S : B(\bv{x},r) \cap A \subseteq S \}$. For readability, we also write $\bar{\bv{v}}_y = \bar{\bv{v}}^{(y)} \in \R^{m+1}$ for each $y \in \cY$, where these are the centered utility profiles defined in the regularity assumptions. We naturally extend to the uncentered utility profiles, defined for each $y \in \cY$ and $i \in [m+1]$ by $v_y(i) \coloneqq v_0(i,y)$. Writing $c_y \coloneqq \frac{1}{m+1} \sum_{i=1}^{m+1} v_0(i,y)$, we have $\bar{\bv{v}}_y = \bv{v}_y - c_y \mathbf{1}_{m+1}$ for each $y \in \cY$.

We shall prove the contrapositive. To start, fix any $\bv{x} \in A \setminus B_A(K_y^{\eps_1},-\lambda)$. That is, there exists $\bv{x}' \in A \setminus K_y^{\eps_1}$ such that $\|\bv{x} - \bv{x}'\|_2 < \lambda$. Since $K_y^{\eps_1}$ is a polytope, a hyperplane tangent to one of its faces must separate $\bv{x}$ and $\bv{x}'$. If this hyperplane corresponds to one of the non-negativity constraints defining $\cX = \{ \bv{w} \in A: w_i \geq 0 \: \forall i \}$, then $\bv{x}' \in A \setminus \cX$, and so we trivially have $\bv{x} \in A \setminus B_A(K_y^{\eps_2},-\lambda) \subseteq A \setminus B_A\bigl(K_y^{\eps_2},-(\lambda + (\eps_2 - \eps_1)/\Delta)\bigr)$. Otherwise, the hyperplane must take the form $\{ \bv{w} \in \R^{m+1} : \bv{w}^\top (\bv{v}_{y'} - \bv{v}_{y}) = \eps_1 \}$ for some $y' \in \cY \setminus \{y\}$. In this case, the vector $\bv{x}'' \in A$ defined by $\bv{x}'' \coloneqq \bv{x}' + \frac{\eps_2 - \eps_1}{\Delta}\frac{\bar{\bv{v}}_{y'} - \bar{\bv{v}}_{y}}{\|\bar{\bv{v}}_{y'} - \bar{\bv{v}}_{y}\|_2}$ must satisfy
\begin{align*}
    (\bv{x}'')^\top (\bv{v}_{y'} - \bv{v}_{y}) &> \eps_1 + \frac{\eps_2 - \eps_1}{\Delta}\frac{(\bar{\bv{v}}_{y'} - \bar{\bv{v}}_{y})^\top (\bv{v}_{y'} - \bv{v}_{y})}{\|\bar{\bv{v}}_{y'} - \bar{\bv{v}}_{y}\|_2}\\
    &= \eps_1 + \frac{\eps_2 - \eps_1}{\Delta}\frac{(\bar{\bv{v}}_{y'} - \bar{\bv{v}}_{y})^\top (\bar{\bv{v}}_{y'} - \bar{\bv{v}}_{y} + (c_{y'} - c_y)\mathbf{1}_{m+1})}{\|\bar{\bv{v}}_{y'} - \bar{\bv{v}}_{y}\|_2}\\
    &= \eps_1 + \frac{\eps_2 - \eps_1}{\Delta}\frac{(\bar{\bv{v}}_{y'} - \bar{\bv{v}}_{y})^\top (\bar{\bv{v}}_{y'} - \bar{\bv{v}}_{y})}{\|\bar{\bv{v}}_{y'} - \bar{\bv{v}}_{y}\|_2}\\
    &= \eps_1 + \frac{\eps_2 - \eps_1}{\Delta}\|\bar{\bv{v}}_{y'} - \bar{\bv{v}}_y\|_2 > \eps_2,
\end{align*}
using our minimum distance assumption. Thus, $\bv{x}'' \in A \setminus K_y^{\eps_2}$ and, since $\|\bv{x}'' - \bv{x}\|_2 < \lambda + \frac{\eps_2 - \eps_1}{\Delta}$, we obtain $\bv{x} \in A \setminus B_A\bigl(K_y^{\eps_2},-(\lambda + (\eps_2 - \eps_1)/\Delta)\bigr)$, as desired.
\end{proof}

As a simple consequence, we obtain Lemma~\ref{lem:best-response-margin}.

\begin{proof}[Proof of Lemma~\ref{lem:best-response-margin}]
Setting $\lambda = \eps_2 = 0$ and $\eps_1 = -\eps$, we find that $B\bigl(K_y,-\eps/\Delta\bigr) \subseteq K_y^{-\eps}$.
\end{proof}

Next, to analyze the sampling procedure of \textsc{ConservativeBestResponse}, we recall a standard lower bound for the volume of a spherical cap (see, e.g., Lemma 9 of \citealp{feige2002maxcut}). We provide a brief proof below to clarify the constant prefactor.

\begin{lemma} 
\label{lem:spherical-cap-bound}
Let $\bv{Z} \sim \Unif(\mathbb{S}^{m-1})$ and $t \geq 0$. Then $\Pr(Z_1 > t) \geq \frac{1}{\sqrt{2\pi m}}(1-t^2)^{(m-1)/2}$.
\end{lemma}
\begin{proof}
The set $\{ \bv{z} \in \mathbb{S}^{m-1} : z_1 > t\}$ is an open spherical cap, whose boundary is an $(m-1)$-dimensional sphere with radius $\sqrt{1-t^2}$. The surface area of the cap is bounded from below by the volume of the sphere, which is given by $(1-t^2)^{(m-1)/2} \frac{\pi^{(m-1)/2}}{\Gamma((m+1)/2)}$. Normalizing by the surface area of $\mathbb{S}^{m-1}$ gives
\begin{equation*}
    \Pr(Z_1 > t) \geq \frac{(1-t^2)^{(m-1)/2} \Gamma(m/2)}{\sqrt{\pi}\Gamma((m+1)/2)} \geq \frac{(1-t^2)^{(m-1)/2}}{\sqrt{2\pi m}},
\end{equation*}
as desired, using that $\Gamma(m/2)/\Gamma((m+1)/2) \geq 1/\sqrt{2m}$.
\end{proof}

Finally, we prove the desired guarantee for \textsc{ConservativeBestResponse}.

\begin{proof}[Proof of Lemma~\ref{lem:conservative-membership-oracle}]
Suppose that a query $\bv{x} \in \cX$ does not lie robustly within $K_{y}$, in that $\bv{x} \not\in B\bigl(K_{y}^{-2\eps},-\frac{\lambda}{2\sqrt{m}}\bigr)$. Then, applying \Cref{lem:best-response-polytope-margin} with $\eps_1 = -2\eps$, $\eps_2 = \eps$, and margin $\frac{\lambda}{2\sqrt{m}}$, we find that $\bv{x} \not\in B\bigl(K_{y}^{\eps},-\frac{\lambda}{2\sqrt{m}}-3\eps/\Delta \bigr)$. Assuming that $\eps \leq \frac{\lambda\Delta}{6\sqrt{m}}$, this implies that $\bv{x} \not\in B\bigl(K_{y}^{\eps},-\lambda/\sqrt{m} \bigr)$. There thus exists an open half-space $H$ tangent to (but disjoint from) $K_{y}^\eps$, such that $d(\bv{x},H) \leq \lambda/\sqrt{m}$. Next, consider $\bv{w} = \bv{x} + \lambda \bv{S}$, for $\bv{S} \sim \Unif(\mathbb{S}^{m-1})$, as in Step~\ref{step:conservative-best-response-sampling}. By \Cref{lem:spherical-cap-bound}, we must have
\begin{align*}
    \Pr(\bv{w} \not\in K_{y}^\eps) &\geq \Pr(\bv{w} \in H)\\
    &= \Pr(\lambda S_1 > \lambda/\sqrt{m})\\
    &= \Pr(S_1 > 1/\sqrt{m})\\
    &\geq \frac{1}{\sqrt{2\pi m} (1 - \frac{1}{m})^{(m-1)/2}}\\
    &\geq \frac{1}{2\sqrt{2\pi m}} \geq \frac{1}{6\sqrt{m}}.
\end{align*}
Consequently, the probability that \textsc{ConservativeBestResponse} returns \textsc{True} is at most
\begin{align*}
    \bigl(1 - \tfrac{1}{6}m^{-1/2}\bigr)^{6 \sqrt{m} \log \delta^{-1}} \leq \exp(\log \delta) = \delta.
\end{align*}
On the other hand, if $\bv{x} \in B\bigl(K_{y}^{-2\eps},-\lambda\bigr)$, the algorithm will return \textsc{True} with probability 1, as desired. Regardless of whether $\eps$ is sufficiently small, the total number of calls to \Oracle{} is at most $\lceil 6\sqrt{m} \log \delta^{-1} \rceil \leq 12\sqrt{m}\log \delta^{-1}$.
\end{proof}

\subsection{Robust linear optimization with membership queries (proof of Lemma~\texorpdfstring{\ref{lem:optimization-with-membership-queries}}{F.4})}
\label{prf:optimization-with-membership-queries}

Our second subroutine, \textsc{MembershipOpt} (\Cref{alg:opt-with-membership-queries}), seeks to maximize the linear objective $u(\cdot,y)$ over a fixed best response region $K_y$, using queries to an $\eps$-approximate best response oracle \Oracle{}. Since the feedback $\mathds{1}\{\Oracle(\bv{x}) = y\}$ approximates $\mathds{1}\{\bv{x} \in K_y\}$, our approach mirrors existing work for robust convex optimization with membership queries (c.f.\ \citealt{lee2018efficient}).

In particular, we introduce a method \textsc{SimulatedSep} (\Cref{alg:simulated-sep}) that simulates a conservative separation oracle for $K_y$. That is, unless a query $\bv{x} \in \R^m$ lies within $K_y^{-2\eps}$, \textsc{SimulatedSep} returns a normal vector $\bv{w} \in \mathbb{S}^{m-1}$ for a half-space approximately separating $\bv{x}$ from $K_y$. To achieve this, we implement a conservative membership oracle for $K_y$ using \textsc{ConservativeBestResponse} and apply a reduction from separation to membership due to \cite{lee2018efficient}.

We then apply the standard center of gravity method (c.f.\ Section 2.1 of \citealt{bubeck2015}) to maximize our objective over $K_y$ using separation queries. More precisely, at each round $t$, we query \textsc{SimulatedSep} at the centroid $\bv{x}_t$ of the current search space $S_t$ and update $S_{t+1}$ to incorporate the obtained feedback, either intersecting $S_t$ with the returned half-space or eliminating all strategies with $u(\bv{x},y) < u(\bv{x}_t,y)$. After a moderate number of queries, \textsc{MembershipOpt} returns a queried point which maximizes $u(\cdot, y)$, among those which \textsc{SimulatedSep} failed to separate from $K_y$. We note that \Clinch{} from \Cref{sec:clinch} has a similar flavor, since both are cutting-plane methods.\medskip

\begin{algorithm}[H]
\caption{\textsc{MembershipOpt}: robust linear optimization via membership queries
}\label{alg:opt-with-membership-queries}
\DontPrintSemicolon
\SetAlgoNoLine
\SetKwInOut{Input}{input}
\SetKwInOut{Output}{output}
\Input{action $y \in \cY$, optimization accuracy $\delta \geq 0$, initial point $\bv{x}_0 \in \cX$, radius $\rho > 0$,\\
approximate best response oracle \Oracle{}}
\Output{approximate minimizer $\hat{\bv{x}}$ for $\max_{\bv{x} \in K_y} u(\bv{x},y)$, or ``$\perp$''}
$S_0 \gets B(\bv{x}_0,\sqrt{2})$, $A \gets \emptyset$, $t_f \gets \bigl\lceil 3m\log \frac{16}{\delta \rho}\bigr\rceil $, $\alpha \gets \min\{\frac{\delta}{t_f+1}, \frac{\delta \rho}{16 \sqrt{2m} }\}$\;
\For{$t = 0, \dots, t_f$}{
    \If{$\textsc{SimulatedSep}(y,\bv{x}_t,\rho,\alpha,\Oracle{})$ returns $\bv{w} \in \mathbb{S}^{m-1}$}{
        $S_{t+1} \gets \{ \bv{x} \in S_t: \bv{w}^\top \bv{x} \geq \bv{w}^\top \bv{x}_t \}$\;\label{step:center-of-gravity-cut}
    }
    \Else{
        $S_{t+1} \gets \{ \bv{x} \in S_t : u(\bv{x},y) \geq u(\bv{x}_t,y) \}$\;\label{step:center-of-gravity-cut2}
        $A \gets A \cup \{\bv{x}_t\}$
    }
    $\bv{x}_{t+1} \gets \E_{\bv{x} \sim \Unif(S_t)}[\bv{x}]$\;
}
\textbf{if} $A = \emptyset$ \textbf{then} \Return{$\perp$} \textbf{else} \Return{$\argmax_{\bv{x} \in A} u(\bv{x},y)$}\label{step:membership-opt-return}\\
\end{algorithm}\smallskip

\begin{algorithm}[H]
\caption{\textsc{SimulatedSep}: simulation of separation oracle via best response queries}\label{alg:simulated-sep}
\DontPrintSemicolon
\SetAlgoNoLine
\SetKwInOut{Input}{input}
\SetKwInOut{Output}{output}
\Input{action $y \in \cY$, query $\bv{x} \in \cX$, radius $\rho$, separation accuracy $\delta \in (0,1)$,\\
approximate best response oracle \Oracle{}}
\Output{normal vector $\bv{w} \in \mathbb{S}^{m-1}$ of half-space approximately separating $\bv{x}$ from $K_y$, or ``$\perp$''}
\textbf{if $\bv{x} \not\in \cX$} \textbf{then return} any $\bv{w} \in \mathbb{S}^{m-1}$ such that $\bv{w}^\top(\bv{x} - \bv{z}) \leq 0$ for all $\bv{z} \in \cX$
$\lambda \gets \frac{\delta^6\rho^6}{2^{36}m^{7/2}}$, $Q \gets 2m\lceil \log_2(2/\lambda) \rceil + 1, \gamma \gets \frac{\delta}{3Q}$\;
Define membership oracle \textsc{Mem} with domain $\R^m$ by $\textsc{Mem}(\bv{x}) \gets \textsc{ConservativeBestResponse}(y,\bv{x},\lambda,\gamma,\textsc{Oracle})$\label{step:simulated-sep-mem}\;
Run Algorithm 1 of \cite{lee2018efficient} with query access to \textsc{Mem} and parameters $\text{``$n$''} \gets m$, $\text{``$r$''} \gets \rho/2$, $\text{``$R$''} \gets \sqrt{2}$, $\text{``$\eps$''} \gets \lambda$; terminate after $Q$ queries to \textsc{Mem}\label{step:simulated-sep-LSV}\;
\textbf{if} \emph{Algorithm 1 asserts that no cut exists or is terminated before completion} \Return{$\perp$}\;
\textbf{else if} \textit{Algorithm 1 returns half-space defined by} $\tilde{\bv{g}}$ \textbf{then} \Return{$-\tilde{\bv{g}}/\|\tilde{\bv{g}}\|_2$}

\end{algorithm}\medskip

We first bound the query complexity of \textsc{SimulatedSep}. While this method always terminates after a fixed number of queries, we only obtain meaningful performance guarantees if $\eps$ is sufficiently small and $K_y$ contains a ball with radius bounded from below.

\begin{lemma}[Membership to separation] 
\label{lem:membership-to-separation}
Fix $y \in \cY$, $\bv{x} \in \R^m$, radius $\rho > 0$, and accuracy $\delta \in (0,1)$. Let \Oracle{} be an $\eps$-approximate best response oracle for some $\eps \geq 0$. Then, $\textsc{SimulatedSep}(y,\bv{x},\rho,\delta,\Oracle{})$ terminates after at most $10^3 m^{1.5} \log^2\bigl(\frac{100m}{\delta\rho}\bigr)$ queries to \Oracle{}. If further $K_y$ contains a ball of radius $\rho$ and $\eps \leq \frac{\delta^6\rho^6\Delta}{2^{39}m^4}$, then, with probability at least $1 - \delta$,  $\textsc{SimulatedSep}$ only returns ``$\perp$'' if $\bv{x} \in K_y^{-2\eps}$, and, if, $\textsc{SimulatedSep}$ returns $\bv{w} \in \mathbb{S}^{m-1}$, then $\bv{x}^\top \bv{w} \leq \bv{z}^\top \bv{w} + \delta$ for every $\bv{z} \in K_y^{-2\eps}$.
\end{lemma}
\begin{proof}
We first address sample complexity. Due to the manual cutoff at Step~\ref{step:simulated-sep-LSV}, we make at most $Q = 2m\lceil \log_2(2/\lambda) \rceil + 1$ queries to the simulated oracle \textsc{Mem}. By \Cref{lem:conservative-membership-oracle}, each query to \textsc{Mem} uses at most $12\sqrt{m}\log \frac{3Q}{\delta}$ queries to \Oracle{}. Combining gives the stated query complexity bound of 
\begin{align*}
    Q \cdot 12\sqrt{m}\log \frac{3Q}{\delta} &\leq 48m^{1.5} \log_2\left(\frac{2}{\lambda}\right)\left(\log \frac{12m \log_2\frac{2}{\lambda}}{\delta}\right)\\
    &= 48m^{1.5} \log_2\left(\frac{2^{37}m^{7/2}}{\delta^6\rho^6}\right)\left(\log \frac{12m \log_2\frac{2^{37}m^{7/2}}{\delta^6\rho^6}}{\delta}\right)\\
    &\leq 10^3 m^{1.5} \log^2\left(\frac{100m}{\delta\rho}\right).
\end{align*}
For the remainder of the proof, we assume that $\eps \leq \frac{\delta^6\rho^6\Delta}{2^{39}m^4}$. This bound was taken to ensure that $\eps \leq \min\{ \frac{\lambda \Delta}{6\sqrt{m}}, \frac{r\Delta}{4} \}$, and our choice of $\lambda$ was taken to ensure that $3600 m^{7/6} \lambda^{1/3} \rho^{-2} \delta^{-1} \leq \delta$. 

First, since $\eps \leq \frac{\lambda\Delta}{6\sqrt{m}}$, the accuracy guarantee of \Cref{lem:conservative-membership-oracle} holds for all queries to \textsc{Mem} with probability at least $1 - \alpha/3$, by a union bound. Writing $K = K_y^{-2\eps}$, we have by \Cref{lem:best-response-margin} that $B(K_y, -2\eps/\Delta) \subseteq K$. Since $K_y$ contains a ball of radius $\rho$ and $\eps \leq \rho\Delta/4$, $K$ must contain a ball of radius $\rho - \rho/2 = \rho/2$. Moreover, by \Cref{lem:conservative-membership-oracle}, the simulated oracle \textsc{Mem}, defined at Step~\ref{step:simulated-sep-mem}, is a $\lambda$-approximate, conservative membership oracle for the set $K$; that is, \textsc{Mem} only returns \textsc{True} for a query $\bv{z}$ if $\bv{z} \in K$ and only returns \textsc{False} if $\bv{z} \not\in B(K,-\lambda)$. Of course, the lemma's guarantee is slightly stronger, but this relaxation suffices.

Our result now nearly follows from the proof of Theorem 14 in \cite{lee2018efficient}, which provides an optimization guarantee for the Algorithm 1 which we apply at Step~\ref{step:simulated-sep-LSV}. We will slightly adapt their analysis to obtain explicit constants and to incorporate the conservative nature of our simulated membership oracle. First, we observe that the query cap of $Q = 2m \lceil \log_2(2/\lambda) \rceil + 1$ never goes into effect. Indeed, their Algorithm 1 calls \textsc{Mem} once at the beginning, and then at most $\lceil \log_2(2/\lambda) \rceil$ times within each of $2m$ binary searches performed by their Algorithm 2 subroutine.

Next, we verify our first guarantee, when \textsc{SimulatedSep} fails to find a separating hyperplane and returns ``$\perp$.''
Given a query $\bv{x} \in \R^m$, we only return ``$\perp$'' when their Algorithm 1 fails to return a half-space, which only occurs when \textsc{Mem} returns \textsc{True}; in this case, we must have $\bv{x} \in K$, as desired. If $\bv{x} \notin \cX$, then the returned half-space separates $\bv{x}$ from $\cX$ (and thus $K$) with no error. 

Otherwise, we must have $\bv{x} \in \cX \setminus B(K,-\lambda) \subseteq B(0,\sqrt{2}) \setminus B(K,-\lambda)$. In this case, we perform the same error analysis appearing in their proof of Theorem 14, but explicitly state constants appearing due to an implicit use of Markov's inequality. With our notation, they start by proving (within their Lemma 13) that the returned vector $\tilde{\bv{g}} \in \R^m$ satisfies
\begin{equation*}
    \frac{600}{\delta} m^{7/6} \lambda^{1/3} \rho^{-1} \geq \tilde{\bv{g}}^\top(\bv{z}-\bv{x})
\end{equation*}
for all $\bv{z} \in K$, with probability at least $1 - \delta/3$. Then, they lower bound
\begin{equation*}
    \tilde{\bv{g}}^\top \bv{x} \geq \|\bv{x}\|_2 - \zeta \|\bv{x}\|_\infty - 64 m^{7/6} \lambda^{1/3} \rho^{-1} ,
\end{equation*}
where $\zeta$ is a non-negative random variable with $\E[\zeta] < 24 m^{7/6} \lambda^{1/3} \rho^{-1}$. Applying Markov's inequality, we find that
\begin{align*}
    \tilde{\bv{g}}^\top \bv{x}  &\geq \frac{\rho}{2} - \lambda - \left(\frac{72}{\delta} - 64\right) m^{7/6} \lambda^{1/3} \rho^{-1} \geq \frac{\rho}{2} - \lambda^{1/3}\left(\frac{72 m^{7/6}}{\delta \rho}\right)
\end{align*}
with probability at least $1 - \delta/3$. Since $\lambda \leq \frac{\rho^6\delta^3}{2^{25}m^{7/2}}$, this implies that $\tilde{\bv{g}}^\top\bv{x} \geq \rho/4$, in which case $\|\tilde{\bv{g}}\| \geq \frac{\rho}{4\|\bv{x}\|} > \rho/6$. Accounting for a normalization factor of $1/\|\tilde{\bv{g}}\|_2$ and taking a union bound, we obtain an error bound of $3600 m^{7/6} \lambda^{1/3} \rho^{-2} \delta^{-1} \leq \delta$ with cumulative error probability at most $\delta$.
\end{proof}

Next, we compare the maximum value of $u(\bv{x},y)$ over $\bv{x} \in K_y$ to that over $\bv{x} \in K_y^{-2\eps}$.

\begin{lemma}
\label{lem:best-response-opt-sensitivity}
Let $\eps \geq 0$, and fix any $y \in \cY$ such that $K_y$ contains a $\ell_2$-ball of radius $\rho > 0$. We then have $\max_{\bv{x} \in K_y^{-2\eps}} u(\bv{x},y) \geq \max_{\bv{x} \in K_y} u(\bv{x},y) - \frac{2\sqrt{2 m}\eps}{\rho\Delta}$.
\end{lemma}
\begin{proof}
By \Cref{lem:best-response-margin}, we have $B(K_y,-2\eps/\Delta) \subseteq K_y^{-2\eps}$. Moreover, since the principal utilities lie in $[0,1]$, $u(\cdot,y)$ is $\sqrt{m}$-Lipschitz under the $\ell_2$-norm. Fixing $\lambda = 2\eps/\Delta$ and $K = K_y$, it suffices to show that, for each $\bv{x} \in K$, there exists $\bv{x}' \in B(K,-\lambda)$ with $\|\bv{x} - \bv{x}'\|_2 \leq \sqrt{2}\lambda/\rho$ (i.e., a Hausdorff distance bound between the setes $B(K,-\lambda)$ and $K$). If $\lambda > \rho$, then this is trivially true by the diameter of $K$. Otherwise, fix any $\bv{x}_0 \in B(K,-\rho)$. Since $K$ is convex and $B(\bv{x}_0,\rho) \subseteq K$, we have $\conv(\{\bv{x}\} \cup B(\bv{x}_0,\rho)) \subseteq K$. This convex hull contains the ball $B(\frac{\lambda}{\rho} \bv{x}_0 + (1 - \frac{\lambda}{\rho}) \bv{x}, \lambda)$, since $\lambda \leq \rho$ and
\begin{align*}
    B\left(\frac{\lambda}{\rho} \bv{x}_0 + \left(1 - \frac{\lambda}{\rho}\right) \bv{x}, \lambda\right) &= \left(1 - \frac{\lambda}{\rho}\right) \bv{x} + \frac{\lambda}{\rho} B\left(\bv{x}_0,\rho\right).
\end{align*}
Thus, we have $\bv{x}' = \frac{\lambda}{\rho} \bv{x}_0 + (1 - \frac{\lambda}{\rho}) \bv{x} \in B(K,-\lambda)$, with
\begin{align*}
    \|\bv{x} - \bv{x}'\|_2 \leq  \frac{\lambda}{\rho}\|\bv{x}_0 - \bv{x}\| \leq \frac{\lambda \sqrt{2}}{r},
\end{align*}
as desired.
\end{proof}

Finally, we are equipped to analyze \textsc{MembershipOpt}.

\begin{proof}[Proof of \Cref{lem:optimization-with-membership-queries}.]

We first address query complexity. Recall our parameter settings of $t_f = \bigl\lceil 3m \log \frac{16}{\delta \rho} \bigr\rceil$ and $\alpha = \min\{\frac{\delta}{t_f+1},\frac{\delta \rho}{16 \sqrt{2m} }\}$. By Lemma~\ref{lem:membership-to-separation}, we query \textsc{Oracle} at most
\begin{align*}
    (t_f+1) \cdot 10^3 m^{1.5} \log^2\left(\frac{100m}{\alpha\rho}\right) &\leq 6\cdot 10^3 m^{2.5} \log^2\left(\frac{100\cdot 16 \sqrt{2} \cdot 6 m^2 \log \frac{16}{\delta\rho}}{\delta\rho^2}\right)\log\frac{16}{\delta \rho}\\
    &\leq 6\cdot 10^3 m^{2.5} \log^2\left(\frac{100\cdot 16 \sqrt{2} \cdot 6 m^2 \log \frac{16}{\delta\rho}}{\delta\rho^2}\right)\log\frac{16}{\delta \rho}\\
    &\leq 10^5 m^{2.5} \log^3 \left(\frac{120m}{\delta \rho}\right)
\end{align*}
times, as claimed.

From now on, we suppose that $\eps \leq \left(\frac{\delta\rho}{140 m}\right)^{13}\Delta \leq \frac{\alpha^6\rho^6\Delta}{2^{39}m^4}$ and that $K_y$ contains a ball of radius $\rho$.
Condition on the Lemma~\ref{lem:membership-to-separation} guarantee for \textsc{SimulatedSep} holding for all of its $t_f + 1$ calls. Since $\eps$ obeys the lemma's bound, this event has probability at least $1 - \delta$ by a union bound. Under this event, \textsc{SimulatedSep} only returns ``$\perp$'' if a query $\bv{x}$ lies in $K = K_y^{-2\eps}$, and only returns a normal vector $\bv{w} \in \mathbb{S}^{m-1}$ if $(\bv{x} - \bv{z})^\top \bv{w} \leq \alpha$ for all $\bv{z} \in K$. 

As in the proof of \Cref{lem:membership-to-separation}, we note that $K$ must contain a ball of radius $\rho/2$ by \Cref{lem:best-response-margin}. Write $V_m = \pi^{m/2} \Gamma(m/2 +1)^{-1}$ for the volume of the unit ball in $\R^m$ and set $\tau = \frac{\delta}{4\sqrt{2m}}$. Now, we have $\vol_m(S_0) = 2^{m/2} V_m$, and, by Grünbaum's inequality (Lemma~\ref{lem:grunbaums-inequality}), $\vol_m(S_{t+1}) \leq (1 - 1/e) \vol_m(S_t)$. Fixing $\bv{x}^\star_\eps \in \argmax_{x \in K} u(x,y)$, we define $C = [(1-\tau) \bv{x}^\star_\eps + \tau K] \cap B(K,-\alpha)$ and bound its volume from below by
\begin{equation*}
    \vol_m(C) \geq \vol_m\bigl(B(\tau K, -\alpha)\bigr) \geq (\tau \rho/2 - \alpha)^m V_m \geq (\tau \rho / 4)^m V_m,
\end{equation*}
using that $\alpha \leq \tau \rho/4$. Thus, after $t_f \geq m \log_2 \bigl(\frac{\tau \rho}{4 \sqrt{2}}\bigr) / \log_2 \bigl(1-1/e\bigr)$ rounds, we cannot have $C \subseteq S_t$, and so there exists $r \in \{0, \dots, t_f\}$ for which we have some $\bv{x} \in C \cap S_r \setminus S_{r+1}$. This $x$ may not be removed at Step~\ref{step:center-of-gravity-cut}, because $\bv{x} \in C \subseteq B(K,-\alpha)$. Thus, $\bv{x}$ must be removed at Step~\ref{step:center-of-gravity-cut2} with $u(\bv{x}_r,y) > u(\bv{x},y)$, and $\bv{x}_r$ must be added to the set of candidate maximizers $A$. 

In particular $A \neq \emptyset$, and so we do not return ``$\perp$.'' Instead, for the returned strategy $\hat{\bv{x}} = \argmax_{\bv{x} \in A}u(\bv{x},y)$, we have 
\begin{align*}
    u(\hat{\bv{x}},y) \geq u(\bv{x}_r,y) &> u(\bv{x},y)\\
             &\geq  u(\bv{x}^\star_\eps,y) - \|\bv{x}_r - \bv{x}^\star_\eps\|_2\sqrt{m}\\
             &\geq u(\bv{x}^\star_\eps,y) - \diam(\tau K)\sqrt{m}\\
             &\geq u(\bv{x}^\star_\eps,y) - 2\sqrt{2m}\tau\\
             &= u(\bv{x}^\star_\eps,y) - \delta/2\\
             &\geq \max_{\bv{x} \in K_y} u(\bv{x},y) - \delta,
\end{align*}
as desired. The second inequality uses that $u(\cdot,y)$ is $\sqrt{m}$-Lipschitz under the $\ell_2$-norm, and the final inequality uses \Cref{lem:best-response-opt-sensitivity}. 
\end{proof}

%% file: A7-strategic-classification.tex
\section{Supplementary Material for Strategic Classification (Section~\texorpdfstring{\ref{ssec:strategic-classification}}{5.3})}
\label{app:classification}

First, we derive two lemmas from the regularity assumptions.

\begin{lemma}
\label{lem:strategic-classification-agent-regularity}
At any strategic round $t$, the agent's payoff is bounded from above by $R^2(1 + 1/\alpha)$ and the best response $\brr_t(\bm{\theta}_t)$ is unique with payoff at least $-R^2$.
\end{lemma}

\begin{proof}
First, by the $\alpha$-strong convexity assumption on $f_t$, $v_{a_t}$ is $\alpha$-strongly concave in $\bm{\theta}$, and so the best response $\brr_t(\bm{\theta}_t)$ is unique.
In the proof of Theorem 2 on page 8 of \cite{dong2018}, the authors show that $v_{a_t}(\bm{\theta}_t,\brr_t(\bm{\theta}_t)) \leq \bm{\theta}_t^\top \brr_t(\bm{\theta}_t) = \bv{x}_t^\top \bm{\theta}_t  + 2 f_t^*(\bm{\theta}_t)$, where $f_t^*$ is the convex conjugate of $f_t$. In the proof of Claim 2 on page 20, they further show that $f_t^*(\bm{\theta}) = \sup_{\bv{v} \in \mathbb{S}^{d-1}} \frac{ (\bm{\theta}^\top \bv{v})^2}{4 f_t(\bv{v})}$. The numerator of this objective is bounded from above by $R^2$, while the denominator is bounded from below by $2 \alpha$, since $\|\bv{v}\|_2 = 1$ and $f_t$ is $\alpha$-strongly convex with minimum of 0 the origin (due to homogeneity). Thus, $f_t^*(\bm{\theta}_t) \leq \frac{R^2}{2\alpha}$, and so $v_{a_t}(\bm{\theta}_t,\hat{\bv{x}}_t) \leq v_{a_t}(\bm{\theta}_t,\brr_t(\bm{\theta}_t)) \leq R^2(1 + \frac{1}{2\alpha})$.
Finally, by playing $\hat{\bv{x}} = \bv{x}$, the agent obtains payoff $\bm{\theta}_t^\top \bv{x}_t \geq -R^2$.
\end{proof}

\begin{lemma}
\label{lem:strategic-classification-principal-regularity}
Each map $\bm{\theta} \mapsto \ell(\bm{\theta},\brr_t(\bm{\theta}),y_t)$ is convex, $(R+2R/\alpha)$-Lipschitz, and bounded in absolute value by $1+R^2+R^2/\alpha$. Moreover, each map $\hat{\bv{x}} \mapsto \ell(\bm{\theta},\hat{\bv{x}},y_t)$ is $R$-Lipschitz.
\end{lemma}

\begin{proof}

Convexity of $\ell_t(\bm{\theta}) = \ell(\bm{\theta},\brr_t(\bm{\theta}),y_t)$ is implied by Theorem 2 of \cite{dong2018}. For Lipschitzness of $\ell_t$, we turn to the proof of Theorem 7 on page 17 of \cite{dong2018}. The discussion there implies that $\ell_t(\bm{\theta})$ is Lipschitz with constant $\|\bv{x}_t\|_2 \leq R$ plus twice the Lipschitz constant of  $f_t^*$. To compute this, we bound
\begin{align*}
    \left|\sup_{\bv{v} \in \mathbb{S}^{d-1}} \frac{(\bm{\theta}^\top \bv{v})^2}{4 f_t(\bv{v})} - \sup_{\bv{v} \in \mathbb{S}^{d-1}} \frac{(\bm{\theta}'^\top \bv{v})^2}{4 f_t(\bv{v})} \right| &\leq \sup_{\bv{v} \in \mathbb{S}^{d-1}}\left| \frac{(\bm{\theta}^\top\bv{v})^2}{4 f_t(\bv{v})} - \frac{(\bm{\theta}'^\top \bv{v})^2}{4 f_t(\bv{v})} \right| \leq \frac{2 R\|\bm{\theta}-\bm{\theta}'\|_2}{2\alpha},
\end{align*}
using the same lower bound on $f_t(\bv{v})$ as in the proof of \Cref{lem:strategic-classification-agent-regularity}. Combining gives a Lipschitz constant of $R + 2R/\alpha$ for $\ell_t$. Finally, discussion on page 18 of \cite{dong2018} implies that $|\ell_t(\bm{\theta})| \leq 1+ |\bm{\theta}^\top \bv{x}_t \rangle| + 2f_t^*(\bm{\theta})$, which we bound by $1 + R^2(1 + \frac{1}{2\alpha})$ as in the proof of \Cref{lem:strategic-classification-agent-regularity}. Finally, the same discussion on page 17 implies that the map $\hat{\bv{x}} \mapsto \ell(\bm{\theta},\hat{\bv{x}},y_t)$ has Lipschitz constant bounded by that of the map $\hat{\bv{x}} \mapsto \bm{\theta}^\top\hat{\bv{x}}$, which is $\|\bm{\theta}\|_2 \leq R$.
\end{proof}

\paragraph{Bandit convex optimization.}
By \Cref{lem:strategic-classification-principal-regularity}, each loss function $\ell_t(\bm{\theta}) = \ell(\bm{\theta},\brr_t(\bm{\theta}),y_t)$ in the myopic setting is convex, Lipschitz, and bounded. When $y_t = 1$, feedback $\hat{\bv{x}}_t$ is sufficient to determine the gradient $\nabla \ell_t(\bm{\theta}_t)$, suggesting regret minimization via online convex optimization (OCO). Although this is not the case when $y_t = -1$, since the agent's manipulation costs encoded by $\mathsf{d}_t$ are hidden, the regime of OCO with one-point function evaluations, or \emph{bandit convex optimization}, is well-studied. \cite{dong2018} employ the classic ``gradient descent without a gradient'' procedure \GDwoG{} of \cite{flaxman2005} (\Cref{alg:gradient-descent-without-a-gradient}) to obtain regret $O(\sqrt{d}T^{3/4})$ against myopic agents, computing unbiased gradient estimates from stochastic function evaluations.\medskip

\begin{algorithm}[H]
    \caption{Online \textbf{G}radient \textbf{D}escent \textbf{w}ith\textbf{o}ut a \textbf{g}radient \citep{flaxman2005}}\label{alg:gradient-descent-without-a-gradient}
    \DontPrintSemicolon
    \SetAlgoNoLine
    \SetKwInOut{Input}{input}
    \Input{domain $S \subset \R^d$ with $\mathbb{B} \subseteq S \subseteqq R \mathbb{B}$, $L$-Lipshitz convex fn.s $c_1, \dots, c_T : S \to [-C,C]$}
    $\delta \gets \sqrt{\frac{RdC}{3(L+C)}}T^{-1/4}$,\, $\eta \gets \frac{R}{C\sqrt{T}}$, \, $v_1 \gets (0,\dots,0) \in \R^d$\;
    \For{round $t = 1,\dots,T$}{
        Sample unit vector $\bv{s}_t \in \mathbb{S}^{d-1}$ uniformly at random\;
        $\bv{u}_t \gets \bv{v}_t + \delta \bv{s}_t$\;
        $\bv{v}_{t+1} \gets \Pi_{(1-\delta)S}(\bv{v}_t - \eta c_t(\bv{u}_t) \bv{s}_t)$ \tcp*{$\Pi_K$ is the Euclidean projection onto $K$}
    }
\end{algorithm}

\begin{lemma}[\cite{flaxman2005}, Theorem 2]
\label{lem:gradient-descent-without-a-gradient}
If functions $c_1,\dots,c_T:S \to [-C,C]$ are convex and $L$-Lipschitz, and $S \subseteq \R^d$ is convex with $\mathbb{B} \subseteq S \subseteq R\mathbb{B}$, then the queries $\bv{u}_1,\dots,\bv{u}_T$ of \GDwoG{} satisfy $\E\left[\sum_{t=1}^T c_t(\bv{u}_t)\right] - \min_{\bv{u} \in S} \sum_{t=1}^T c_t(\bv{u}) \leq 6 T^{3/4} \sqrt{RdC(L+C)} + 5C(Rd)^2$.
\end{lemma}

\paragraph{Our extension to non-myopic agents.} We now extend this
approach to robust learning, due to an intrinsic robustness of \GDwoG{}. First, we prove the relevant error bound.

\begin{lemma}
\label{lem:strategic-classification-robust-bd}
If the agent chooses $\hat{\bv{x}}_t \in \BestResp^\eps(\bm{\theta}_t)$, then  $|\ell(\bm{\theta}_t,\hat{\bv{x}}_t,y_t) - \ell(\bm{\theta}_t,\brr_t(\bm{\theta}_t),y_t)| \leq R\sqrt{2\eps/\alpha}$.
\end{lemma}
\begin{proof}
By the strong convexity assumption, we have $\|\hat{\bv{x}}_t - \brr_t(\bm{\theta}_t)\|_2 \leq \sqrt{2\eps/\alpha}$, and so the result follows from the Lipschitz bound $R$ from \Cref{lem:strategic-classification-principal-regularity}.
\end{proof}

We next provide an appropriate robust learning guarantee for \GDwoG{}.

\begin{lemma}
\label{lem:robust-gradient-descent-without-a-gradient}
Under the setting of \Cref{lem:gradient-descent-without-a-gradient}, \GDwoG{} achieves the same regret up to an additive factor of $\lambda R d T/\delta$ if each $c_t(\bv{u}_t)$ is substituted with an adversarial perturbation $\tilde{c}_t(\bv{u}_t) \in [c_t(\bv{u}_t) \pm \lambda]$.
\end{lemma}
\begin{proof}[Proof sketch]
The proof of \Cref{lem:gradient-descent-without-a-gradient} for the unperturbed setting goes through the smoothed functions $\bar{c}_t(\bv{u}) = \E[c_t(\bv{u} + \delta \bv{s}_t)]$. The key observations are that each $\bar{c}_t$ is convex and Lipschitz with $|\bar{c}_t(\bv{u}) - c_t(\bv{u})| \leq L\delta$ and, crucially, $\E[\frac{d}{\delta}c_t(\bv{u}_t)\bv{s}_t] = \nabla \bar{c}_t(\bv{v}_t)$. With adversarial perturbations, we have
\begin{equation*}
    \big\|\E\big[\tfrac{d}{\delta}\tilde{c}_t(\bv{u}_t)\bv{s}_t\big] - \nabla \bar{c}_t(\bv{v}_t)\big\|_2 = \big\|\E\big[\tfrac{d}{\delta}\left(\tilde{c}_t(\bv{u}_t) - c_t(\bv{u}_t)\right) \bv{s}_t\big]\big\|_2 \leq \frac{d \lambda}{\delta}.
\end{equation*}
At this point, the proof of Theorem 2 in \cite{flaxman2005} for the unperturbed setting nearly applies, so long as their Lemma 2 is adjusted to our setting where gradient estimates $\bv{g}_t$ have bias $b = \frac{d\lambda}{\delta}$ rather than $b=0$. Switching to their notation for the lemma, if we have $\E[\bv{g}_t|\bv{x}_t] =  \nabla c_t(\bv{x}_t) + \bm{\xi}_t$ with $\|\bm{\xi}_t\|_2 \leq b$, then the final chain of inequalities in their proof of Lemma 2 still holds, up to an added term of $\sum_{t=1}^n \bm{\xi}_t^\top( \bv{x}_t - \bv{x}_\star )\leq n b R$. Switching back to our notation and substituting our value for $b$, this gives the claimed regret overhead of $O(\frac{\lambda R Td}{\delta})$.
\end{proof}

Finally, we provide our combined algorithm and prove its non-myopic regret bound.\medskip

\begin{algorithm}[H]
    \caption{Cycled gradient descent without a gradient (\textsc{CGDwoG})}\label{alg:non-myopic-strategic-classification}
    \DontPrintSemicolon
    \SetAlgoNoLine
    $\eps \gets \alpha(R^4 dT^{2.5})^{-1}$, $D \gets \lceil T_\gamma \log(R^2(1+1/\alpha)T_\gamma/\eps) \rceil$\;
    Initialize copies $\cA_1,\dots,\cA_{D}$ of \Cref{alg:gradient-descent-without-a-gradient} w/ $S = \Theta$, $C = 1 + R^2 + R^2/\alpha$, and $L = R + 2R/\alpha$\;
    \For{round $t = 1,\dots,T$}{
        Write $t = D(k-1) + (r-1)$ for $k,r \in \mathbb{Z}_{>0}$\;
        Simulate query $\bv{u}_k$ and perturbed feedback $\tilde{c}_k(\bv{u}_k)$ for $\cA_r$ using $\bm{\theta}_t$ and $\ell(\bm{\theta}_t,\hat{\bv{x}}_t,y_t)$
    }
\end{algorithm}

\begin{proof}[Proof of Theorem~\ref{thm:strategic-classification}]
Since Stackelberg regret is subadditive over disjoint sequences of rounds, we obtain regret $D R_{\cA_1}^\eps(\lceil T/D \rceil)$ against $\eps$-approximate best-responding agents. Combining \Cref{lem:strategic-classification-principal-regularity,lem:strategic-classification-robust-bd,lem:robust-gradient-descent-without-a-gradient} and substituting our choices of constants,  we bound the regret of any 
single copy by \begin{align*}
    R_{\cA_1}^\eps(T) &\leq \E\left[\sum_{t=1}^T \ell(\bm{\theta}_t,\hat{\bv{x}}_t,y_t) - \min_{\bm{\theta} \in \Theta} \sum_{t=1}^T \ell(\bm{\theta},\brr_t(\bm{\theta}),y_t)\right]\\
    &\leq \E\left[\sum_{t=1}^T \ell(\bm{\theta}_t,\brr_t(\bv{x}_t),y_t) - \min_{\bm{\theta} \in \Theta} \sum_{t=1}^T \ell(\bm{\theta},\brr_t(\bm{\theta}),y_t)\right] + TR \sqrt{2\eps/\alpha}\\
    &\leq 6T^{3/4}\sqrt{RdC(L+C)} + 5C(Rd)^2 + T R\sqrt{\frac{2\eps}{\alpha}}(Rd/\delta+1)\\
    &= O\left(R^{5/2}\hat{\alpha}^{-1}\sqrt{d}T^{3/4} + R^4\hat{\alpha}^{-2}d^2\right),
\end{align*}
where $\hat{\alpha} = \min\{\alpha,1\}$. By \Cref{prop:delayed-feedback} and \Cref{lem:strategic-classification-agent-regularity}, our feedback delay induces $\eps$-approximate best responses, so we obtain a final regret bound of
\begin{equation*}
    O\left(R^{5/2}\hat{\alpha}^{-1}T_\gamma^{1/4}\sqrt{d}T^{3/4}\log^{1/4}(T R d / \alpha) + R^4\hat{\alpha}^{-2}d^2\right).\qedhere
\end{equation*}
\end{proof}

\begin{remark}
As noted in \citep{dong2018}, \GDwoG{} can be replaced by more modern methods with regret $\widetilde{O}(\poly(d)\sqrt{T})$ \citep{bubeck2021kernel}, at the cost of substantial complexity and worse scaling with $d$. While beyond the scope of this paper, robustifying such algorithms would imply improved non-myopic regret bounds.
\end{remark}